\documentclass[aps,pra, two column,10pt,superscriptaddress,nofootinbib]{revtex4-2}
\usepackage{pgfplots}
\pgfplotsset{compat=1.18}
\usepackage{amsmath}
\usetikzlibrary{plotmarks}
\usepackage{eufrak} 
\usepackage{physics}
\usepackage{amsmath,amssymb,amsthm,amstext,mathrsfs}
\usepackage[colorlinks=true,citecolor=blue,urlcolor=blue]{hyperref}
\usepackage{comment}
\usepackage{eufrak} %
\usepackage{enumerate}
\usepackage{times,txfonts}
\usepackage{braket}
\usepackage{svg}
\usepackage{color}
\usepackage{natbib}
\usepackage{amsmath,blkarray}
\usepackage{mathtools}
\usepackage{latexsym}
\usepackage{tabularx, booktabs}
\usepackage{graphics,epstopdf}
\usepackage{graphicx, lipsum}
\usepackage{amsfonts}
\usepackage{enumerate}
\usepackage{color,soul}

\newcommand{\be}{\begin{equation}}
\newcommand{\ee}{\end{equation}}
\newcommand{\ba}{\begin{eqnarray}}
\newcommand{\ea}{\end{eqnarray}}

\newtheorem{proposition}{Proposition}

\newtheorem{thm}{Theorem}
\newtheorem{corollary}{Corollary}

\definecolor{ss}{RGB}{250,80,220}
\begin{document}
\title{Certified quantum supremacy in entanglement-assisted prepare-measure random-access-code}

	\author{Rajdeep Paul}\email{rajdeeppaul22@gmail.com}\affiliation{Department of Physics, Indian Institute of Technology Hyderabad, Telangana-50228 India }
		\author{Prabuddha Roy}
        \email{prabuddhar@kias.re.kr}
         \affiliation{Quantum Universe Center, Korea Institute for Advanced Study, Seoul 02455, Republic of Korea }
         \author{A. K. Pan}
	\email{akp@phy.iith.ac.in}
	 \affiliation{Department of Physics, Indian Institute of Technology Hyderabad, Telangana-50228 India }
     
\begin{abstract} 

We develop a family of semi-device-independent (SDI) entanglement-assisted prepare-measure (PM) communication games involving two parties,  within the $n\rightarrow l$ random-access code (RAC) framework where the sender Alice holds a $n$-bit string and communicates  $l<n$ bits or qubits to the receiver Bob. In contrast to the standard quantum PMRAC, here the parties share a prior entanglement, and  Alice applies quantum operations on her sub-system to encode her inputs and sends to  Bob.  We first consider the $4\rightarrow l$ entanglement-assisted PMRAC with $l=1$ and $2$ and derive the optimal quantum success probabilities using an elegant analytical technique. We demonstrate quantum supremacy over both classical RACs and conventional quantum PMRACs. Moreover, we exhibit that the optimal quantum advantage allows one to certify Alice’s unitary operations. We then derive an upper bound on the quantum success probabilities for  $5\rightarrow  l$  entanglement-assisted PMRAC with $l=1,2$ and $3$.  Further, we extend the demonstration of quantum advantage for  $n\rightarrow n-2$ case where $n$ is arbitrary. 

%Thus, our results open a new avenue for quantifying how the number of communicated qubits in a prepare-and-measure game can be harnessed to enhance information-processing tasks.}

%{\color{blue}Quantum resources yield correlations with no classical equivalent. Recently, there has been a surge of interest in leveraging these correlations—specifically entanglement paired with quantum communication—as operational assets within semi-device-independent (SDI) prepare-and-measure (PM) frameworks. These frameworks serve as the foundational architecture for a diverse array of quantum communication protocols. Motivated by this, we investigate a newly introduced entanglement-assisted PM task in the setting of random-access codes (RACs), namely the $4$-bit PMRAC variant of \cite{Raj2025}. Building on \cite{Raj2025}, we analytically examine two instances of this $4$-bit PMRAC game, each involving a different number of communicated qubits, and establish quantum superiority over both classical RACs and the standard $4$-bit PMRAC scheme. The optimal quantum success probabilities are obtained in closed form via the sum-of-squares (SOS) technique. We additionally demonstrate that, in the single-qubit communication setting, this optimal quantum advantage permits self-testing of Alice’s unitary operations, whereas such certification fails once two qubits are communicated. Conversely, GHZ-type states achieve the maximal success probability, enabling a sharp discrimination between the GHZ and W entanglement classes.}

\end{abstract}
\maketitle
\section{Introduction}

A quintessential feature of quantum information theory is to seek novel tasks for which quantum communication offers an advantage over classical communication. A widely used framework for studying such advantages is the PM paradigm, where a sender (Alice) encodes information into a physical state, transmits it through a classical or quantum channel, and a receiver (Bob) performs a measurement to decode a specific property of the input without access to Alice's encoding choice. This paradigm underpins numerous information-processing applications, including foundational tests of quantum theory \cite{Pawlwski2009,Brassard2006}, dimension witnessing \cite{Brunner2013,Wehner2008,Gallego2010,De2017,Shyam2020}, random access coding \cite{Ambainis1999,Ambainis2008,Armin2015prl}, quantum random-number generation \cite{Li2011,Li2012}, quantum key distribution \cite{Paw2011,Woodhead2015}. 

A PM scenario is naturally studied in the  SDI framework, in that an upper bound on the Hilbert-space dimension of the communicated system is assumed, without relying on the inner workings of the preparation and measurement devices \cite{Tavakoli2018st,Farkas2019,Mironowicz2019,Tavakoli2020sc, Smania2020, Miklin2020, Gomez2016, Gomez2018, Pan2021, Gomez2016, supic2020input}. Lately, various SDI protocols in the PM scenario have been proposed to certify non-projective measurements \cite{Mironowicz2019,Smania2020}, quantum states and sharp measurements \cite{Tavakoli2018,Tavakoli2020sdic}, mutually unbiased bases \cite{Farkas2019}, randomness \cite{Li2011, Li2012, Lunghi2015}, unsharpness parameters \cite{Mohan2019,Miklin2020,Abhyoudai2023,Paul2024}, unitary operations \cite{Raj2025}, and dimension witnesses~\cite{Vicente2017,Tavakoli2020sdic}. 

 Recently, motivated from the dense-coding protocol \cite{Holevo1973}, Tavakoli et al.~\cite{TavakoliPRXQ2021} introduced a hybrid PM scenario involving prior entanglement between the sender and receiver, termed as the entanglement-assisted PM scenario. In this work, we consider a family of generalized $n\rightarrow l$ two-level RAC in entanglement-assisted PM scenario and demonstrate that $n$-bit PMRAC with pre‑shared entanglement involving $l$ qubit communication outperforms both classical RACs \cite{Ambainis2009} and standard  quantum RACs \cite{Galvao2001,Hayashi2006,Tavakoli2015,Chailloux2016,Pawlowski2017,Hameedi2017a,Tavakoli2018}.
 
 %In particular, we put forth various cases of $n=4,5$‑bit entanglement-assisted PMRAC depending on the number of qubits $l$ transmitted from Alice to Bob. In our $n \rightarrow l$ entanglement-assisted PMRAC, Alice and Bob initially share a bipartite entangled state $\rho\in \mathcal{L}(\mathbb{C}^{2^l} \otimes \mathbb{C}^{2})$. Alice receives an input string $x \in \{0,1\}^n$, sampled uniformly at random, while Bob is independently given a uniformly random input $y \in [n]$. Upon receiving $l$ qubits from Alice, Bob's task is to guess the $y$-th bit of Alice’s string $x$. Here, we aim to find suitable encoding and decoding strategies that maximize the winning probability.

In a standard two-level $n \rightarrow l$ PMRAC game, Alice receives an input string $x \in \{0,1\}^n$, sampled uniformly at random, while Bob is independently given a uniformly random input $y \in [n]$. Bob's task is to guess the $y^{th}$ bit of Alice’s string $x$. In classical RAC, Alice sends $l$ bits to Bob. In the standard two-level quantum PMRAC, Alice encodes her inputs $x$ into the $l$-qubit quantum states $\rho_x$ and sends them to Bob. In our entanglement-assisted PMRAC, Alice and Bob initially share a bipartite entangled state $\rho\in \mathcal{L}(\mathbb{C}^{2^l} \otimes \mathbb{C}^{2})$, and upon receiving the input $x$, Alice applies quantum operations $\Lambda_x$ on her part of the subsystem and sends her $l$ qubits to Bob. The rest of the protocol resembles the standard RAC scenario. Note that $2\rightarrow 1$ entanglement-assisted PMRAC is trivial, which perfectly mimics the dense-coding protocol. The $3\rightarrow 1$ entanglement-assisted PMRAC has been explicitly demonstrated in \cite {Raj2025}, where the self-testing of unitaries and the associated measurements were explored.

In this work, we first consider the entanglement-assisted $4 \rightarrow l$ PMRAC for $l=1$ and $2$. We analytically derive the optimal success probability in each case and demonstrate the quantum advantage over  $4$-bit classical and standard quantum PMRAC. We further show that the optimal success probability in $4 \rightarrow 1$ PMRAC certifies the maximally entangled two-qubit state shared by Alice and Bob, Alice's local unitary operations, and Bob’s measurements. For $4 \rightarrow 2$ entanglement-assisted PMRAC, we prove that a three-qubit GHZ state is necessary to attain the optimal quantum value and provide an explicit example for the construction of states and measurements achieving the optimal quantum success probability. We then consider $5 \rightarrow l$ PMRAC with $l=\{1,2,3\}$ and analytically derive upper bounds on the quantum success probabilities. Finally, we extend the framework to general $n\rightarrow l$ entanglement-assisted PMRAC with $l=n-2$ and demonstrate the quantum advantage over  classical RAC.

The paper is organized as follows. In Sec.~\ref {prem}, we briefly describe the classical, quantum RAC along with the $n$-bit entanglement-assisted PMRAC. In Sec .~\ref{4PMRAC}, we derive the maximum quantum success probabilities for the $4 \to 1$ and $4 \to 2$ entanglement-assisted PMRAC. Subsequently we demonstrate certification of unitary operations for the $4 \to 1$ case, while for $4 \to 2$, we establish the necessity of GHZ state for obtaining the optimal value. In Sec .~\ref{5PMRAC} and \ref{nPMRAC}, we introduced the $5 \rightarrow l, \forall l \in [3]$ and generalized $n\rightarrow n-2$ entanglement-assisted PMRAC respectively and derive the upper bound on the quantum success probabilities. Finally, in Sec.~\ref{conc}, we summarize our findings and discuss future directions.
%%%%%%%%%%%%%%%%%%%%%%%%%%%%%%%%%%%%%%%%%%%%%%%%%%%%%%%%%%%%%%%%%%%%%%%%%%%%
\section{Preliminaries}\label{prem} 
To set the stage for our main results, we first provide a brief description of the  classical RAC, the standard quantum PMRAC and the entangled-assisted PMRAC~\cite{Hayashi2006, Ambainis2009, Spekkens2009,  Chailloux2016, Hameedi2017b, Tavakoli2018, Mohan2019, Pauwels2022, Pauwels2022njp,Piveteau2022,Abhyoudai2023, Singh2025}.

%In the standard PM scenario, the sender (Alice) encodes an input variable $x$ into a message of reduced dimension and transmits it to the receiver (Bob) through a communication channel. The classical or quantum message is constrained by an underlying Hilbert-space dimension $d$. Bob then selects an input $y$ and performs a decoding procedure that yields an output $b$ according to the conditional distribution $p(b|x,y)$. More recently, PM scenarios have also been investigated considering the assumption that Alice and Bob pre-share entanglement~\cite{TavakoliPRXQ2021}.

The two-level RAC can be described as follows. Alice receives an input from a uniform random set of $n$-bit strings
$x\in\{0,1\}^{n}$. Upon receiving the input $x$
she uses the preparation procedure $P_x$ to prepare the physical state and sends it to
Bob. On the other hand, Bob randomly receives an input $y\in[n]$ and performs a dichotomic
measurement $B_y$ which produces output $b\in \{0,1\}$. The winning condition of the game is $b = x_y$, i.e., when Bob outputs the  $y^{\text{th}}$ bit of Alice's input. The success of the RAC task is determined by optimizing the winning probability $p(b=x_y|P_{x},B_{y})$, given by
	\begin{align}
		\label{GSP}
		\mathcal{P}_n=\dfrac{1}{2^{n}n}\sum_{y=1}^{n}\sum_{x\in \{0,1\}^{n}}p(b=x_{y}|P_{x},B_{y}).
	\end{align}

   In the following, we briefly recapitulate the classical RAC, the standard  quantum PMRAC, and the  entanglement-assisted PMRAC.

\subsection{The classical and the standard quantum PMRAC}

In a $n \to l$ two-level classical RAC \cite{Ambainis2008}, where given a $n$-bit string, Alice sends a message of length $l(<n)$. The average classical success probability $\mathcal{P}^{n\rightarrow l}_C$ for arbitrary $n$ and $l$ remains a challenging task. Nonetheless, it is simple to understand that for a given $n$, the higher the value of $l$, the higher the success probability can be obtained. For example, for $n=3$, if Alice transmits a single bit ($l=1$) to Bob, the classical success probability is $\mathcal{P}_{C}^{3\rightarrow 1}=0.75$. However, if Alice sends two bits ($l=2$), the success probability becomes $\mathcal{P}_{C}^{3\rightarrow 2}=0.833$.

 In the standard quantum PMRAC ~\cite{Galvao2001,Hayashi2006,Tavakoli2015,Chailloux2016,Tavakoli2018}, corresponding to random input string $x\in\{0,1\}^n$, Alice prepares quantum states $\rho_x\in\mathcal{L}(\mathbb{C}^{2^l})$. The state is then transmitted to Bob, who receives a random input $y\in [n]$ and performs a dichotomic measurement $\{\Pi_{y}^b\}$, where $\Pi_{y}^{b} = (\openone+(-1)^{b}B_{y})/{2}$ with $b\in\{0,1\}$. For the winning condition $b = x_y$, the average success probability $\mathcal{P}_{SQ}^{n\rightarrow l}$ can be derived. 

\subsection{The Entanglement-assisted PMRAC}

\begin{figure}[h]
                 \centering
                 \includegraphics[width=8.0cm, height=4.3cm]{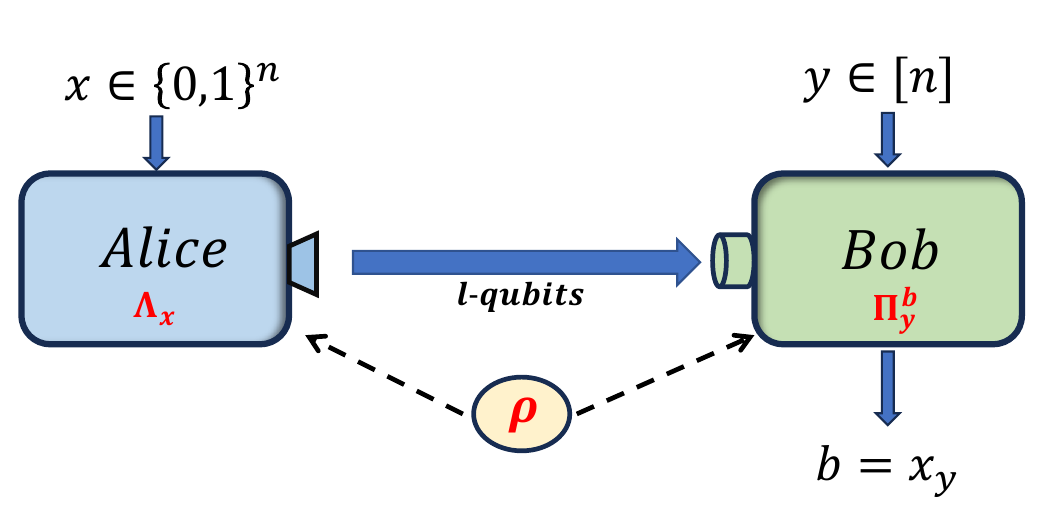}
                 \caption{\textbf{Entanglement-assisted quantum PMRAC:} Alice and Bob share a quantum state $\rho$. Alice performs quantum channel operation $\Lambda_x$ on her qubit given input $x\in\{0,1\}^n$ and transmits it to Bob, who performs projective measurement $\Pi_{y}^{b}$ on both qubits and produces a binary outcome.}\label{eapmrac}
                 \end{figure}
                 
In this work, we consider the entanglement-assisted PMRAC first introduced in \cite{TavakoliPRXQ2021}. In an $n \to l$ entanglement-assisted PMRAC, Alice and Bob share a bipartite entangled state $\rho\in \mathcal{L}(\mathbb{C}^{2^l} \otimes \mathbb{C}^{2})$. Given an input string $x \in \{0,1\}^n$, Alice applies a completely positive trace-preserving map $\Lambda_x$ to her subsystem, as depicted in Fig.~\ref{eapmrac}, and then transmits her $l$ qubits to Bob. Consequently, Bob receives one of the states $\rho_x = (\Lambda_x\otimes\openone)(\rho)$ and performs a joint  projective measurement with two possible outcomes $\{\Pi_{y}^b\}$ on this system. Given the winning condition $b=x_y$, the average quantum success probability of the $n \to l$ entanglement-assisted PMRAC is thus given as

\begin{eqnarray}\label{VQSP}
  \mathcal{P}^{n\rightarrow l}_{Q}=\dfrac{1}{2^{n}n}\sum_{y=1}^{n}\sum_{x\in \{0,1\}^{n}}\Tr\hspace{2pt}[\rho_{x}\Pi^b_y]
\end{eqnarray}
where $l$ denotes the number of qubits transmitted from Alice to Bob. It should be noted that, for a meaningful comparison, $n> l$ must be satisfied. 

In the $3\rightarrow 1$ entanglement-assisted PMRAC \cite{Raj2025}, the optimal quantum success probability is derived as 
\begin{eqnarray}\label{ps31}
  \mathcal{P}^{3\rightarrow 1}_{Q}=\dfrac{1}{2}+\frac{1}{\sqrt{6}}\approx 0.908
\end{eqnarray}
Note that the standard quantum PMRAC framework, the maximal success probability is $\mathcal{P}^{3\rightarrow 1}_{SQ}=0.788$, which implies $\mathcal{P}^{3\rightarrow 1}_{Q}>\mathcal{P}^{3\rightarrow 1}_{SQ}>\mathcal{P}^{3\rightarrow 1}_C$.
%The detailed derivation of the success probability along with optimization is provided in \cite{Raj2025}. 

In the preceding section, we explicitly construct the $4 \to l$ entanglement-assisted PMRAC and analytically derive their optimal success probabilities, thereby exhibiting a quantum advantage over classical RAC. \\

%It is important to note that in all the related works in the prepare-and-measure RAC,the state preparation involves arbitrary unitaries to rotate the shared state $\rho$ between Alice and Bob so that the RAC task attains the maximum success probability. Whereas, in order to establish our result, we have considered one shared state between Alice and Bob and rotated it using $2^n-1$ number of unitaries to prepare the corresponding states. This, in turn, helps us optimize the RAC's success probability with fewer resources. Crucially, we employ an analytical approach to demonstrate the certification protocol considering the fixed dimension of the system.

\section{\texorpdfstring{$4\rightarrow l$ entanglement-assisted PMRAC}{4-bit entanglement-assisted PMRAC}}\label{4PMRAC}

In  $4 \rightarrow l$ entanglement-assisted PMRAC, there are two distinct RACs corresponding to $l=1$ and $2$. We explicitly demonstrate quantum advantage in each case and discuss the certification of unitary operation.

In classical RAC, Alice receives $4$-bit strings, drawn uniformly at random from $x \in \{0,1\}^4$. Bob receives inputs $y \in [4]$ uniformly at random and performs a dichotomic measurement with the outcome $b \in \{0,1\}$. Given the task's winning condition as $b=x_y$, if Alice communicates one or two bits to Bob then the optimal classical success probability is $\mathcal{P}^{4\rightarrow 1}_C=0.69$ and $\mathcal{P}^{4\rightarrow 2}_C=0.75$ respectively. We derive the quantum success probabilities when Alice transmits one or two qubits. 

\subsection[4 to 1 case]{\(4 \rightarrow 1\) entanglement-assisted PMRAC}

In quantum theory, Alice and Bob share a two-qubit quantum state $\rho \in \mathcal{L}(\mathbb{C}^2\otimes \mathbb{C}^2)$. Alice encodes her input $x\in\{0,1\}^4$ by applying local unitary operation $U_{x}$ corresponding to $x$ so that $\rho_{x}=(U^\dagger_{x}\otimes\openone) \ \rho \ (U_{x}\otimes\openone)$. She then communicates her  qubit to Bob. Bob then has two qubits in his possession, and he performs binary outcome measurements $\{\Pi^b_{4,y}\}$ where $\Pi_{4,y}^{b}=(\openone_4+B_{4,y})/2, \forall y\in[4]$ (with $\sum_{b}\Pi^b_{4,y}=\openone_4$ and $b\in\{0,1\}$) for decoding $y^{\text{th}}$ bit of $x$. The quantum success probability from Eq.~(\ref{VQSP}) can be written as
\begin{widetext}
    \begin{eqnarray}\label{cal41}
	\mathcal{P}^{4\rightarrow 1}_{Q}=\dfrac{1}{2}&
    +\dfrac{1}{128}\ Tr[(\rho_{0000}-\rho_{1111}+\rho_{0011}-\rho_{1100})(B_{4,1}+B_{4,2})+(\rho_{0000}-\rho_{1111}-\rho_{0011}+\rho_{1100})(B_{4,3}+B_{4,4})\nonumber \\
	&+(\rho_{0101}-\rho_{1010}+\rho_{0110}-\rho_{1001})(B_{4,1}-B_{4,2})+(\rho_{0101}-\rho_{1010}-\rho_{0110}+\rho_{1001})(B_{4,3}-B_{4,4})\nonumber \\
	&+(\rho_{0001}-\rho_{1110}+\rho_{0010}-\rho_{1101})(B_{4,1}+B_{4,2})+(\rho_{0001}-\rho_{1110}-\rho_{0010}+\rho_{1101})(B_{4,3}-B_{4,4})\nonumber \\
	&+(\rho_{0100}-\rho_{1011}+\rho_{0111}-\rho_{1000})(B_{4,1}-B_{4,2})+(\rho_{0100}-\rho_{1011}-\rho_{0111}+\rho_{1000})(B_{4,3}
 +B_{4,4})]
\end{eqnarray}
\end{widetext}

Now, let us examine the first term in Eq. (\ref{cal41}). To obtain the maximum value of $\mathcal{P}^{4\rightarrow 1}_{Q}$, the states $\rho_{0000}, \rho_{0011}, \rho_{1100}$ and $\rho_{1111}$ must be eigen-states of $B_{4,1}$ and $B_{4,2}$. This demands that the states are mutually orthogonal pure states and satisfy $\rho_{0000}+\rho_{1100}+\rho_{0011}+\rho_{1111}=\openone_4$. A similar argument holds for each of such terms in Eq. (\ref{cal41}). This leads us to define two-qubit operators of the form
\begin{widetext}
    \begin{eqnarray}
	\label{obserg42}
	\nonumber
	M^1_{4,1}&=&\rho_{0000}-\rho_{1111}+\rho_{0011}-\rho_{1100}; \  
	M^1_{4,3}=
	\rho_{0001}-\rho_{1110}+\rho_{0010}-\rho_{1101},
	M^2_{4,1}=\rho_{0000}-\rho_{1111}-\rho_{0011}+\rho_{1100};\nonumber\\
    M^2_{4,3}&=&\rho_{0001}-\rho_{1110}-\rho_{0010}+\rho_{1101};
	M^1_{4,2}=
	\rho_{0101}-\rho_{1010}+\rho_{0110}-\rho_{1001};\
 M^1_{4,4}=
	\rho_{0100}-\rho_{1011}+\rho_{0111}-\rho_{1000},\nonumber \\
	M^2_{4,2}&=&
	\rho_{0101}-\rho_{1010}-\rho_{0110}+\rho_{1001};\
	M^2_{4,4}=\rho_{0100}-\rho_{1011}-\rho_{0111}+\rho_{1000}
\end{eqnarray}
\end{widetext}
  
Thus, the success probability in Eq. (\ref{cal41}) takes the form \begin{eqnarray}
	\label{QSP42}
	\mathcal{P}^{4\rightarrow 1}_{Q}&=&\dfrac{1}{2}+\dfrac{\mathscr{I}^{4\rightarrow 1}}{128}
\end{eqnarray}
where $\mathscr{I}^{4\rightarrow 1}$ is the correlation function 
\begin{eqnarray}\label{p142}
\mathscr{I}^{4\rightarrow 1}&=&\Tr\Big[(M^1_{4,1}+M^1_{4,2}+M^1_{4,3}+ M^1_{4,4})B_{4,1}\nonumber\\ 
	&&+(M^1_{4,1}-M^1_{4,2}+M^1_{4,3}- M^1_{4})B_{4,2}\nonumber\\ 
	&&+ (M^2_{4,1}+M^2_{4,2}+ M^2_{4,3}+M^2_{4,4})B_{4,3}\nonumber\\ 
	&&+(M^2_{4,1}-M^2_{4,2}- M^2_{4,3}+M^2_{4,4})B_{4,4}
 \Big]
\end{eqnarray}

Following we introduce an elegant analytical technique to derive the maximum value of the above correlation function.%\begin{thm}\label{thm1} If Alice communicates a qubit $\rho_x$ and Bob performs a binary outcome measurement $\{\Pi^b_{4,y}\}$, the maximum success probability $(\mathcal{P}^{4\rightarrow 1}_Q)^{opt}=\frac{1}{2}+\frac{1}{2\sqrt{2}}$, is attained. 
%\end{thm}
%\begin{proof}
In order to find the maximum value of $\mathscr{I}^{4\rightarrow 1}$, we write the correlation function in Eq.~(\ref{p142}) as
\begin{eqnarray}
\mathscr{I}^{4\rightarrow 1}&=&\Tr[\sum\limits_{y=1}^{4}\omega_y\mathscr{M}_{4,y} B_{4,y} ]
\end{eqnarray}
where $\mathscr{M}_{4,y}$ and $\omega_y$'s are defined as following
\begin{widetext}
    {\small \begin{eqnarray}\label{Mop}
     &&\mathscr{M}_{4,1} = \frac{M^1_{4,1}+M^1_{4,2}+M^1_{4,3}+ M^1_{4,4}}{\omega_1}, \mathscr{M}_{4,2} = \frac{M^1_{4,1}-M^1_{4,2}+M^1_{4,3}- M^1_{4,4}}{\omega_2}, \mathscr{M}_{4,3} = \frac{M^2_{4,1}+M^2_{4,2}+ M^2_{4,3}+M^2_{4,4}}{\omega_3}, \mathscr{M}_{4,4} = \frac{ M^2_{4,1}-M^2_{4,2}- M^2_{4,3}+M^2_{4,4}}{\omega_4}\quad\\ 
     &&\omega_1=||M^1_{4,1}+M^1_{4,2}+M^1_{4,3}+M^1_{4,4}||,\omega_2=||M^1_{4,1}-M^1_{4,2}+M^1_{4,3}-M^1_{4,4}||,\omega_3=||M^2_{4,1}+M^2_{4,2}+M^2_{4,3}+M^2_{4,4}||,\omega_4=||M^2_{4,1}-M^2_{4,2}-M^2_{4,3}+M^2_{4,4}||
     \label{omegai42aap}
\end{eqnarray}}
\end{widetext}

with $||\mathscr{M}_{4,y}||=1$, where $||\cdot||$ denotes the scaled Frobenius norm defined by $||\mathcal{O}||=\frac{1}{2}\sqrt{\Tr[\mathcal{O}^{\dagger}\mathcal{O}] }$.

Notably, since both $\mathscr{M}_{4,y}$ and $B_{4,y}$ are dichotomic, normalized observables, the maximum value of $\mathscr{M}_{4,y} B_{4,y}$ reaches only when $\mathscr{M}_{4,y}=B_{4,y}$, which in turn leads to
\begin{equation}\label{optbn42}
\mathscr{I}^{4\rightarrow 1}_{opt}= \Tr[\max\left(\sum\limits_{y=1}^{4}\omega_{y} \right)\openone_4]=4\max\left(\sum\limits_{y=1}^{4}\omega_{y} \right)
\end{equation}
where $\mathscr{I}_{opt}^{4\rightarrow 1}$ denotes the optimal value of the correlation function for the qubit system.

Now, applying the convex inequality $\sum_{y=1}^{n}\omega_{y}\leq \sqrt{n \sum_{y=1}^{n} (\omega_{y})^{2}}$, where equality is achieved only if $\omega_y=\omega_{y^{\prime}},\forall y\neq y^{\prime}\in [n]$, we obtain from Eq.~(\ref{optbn42}) the expressions for $\omega_y$ as
\begin{eqnarray}\label{A141}
\mathscr{I}^{4\rightarrow 1}&\leq&4\qty( \max \sqrt{4 \ \qty(\omega_1^2+\omega_2^2+\omega_3^2+\omega_4^2)} )\nonumber \\
&\leq&4\max\Big[8 \ \Big(8+Tr[M^1_{4,1} M^1_{4,3} +M^1_{4,2}M^1_{4,4}+M^2_{4,1}M^2_{4,4}\nonumber\\
&&+M^2_{4,2} M^2_{4,3}]/2\Big)\Big]^\frac{1}{2}
\end{eqnarray}
We explicitly derived the maximum value of Eq.~(\ref{A141}) in Appendix.~\ref{app:optimal} and get
\begin{eqnarray}\label{AMr}
    \Tr[M^1_{4,1} M^1_{4,3} +M^1_{4,2}M^1_{4,4}+M^2_{4,1}M^2_{4,4}+M^2_{4,2} M^2_{4,3}]=16\ 
\end{eqnarray}
which fixes the value of $\omega_y=2\sqrt{2},\forall y\in[4]$. Therefore, from Eq.~(\ref{A141}), we get $\mathscr{I}^{4\rightarrow 1}_{opt}=32\sqrt{2}$. 

Hence, the optimal quantum success probability of Eq.~(\ref{QSP42}) can be derived as
\begin{equation}\label{p41max}
(\mathcal{P}^{4\rightarrow 1}_Q)^{opt}=\dfrac{1}{2}+\frac{1}{2\sqrt{2}}\approx 0.853
\end{equation}
A detailed derivation of Eq.~(\ref{AMr}) is provided in the Appendix~\ref{app:optimal}.
%\end{proof}
Hence, $(\mathcal{P}^{4\rightarrow 1}_Q)^{opt}>\mathcal{P}_C^{4\rightarrow 2}>\mathcal{P}_C^{4\rightarrow 1}$, meaning that the optimal quantum success probability exceeds that of the classical RAC even when two classical bits are communicated. Moreover, the best-known success probability for the standard $4\rightarrow 1$ quantum PMRAC is $\mathcal{P}_{SQ}^{4\rightarrow 1} = 0.7415$~\cite{Mukherjee2021}, which is less than $(\mathcal{P}^{4\rightarrow 1}_Q)^{opt}$.

Note that the optimal quantum success probability for the $4 \to 1$ entanglement-assisted PMRAC fixes certain conditions on state and measurements which is encapsulated in the following corollary. 
\begin{corollary}
\label{coro1} The optimal quantum success probability $(\mathcal{P}^{4\rightarrow 1}_Q)^{opt}$ certifies the following statements.

\begin{enumerate}[(i)]
 \item Given two-qubit density operators $\{\rho_{x_1x_2x_3x_4}\}_{x_i\in\{0,1\}}$, the set of mutually orthogonal states satisfy
\begin{eqnarray}\label{cor1i}
    \sum_{\substack{x_3,x_4\in\{0,1\}\\ x_3\oplus x_4\in\{0,1\}}}\Big(\rho_{s\,x_3x_4}+\rho_{\bar s\,x_3x_4}\Big)=\openone_4;\ \forall s
\end{eqnarray}
 
    where $s\in\{0,1\}^2$, and $\bar s:=(s\oplus 11)$ denotes its bitwise complement.

    \item The states follow $$\ \Tr[\rho_{x}  \ \rho_{x^{\prime}}]_{\forall x,x^{\prime}\in\{0,1\}^4}=   
   \begin{cases}
    1 ;\ \ \ \ \quad \text{when}\ \  x^{\prime}=x\\
    \mathcal{O}\otimes \mathcal{O} \quad \text{when}\, x^{\prime} \neq x
    \end{cases}$$
    \text{where $\mathcal{O}=\begin{pmatrix} 1 & 1/2 & 1/2 & 0 \\ 1/2 & 1 & 0 & 1/2 \\ 1/2 & 0 & 1 & 1/2 \\ 0 & 1/2 & 1/2 & 1 \end{pmatrix}$ \  (Table~\ref{tab1}, Appendix~\ref{app:coro1})}\\ 
    \item  Bob's observables $B_{4,y}=\mathscr{M}_{4,y} \quad \forall y\in[4]$, where $\mathscr{M}_{4,y}$ are two-qubit operators defined in Eq.~(\ref{Mop}) with $\omega_y=2\sqrt{2}$, $\forall y$. 
\end{enumerate}
\end{corollary}
\begin{proof}
    Note that statement (i) has already been proved immediately after Eq. (\ref{cal41}). A detailed proof of statements (ii) and (iii) is placed in Appendix~\ref {app:coro1}.

\end{proof}
In particular, Corollary \ref{coro1} establishes the relationships among the states $\rho_{x}$ and shows that Bob’s observables can be completely specified with respect to these states.

 We now use these findings to demonstrate that the initially shared state must be maximally entangled.

     \begin{proposition}\label{prop1}
For optimal quantum success probability $(\mathcal{P}^{4\rightarrow 1}_Q)^{opt}$ the initial shared state between Alice and Bob is maximally entangled of the form
    \begin{eqnarray}
\label{entan}
    \rho=\frac{1}{4}\left(\openone_4-Q_1\otimes Q_2-P_1\otimes P_2-R_1\otimes R_2\right)
\end{eqnarray}
where $\{P_1, Q_1, R_1\}$ and $\{P_2, Q_2,R_2\}$ are the sets of mutually anti-commuting qubit operators. 
\end{proposition}
\begin{proof}

       The optimal quantum value requires each set of $\{\rho_{x_1x_2x_3x_4}\}$  to be mutually orthogonal pure states and form a complete basis, as explicitly proved in statement Eq.~(\ref{cor1i}). Since, only Alice can apply the quantum operation on the shared state $\rho$, to produce four mutually orthogonal two-qubit states, the state $\rho$ has to be maximally entangled. For example, starting with $\ket{\psi}=a\ket{00} + b\ket{11}$ and applying local unitary operations on one side of the state, a basis of four orthogonal states can be produced only if $|a| = |b|$ holds. This proves the maximally entangled two-qubit state can always be written in the form in Eq. (\ref{entan}). The detailed proof is given in the appendix \ref{app:maxent}.
       %For example, starting with $\ket{\psi}=a\ket{00} + b\ket{11}$ and applying local unitary operations on one side of the state, a basis of four orthogonal states can be produced only if $|a| = |b|$ holds. This proves that a maximally entangled two-qubit state can always be written in the form in Eq. (\ref{entan}). 
     
     \end{proof}

     In the preceding subsection, we use the Corollary \ref{coro1} and proposition \ref{prop1} to establish the fact that the maximally entangled state of the form Eq. (\ref{entan}) has the power to self-test Alice's unitary operations when the optimal success probability is achieved. This is captured through the following Theorem \ref{thm2}.
 
\begin{thm}\label{thm2}
 The optimal quantum success probability $(\mathcal{P}^{4\rightarrow 1}_Q)^{opt}$ certifies
 \begin{enumerate}
     \item Alice's unitary operators are $U_{0000}=\openone_2$,  $U_{1100}=i\ Q_1$, $U_{0011}=-i \ P_1$ and $U_{1111}=-i \ R_1$ are mutually anti-commuting.
     \item There exists grand unitaries $U^1_{G}= \frac{-\openone_2+U_{0011}}{\sqrt{2}}$, $U^2_{G}= \frac{\openone_2+U_{1100}}{\sqrt{2}}$ and $U^3_G= \frac{\openone_2+U_{1100}-U_{0011}+U_{1111}}{2}$ which gives 
    {\small \begin{eqnarray}\label{ur41}
    &&(U^1_G)^\dagger\{\rho_{0000},\rho_{1111},\rho_{0011},\rho_{1100}\}U^1_G \rightarrow \{\rho_{0001},\rho_{1110},\rho_{0010},\rho_{1101}\}\quad \nonumber \\
    &&(U^2_G)^\dagger\{\rho_{0000},\rho_{1111},\rho_{0011},\rho_{1100}\}U^2_G \rightarrow \{\rho_{0100},\rho_{0111},\rho_{1011},\rho_{1000}\}  \\ \label{unitary}
    &&(U^3_G)^\dagger\{\rho_{0000},\rho_{1111},\rho_{0011},\rho_{1100}\}U^3_G \rightarrow \{\rho_{0101},\rho_{0110}, \rho_{1010}, \rho_{1001}\}\nonumber
\end{eqnarray}}
 \end{enumerate} 
\end{thm}
\begin{proof}
  The proof is carried out in two steps. The first step relies on two key elements: the mutually orthogonality condition satisfied by the states as stated in Eq.~(\ref{cor1i}) along with the general form of the maximally entangled state given in Eq. (\ref{entan}). Together, these two results allow for the construction of the states. In the second step, we use the relation between qubit operators of Alice's subsystem to construct a unitary for transformation within the same basis, as well as three grand unitaries for transformation within different bases. The detailed proof is presented in
Appendix \ref{app:unitary}.
\end{proof}
This establishes a one-to-one correspondence between the optimal quantum success probability of our $\mathcal{P}^{4\rightarrow 1}_Q$ entanglement-assisted PMRAC game and certification of the unitary encoding operation performed by Alice.

\emph{An example of state and measurements:} Consider the shared state $\rho_{0000}=\frac{1}{4}\left(\openone_2\otimes\openone_2 - \sigma_x\otimes\sigma_x - \sigma_y\otimes\sigma_y - \sigma_z\otimes\sigma_z\right)$. A transformation in the same basis, defined as $(U_x^\dagger\otimes \openone_2)\rho_{0000}(U_x \otimes \openone_2)\rightarrow \rho_x,\forall x\in\{0011,1100,1111\}$, gives 
$\rho_{0011}=(i\sigma_y\otimes\openone_2)\rho_{0000}(-i\sigma_y\otimes\openone_2),\rho_{1100}=(-i\sigma_x\otimes\openone_2)\rho_{0000}(i\sigma_x\otimes\openone_2),\rho_{1111}=(i\sigma_z\otimes\openone_2)\rho_{0000}(-i\sigma_z\otimes\openone_2).$

To generate the remaining states in the different bases, we apply unitary operations,  such that Eq.~(\ref{ur41}) is satisfied, where $U^1_G=-\frac{\openone_2+ i \ \sigma_y}{\sqrt{2}}$, $U^2_G=\frac{\openone_2+ i \ \sigma_x}{\sqrt{2}}$, and $U^3_G=\frac{\openone_2 +i\ (\sigma_x+\sigma_y-\sigma_z)}{2}$. 

Bob’s measurement operators are of the form
\begin{eqnarray}
    &&B_1=\frac{\sigma_y\otimes (\sigma_z-\sigma_y)}{\sqrt{2}},\quad
    B_2=-\frac{\sigma_y\otimes (\sigma_z+\sigma_y)}{\sqrt{2}},\\
    &&
    \nonumber
    B_3=\frac{(\sigma_z-\sigma_x)\otimes \sigma_x}{\sqrt{2}},\quad
    B_4=-\frac{(\sigma_z+\sigma_x)\otimes \sigma_x}{\sqrt{2}}.
\end{eqnarray}

%%%%%%%%%%%%%%%%%%%%%%%%%%%%%%%%%%%%%%%%%%%%%%%%%%%%%%%%%%%%%%%%%%%%%%%%%%%%%%%%%%%%%%%%%%%%%%%%%%%%%%%%%%%%%%%%%%%%%%%%%%%%%%%%%%%%%%%%%%%%%%%%%%%%%%%%%%%%%%%%%%%%%%%%%%%%%%%%%%
\subsection{\texorpdfstring{$4 \rightarrow 2$}{4 to 2} entanglement-assisted PMRAC}\label{4to2main}
Similar to $4 \to 1$ scenario, Alice and Bob employ same encoding and decoding strategy for same set of input choices with the modification that for $4 \rightarrow 2$ case the shared state is a three-qubit entangled state $\rho\in\mathcal{L}( \mathbb{C}^4\otimes \mathbb{C}^2)$. The quantum success probability for this case is derived as

\begin{eqnarray}
	\label{QSP43}
	\mathcal{P}^{4\rightarrow 2}_{Q}&=&\dfrac{1}{2}+\dfrac{\mathscr{I}^{4\rightarrow 2}}{128}
\end{eqnarray}
 
where $\mathscr{I}^{4\rightarrow 2}$ denotes the correlation function of the form
\begin{eqnarray}\label{p1}
\mathscr{I}^{4\rightarrow2}&=&Tr[(N^1_{4,1}+N^1_{4,2}) B_{4,1}+ (N^2_{4,1}+N^2_{4,2}) B_{4,2}\nonumber \\
&&+ (N^3_{4,1}+N^3_{4,2}) B_{4,3}+( N^4_{4,1}+N^4_{4,2}) B_{4,4}]
\end{eqnarray}

 $N^y_{4,1}$ and $N^y_{4,2}$ where $y\in [4]$ are defined in Appendix.~\ref{optm42} along with the detailed derivation of quantum success probability. 

%\begin{thm} If Alice and Bob share an entangled state $\rho\in\mathcal{L} (\mathbb{C}^4\otimes \mathbb{C}^2)$, and Alice communicates two qubits to Bob. Consequently, Bob performs a two-outcome measurement $\{\Pi^{b}_{4,y}\}$, then the optimal quantum success probability $(\mathcal{P}^{4\rightarrow 2}_Q)^{\mathrm{opt}}=\frac{1}{2}+\frac{\sqrt{3}}{4}$ is achieved.
%\end{thm}
%\begin{proof}
We again employ the same optimization technique to derive the maximum value of $\mathscr{I}^{4\rightarrow 2}$. Then, the correlation function in Eq.~(\ref{p1}) can be expressed as
\begin{eqnarray}
\mathscr{I}^{4\rightarrow 2}&=&\Tr[\sum\limits_{y=1}^{4}\mu_y\mathscr{N}_{4,y} B_{4,y} ]\label{wab43}
\end{eqnarray}
where $\mathscr{N}_{4,y}$ and the $\mu_y$'s are defined as follows
\begin{eqnarray}
     &&\mathscr{N}_{4,y} = \frac{N^y_{4,1}+N^y_{4,2}}{\mu_y} \ ; \ \ 
    \mu_y=||N^y_{4,1}+N^y_{4,2}||  \ , \forall y\in [4] \label{omegai43}
\end{eqnarray}
 such that $||\mathscr{N}_{4,y}||=1$ and $||\cdot||$ is the scaled Frobenious norm, given by $||\mathcal{O}||=\frac{1}{2\sqrt{2}}\sqrt{Tr[\mathcal{O}^{\dagger}\mathcal{O}]}$.
  
 The maximum value of $\mathscr{N}_{4,y} B_{4,y}$ occurs only when $\mathscr{N}_{4,y}=B_{4,y}$, which gives the optimal value as
\begin{equation}\label{optbn43}
\mathscr{I}^{4\rightarrow 2}_{opt}= \Tr[\max\left(\sum\limits_{y=1}^{4}\mu_{y} \right)\openone_8]=8\max\left(\sum\limits_{y=1}^{4}\mu_{y} \right)
\end{equation}

Now, by evaluating the $\mu_y$'s from Eqs.~(\ref{optbn43}) and applying the convex inequality $\sum\limits_{y=1}^{n}\mu_{y}\leq \sqrt{n \sum\limits_{y=1}^{n} (\mu_{y})^{2}}$, we obtain
\begin{eqnarray}\label{A143}
\mathscr{I}^{4\rightarrow 2}&\leq&8\qty( \max \sqrt{4 \ \qty(\mu_1^2+\mu_2^2+\mu_3^2+\mu_4^2)} )\nonumber \\
&\leq&8\max\Bigg(4 \ \Big[8+Tr[N^1_{4,1} N^1_{4,2}+N^2_{4,1}N^2_{4,2}+N^3_{4,1}N^3_{4,2}\nonumber\\
&&+N^4_{4,1} N^4_{4,2}]/4\Big]^\frac{1}{2}\Bigg)
\end{eqnarray}

Then, the maximum value attained by $\mathscr{I}^{4\rightarrow 2}$ is 
\begin{equation}
    \mathscr{I}^{4\rightarrow 2}_{opt}=32\sqrt{3} \label{optbn143}
\end{equation} 
when $ \Tr[N^1_{4,1} N^1_{4,2}+N^2_{4,1}N^2_{4,2}+N^3_{4,1}N^3_{4,2}+N^4_{4,1} N^4_{4,2}]=16$.

Hence, the quantum success probability given in Eq.~(\ref{QSP43}) can be expressed as
\begin{equation}\label{p4max}
(\mathcal{P}^{4\rightarrow 2}_Q)^{opt}=\dfrac{1}{2}+\frac{\sqrt{3}}{4}\approx 0.933
\end{equation}
%\end{proof}

The detailed quantum success probability is derived in detail in Appendix \ref{optm42}. Note that \((\mathcal{P}^{4\rightarrow 2}_Q)^{opt} > \mathcal{P}_C^{4\rightarrow 3} > \mathcal{P}_C^{4\rightarrow 2}\); implying that the optimal quantum success probability surpasses the classical RAC, even when three classical bits are communicated.

However, in this case, it is particularly challenging to provide analytical proof of both the existence and the verification of a grand unitary. In Appendix~\ref{app:42u}, we examine a particular example of a shared three-qubit state of the form $\rho\in\mathcal{L} (\mathbb{C}^4 \otimes \mathbb{C}^2)$ for which we explicitly construct two different orthonormal bases such that the maximal quantum success probability is achieved. However, no grand unitary transformation exists that converts one basis into the other.

\begin{proposition}
\label{prop2}
To achieve the optimal quantum success probability \((\mathcal{P}^{4\rightarrow 2}_Q)^{opt}\)  the initial state must be a GHZ class of maximally entangled state.
\end{proposition}
\begin{proof}
The proof follows the same logic as Proposition~\ref{prop1}. To optimize the winning probability, Alice performs a local unitary operation on her part of the shared initial three-qubit state. As proved in Appendix~\ref{optm42}, the optimal value is achieved if Alice's unitary operations are able to generate two sets of mutually orthogonal bases satisfying  $\sum\limits_{\hspace{-0.1cm}\substack{\bigoplus_{y=1}^4 x_y=0}}\rho_x= \sum\limits_{\hspace{-0.1cm}\substack{ \bigoplus_{y=1}^4 x_y=1}}\rho_x=\openone_8,$  $\forall x\in \{0,1\}^4$ . This is only possible if the initial state belongs to GHZ class of maximally entangled state. 
%{\color{magenta}The insights presented in Appendix~\ref{optm42} reveal an intriguing result: the shared state is, in fact, a pure state. To further clarify the nature of these shared states, we now examine Alice’s encryption protocol in detail. Through her skilled local operations, Alice generates two sets of bases, each consisting of eight distinct states. Within this framework, we explored various candidates for the shared state, focusing in particular on W-type states and states belonging to the GHZ class. Our analysis yielded a clear outcome: the encryption scheme is realized exclusively by GHZ states, while W-states do not exhibit the required behavior.}
\end{proof}
%The considerations underpin GHZ state to be a powerful resource for the $4\rightarrow 2$-bit entanglement-assisted PMRAC game to attain optimal value. 

%%%%%%%%%%%%%%
%We emphasize that, for both the $4\to1$ and $4\to2$ cases, we independently reproduced this bound via explicit numerical optimization. Hence, the value is not merely a theoretically derived upper limit but is, in fact, a tight one. In the following, we present a table comparing the success probability of our variant of PMRAC with that of existing RAC tasks.

%%%%%%%%%%%%%%%
\section{\texorpdfstring{$5$-bit entanglement-assisted PMRAC}{5-bit entanglement-assisted PMRAC}}\label{5PMRAC}

Similarly to $4 \rightarrow l$ entanglement-assisted PMRAC, here we explicitly demonstrate $5 \rightarrow l$ entanglement-assisted PMRAC with $l=1,2$ and $3$. The average classical success probability for one, two, and three bits communications are calculated as $\mathcal{P}^{5\rightarrow 1}_C=0.69$, $\mathcal{P}^{5\rightarrow 2}_C=0.83$, and $\mathcal{P}^{5\rightarrow 3}_C=0.85$, respectively.

\subsection{\texorpdfstring{$5 \to 1$}{5 to 1} entanglement-assisted PMRAC}
	
Analogous to the $4 \rightarrow 1$ scenario, Alice and Bob share a two-qubit maximally entangled state $\rho\in\mathcal{L}(\mathbb{C}^2\otimes \mathbb{C}^2)$. Alice applies a local unitary $U_x$, chosen uniformly at random with $x \in \{0,1\}^5$, and then sends one qubit to Bob. For each bit of $x$ that he intends to retrieve, Bob performs a binary measurement $\{\Pi^b_{5,y}\}$ with $\sum_b \Pi^b_{5,y} = \openone_4,\forall y \in[5], b \in \{0,1\}$. The resulting quantum success probability is then given by
\begin{eqnarray}
	\label{QSP51}
	\mathcal{P}^{5\rightarrow 1}_{Q}&=&\dfrac{1}{2}+\dfrac{\mathscr{I}^{5\rightarrow 1}}{320}
\end{eqnarray}
where $\mathscr{I}^{5\rightarrow 1}$ is the correlation function of the form
\begin{eqnarray}\label{p152}
\mathscr{I}^{5\rightarrow 1}&&=\Tr\Big[
(M^1_{5,1}+M^1_{5,3}+M^1_{5,5}+M^1_{5,7}+M^1_{5,2}+M^1_{5,4}\nonumber\\
&&+M^1_{5,6}
+M^1_{5,8})B_{5,1}+(M^1_{5,1}+M^1_{5,3}-M^1_{5,5}-M^1_{5,7}+M^1_{5,2}\nonumber\\
&&+M^1_{5,4}-M^1_{5,6}-M^1_{5,8})B_{5,2}+(M^1_{5,1}-M^1_{5,3}+M^1_{5,5}-M^1_{5,7}\nonumber\\
&&+M^1_{5,2}-M^1_{5,4}+M^1_{5,6}-M^1_{5,8})B_{5,3}+(M^2_{5,1}+M^2_{5,3}+M^2_{5,5}\nonumber\\
&&+M^2_{5,7}+M^2_{5,2}+M^2_{5,4}+M^2_{5,6}+M^2_{5,8})B_{5,4}+(M^3_{5,1}+M^3_{5,3}\nonumber\\
&&+M^3_{5,5}+M^3_{5,7}+M^3_{5,2}+M^3_{5,4}+M^3_{5,6}+M^3_{5,8})B_{5,5}\Big]
\end{eqnarray}

where $M^y_{5,i},\forall y\in[3], i\in[8]$ are defined in Appendix.~\ref{optm51}. Employing the optimization method adopted in the preceding two sections, the optimal value of the correlation function is obtained as
\begin{equation}\label{A151}
\mathscr{I}^{5\rightarrow 1}\leq4\qty( \max \sqrt{5 \ \qty(\sum\limits_{y=1}^{5}\omega_y^2)}) \leq 4\sqrt{770}
\end{equation}
and substituting Eq.~(\ref{A151}) into Eq.~(\ref{QSP51}), we obtain
\begin{equation}\label{p51max}
\mathcal{P}^{5\rightarrow 1}_Q\leq\dfrac{1}{2}+\dfrac{4\sqrt{770}}{320}\approx 0.847
\end{equation}

 The detailed derivation of the quantum success probability is provided in Appendix.~\ref{optm51}.

%%%%%%%%%%%%%%%%%%%%%%%%%%%%%%%%%%

\subsection{\texorpdfstring{$5 \rightarrow 2$}{5 -> 2} entanglement-assisted PMRAC}

In this case, Alice and Bob share a three-qubit maximally entangled state $\rho\in\mathcal{L}( \mathbb{C}^4\otimes \mathbb{C}^2)$. Repeating the same encoding-decoding process, if Alice communicates two qubits of information to Bob, the quantum success probability can be explicitly derived as

\begin{eqnarray}
	\label{QSP52}
	\mathcal{P}^{5\rightarrow 2}_{Q}&=&\dfrac{1}{2}+\dfrac{\mathscr{I}^{5\rightarrow 2}}{320}
\end{eqnarray}
 
where $\mathscr{I}^{5\rightarrow 2}$ is the correlation function of the form
\begin{eqnarray}\label{I52}
\mathscr{I}^{5\rightarrow 2}&=&Tr[(N^1_{5,1}+N^1_{5,3}+N^1_{5,2}+N^1_{5,4})B_{5,1}\nonumber\\
    &&+(N^1_{5,1}-N^1_{5,3}+N^1_{5,2}-N^1_{5,4})B_{5,2}\nonumber\\
    &&+(N^2_{5,1}+N^2_{5,3}+N^2_{5,2}+N^2_{5,4})B_{5,3}\nonumber\\
&&+(N^3_{5,1}+N^3_{5,3}+N^3_{5,2}+N^3_{5,4})B_{5,4}\nonumber\\
    &&+(N^4_{5,1}+N^4_{5,3}+N^4_{5,2}+N^4_{5,4})B_{5,5}]
\end{eqnarray}
By optimizing Eq.~(\ref{I52}), we obtain the maximum value of 
 $\mathscr{I}^{5\rightarrow 2}$ a
\begin{equation}
    \mathscr{I}^{5\rightarrow 2}_{max}\leq 8\sqrt{5(20+28)}\leq 32\sqrt{15} \label{optbn151}
\end{equation}

 Hence, the quantum success probability from Eq.~(\ref{QSP42}) is derived as
\begin{equation}\label{p52max}
\mathcal{P}^{5\rightarrow 2}_Q\leq \dfrac{1}{2}+\dfrac{32\sqrt{15}}{320}\approx 0.887
\end{equation}
 The detailed derivation of the quantum success probability is provided in the appendix~\ref{optm52}. 
%%%%%%%%%%%%%%%%%%%%%%%%%%%%%%%%%%%%%%%%%%%%%%%%%%%%%%%%%%%%%%%%%%%%%%%%%%%%%%%%
\subsection{\texorpdfstring{$5 \rightarrow 3$}{5 -> 3} entanglement-assisted PMRAC}

Here, Alice and Bob share a four-qubit maximally entangled state $\rho\in\mathcal{L}( \mathbb{C}^8\otimes \mathbb{C}^2)$. If Alice communicates three qubits of information to Bob, the quantum success probability can be explicitly derived as

\begin{eqnarray}
	\label{QSP53}
	\mathcal{P}^{5\rightarrow 3}_{Q}&=&\dfrac{1}{2}+\dfrac{\mathscr{I}^{5\rightarrow 3}}{320}
\end{eqnarray}
 
where $\mathscr{I}^{5\rightarrow 3}$ is the correlation function of the form
\begin{eqnarray}\label{I53m}
\mathscr{I}^{5\rightarrow 3}&=&\Tr[\sum\limits_{y=1}^{5}\eta_y\mathscr{L}_{5,y} B_{5,y}]
\end{eqnarray}
The exact form of $\mathscr{I}^{5\rightarrow 3}$ is given in Appendix~\ref{optm53}. Optimizing Eq.~(\ref{I53m}), we arrive at the maximal value of 
 $\mathscr{I}^{5\rightarrow 3}$, that is,
\begin{eqnarray}
    \mathscr{I}_{opt}^{5\rightarrow 3}\leq 64\sqrt{5}
\end{eqnarray}

 Hence, the quantum success probability from Eq.~(\ref{QSP53}) is derived as
\begin{eqnarray}\label{p53max}
    \mathcal{P}^{5\rightarrow 3}_{Q}\leq\dfrac{1}{2}+\dfrac{64\sqrt{5}}{320}\approx 0.947
\end{eqnarray}
 The detailed derivation of the quantum success probability is provided in the Appendix~\ref{optm53}. 
  
%%%%%%%%%%%%%%%%%%%%%%%%%%%%%%%%%%%%%%%%%%%%%%%%%%%%%%%%%%%%%%%%%%%%%%%%%%%%%%%%%%%%%%%%%%%%%%%%%%%%%%

In the following, we present a table that compares the success probability of our entanglement-assisted PMRAC with that of existing RAC tasks.

\begin{table}[h]
    \centering
    \caption{Here we provide a table displaying upper bounds and quantum success probabilities for $n \rightarrow l$ RAC games with different qubit communication   \cite{ Ambainis2009,imamichi2018constructions, Doriguello2021, Mukherjee2021,Farkas2025}. }
    \label{tab:rac_comparison}. 
    \begin{tabular}{lcccc}
        \toprule
        \textbf{$n \to l$ } & \textbf{Classical } & \textbf{Standard quantum} & \textbf{Entanglement-assisted } \\
        & \textbf{RAC} $
         (\mathcal{P}_C)$ & \textbf{ PMRAC} $
         (\mathcal{P}_{SQ})$& \textbf{PMRAC}$
         (\mathcal{P}_{Q})$&  \\
        \midrule
        $3 \to 1$ & 0.75 & $ 0.785$ &   0.908 (tight) \\
        $4 \to 1$ & 0.69 & $0.741$ &  0.854 (tight) \\
        $4 \to 2$ & 0.75 & 0.854 &  0.933 (tight) \\
        $5 \to 1$ & 0.69 & $< 0.5$ &  0.853 \\
        $5 \to 2$ & 0.83 & 0.811 & 0.887 \\
      $5 \to 3$ & 0.85 & 0.887 &  0.947 \\
        \bottomrule
    \end{tabular}
\end{table}

\begin{figure}[ht]
\centering
\begin{tikzpicture}
\begin{axis}[
    domain=3:93,
    samples=200,
    xlabel = {Value of $(n)$},
    ylabel = {Success Probability $\mathcal{P}^{n\rightarrow n-2}_Q$},
    xmin=3, xmax=93,
    ymin=0.83, ymax=1.0,
    axis lines=box,
    thick,
    width=8.2cm,
    height=6cm,
    tick label style={font=\small},
    legend style={
        at={(0.65,0.1)},  % ← move this to adjust placement inside graph
        anchor=south west,
        draw=none,
        fill=white,
        font=\small
    },
]

% Quantum curve
\addplot[
    blue,
    thick
]
{0.5 + sqrt(x*(x - 1)) / (2 * x)};
\addlegendentry{Quantum}

% Classical curve
\addplot[
    red,
    dashed,
    thick
]
{1 - 1/(2*x)};
\addlegendentry{Classical}

\end{axis}
\end{tikzpicture}
\caption{The Graph depicts how the quantum success probability changes with the number of settings ($n$) of an $n\rightarrow n-2$ entanglement-assisted PMRAC.}\label{sucp}
\end{figure}
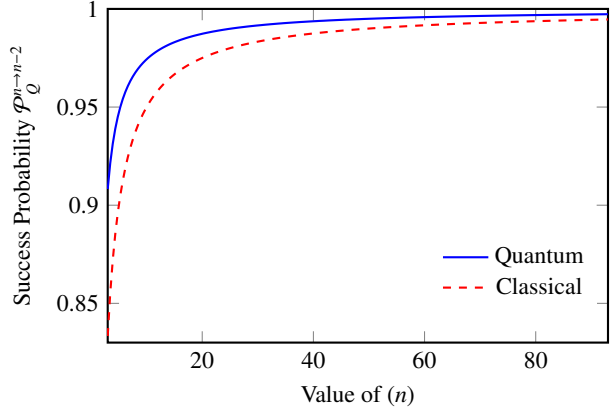
Next, we extend the entanglement-assisted PMRAC game to a general $n$-bit communication scenario by considering $l=(n-2)$ as an example, so that $n-2$ qubits are transmitted from Alice to Bob.
%%%%%%%%%%%%%%%%%%%%%%%%%%%%%%%%%%%%%%%%%%%
\section{\texorpdfstring{$n \rightarrow n-2$}{n to (n-2)} entanglement-assisted PMRAC}\label{nPMRAC}

For the $n\rightarrow n- 2$ classical RAC, Alice sends $l = n-2$ bits to Bob. The upper-bound of the average classical success probability is $\mathcal{P}_C^{n\rightarrow n-2}=(1 - \frac{1}{2n})$ \cite{Nayak1999}. To achieve the bound, Alice and Bob have a prior agreement that she always transmits the first $n-2$ bits, which Bob can decode perfectly, while he guesses the remaining two bits uniformly at random with probability $1/2$.

In quantum theory, Alice and Bob share an entangled state $\rho \in \mathcal{L}(\mathbb{C}^{2^{n-2}}\otimes \mathbb{C}^{2})$. Now, Alice prepares her input states by applying local unitary operation $U_{x}$ \ uniformly at random such that $\rho_{x}=(U^\dagger_{x}\otimes\openone) \ \rho \ (U_{x}\otimes\openone)$ where $x \in\{0,1\}^n $. Then, Alice communicates $n-2$ qubits to Bob. To decode, Bob on the other hand, performs projective measurement $\{\Pi_{n,y}^b\}$  with $\sum_{b}\Pi_{n,y}^b=\openone_{2^{n-1}},\forall y\in[n],b\in\{0,1\}$, for every bit of $x$ he wants to decode. Then,  using Eq. (\ref{VQSP}) the quantum success probability can be written as
\begin{eqnarray}
	\label{B6nn-1}
	\mathcal{P}^{n\rightarrow n-2}_{Q}&=&\dfrac{1}{2}+\dfrac{\mathscr{I}^{n\rightarrow n-2}}{n \ 2^{n+1}}
\end{eqnarray}
 
where the correlation function $\mathscr{I}^{n\rightarrow n-2}$ is derived as
\begin{eqnarray}\label{p1n}
\mathscr{I}^{n\rightarrow n-2}&=&\Tr\Big[\sum_{y=1}^n (M^y_{n,1}+M^y_{n,2})B_{n,y}\Big]
\end{eqnarray}
with 
\begin{eqnarray}\label{Mn1y}
    &&M^y_{n,1}=\sum\limits_{\substack{x\in\{0,1\}^n\\ \bigoplus_{y=1}^n x_y=0}}(-1)^{x_y}\rho_x, M^y_{n,2}=\sum\limits_{\substack{x\in\{0,1\}^n\\ \bigoplus_{y=1}^n x_y=1}}(-1)^{x_y}\rho_x\\
 &&\text{and the states satisfy}\sum\limits_{\substack{x\in\{0,1\}^n\\ \bigoplus_{y=1}^n x_y=0}}\rho_x=\hspace{-0.2cm} \sum\limits_{\substack{x\in\{0,1\}^n\\ \bigoplus_{y=1}^n x_y=1}}\rho_x=\openone_{2^{n-1}}
\end{eqnarray}

 Let $\mathscr{I}^{n\rightarrow n-2}_{opt}$ be the maximum quantum value of the function $\mathscr{I}^{n\rightarrow n-2}$ for the qubit system. The correlation function in Eq.~(\ref{p1n}) can be written as
\begin{eqnarray}
\mathscr{I}^{n\rightarrow n-2}&=&\Tr[\sum\limits_{y=1}^{n}\omega_y\mathscr{M}_{n,y} B_{n,y} ]\label{wabn}
\end{eqnarray}
where $\mathscr{M}_{n,y}$ and $\omega_y$'s are defined as following
\begin{eqnarray}
     &&\mathscr{M}_{n,y} = \frac{M^y_{n,1}+M^y_{n,2}}{\omega_y} \ , \ \ \omega_y=||M^y_{n,1}+M^y_{n,2}||,\ \forall y\in[n] \label{omegain}
\end{eqnarray}
The optimal quantum value of the functional is
\begin{equation}\label{optbnn}
\mathscr{I}^{n\rightarrow n-2}_{opt}= \Tr[\max\left(\sum\limits_{y=1}^{n}\omega_{y} \right)\openone_{2^{n-1}}]=2^{n-1}\max\left(\sum\limits_{y=1}^{n}\omega_{y} \right)
\end{equation}
And it will be achieved only when $\mathscr{M}_{n,y}=B_{n,y}$. Now,
evaluating $\omega_y$ from Eq.~(\ref{omegain}), we get
\begin{eqnarray}
\omega_y&=&\sqrt{2+\langle\{M^y_{n,1}, M^y_{n,2}\}\rangle_\frac{\openone_{2^{n-1}}}{2^{n-1}}}=\sqrt{2+ \frac{\Tr[M^y_{n,1} M^y_{n,2}]}{2^{n-2}}}\ \ 
\end{eqnarray}
Therefore, by applying the convex inequality $\sum\limits_{y=1}^{n}\omega_{y}\leq \sqrt{n \sum\limits_{y=1}^{n} (\omega_{y})^{2}}$, with equality achieved only if $\omega_y=\omega_{y^\prime}$ for all $y\neq y^{\prime}\in [n]$, to Eqs.~(\ref{optbnn}) and (\ref{omegain}), we obtain
\begin{eqnarray}\label{A1n}
\mathscr{I}^{n\rightarrow n-2}&\leq&2^{n-1}\qty( \max \sqrt{n \ \qty(\sum_{y=1}^n\omega_y^2)} )\nonumber \\
&\leq&2^{n-1}\max\qty[n \ \Big[2n+\frac{1}{2^{n-2}}\sum_{y=1}^n\Tr[{M^y_{n,1} M^y_{n,2}}]\Big]]^\frac{1}{2}
\end{eqnarray}
Now, substituting the values of $M^y_{n,1}$ and $M^y_{n,2}$ from Eq.~(\ref{Mn1y}), we get 
\begin{eqnarray}\label{mnnn1}
    \sum_{y=1}^n\Tr[{M^y_{n,1} M^y_{n,2}}]=\hspace{-.5cm}\sum\limits_{\substack{x\in\{0,1\}^n\\ \bigoplus_{y=1}^n x_y=0}}\sum\limits_{\substack{x'\in\{0,1\}^n\\ \bigoplus_{y=1}^n x'_{y}=1}}\hspace{-.3cm}\Tr[\rho_x\rho_{x'}]\sum\limits_{y=1}^n (-1)^{x_y\oplus x'_{y}}
\end{eqnarray}
Note that for $x_y = x'_y$ and $x_y \neq x'_y$, the corresponding values of $(-1)^{x_y\oplus x'_{y}}$ are $+1$ and $-1$, respectively. Now we assume that, for a given $\rho_x$ and $\rho_{x'}$ with opposite parities, the entries differ at position $k_{x,x'}$ and coincide at the remaining $n - k_{x,x'}$ positions. From this assumption, we derive
\begin{eqnarray}
    \sum\limits_{y=1}^n (-1)^{x_y\oplus x'_{y}}=(n-k_{x,x'})-k_{x,x'}=n-2k_{x,x'}
\end{eqnarray}
Putting this in Eq.~(\ref{mnnn1}) we get
\begin{eqnarray}\label{mnnn2}
    \sum_{y=1}^n\Tr[{M^y_{n,1} M^y_{n,2}}]=\hspace{-.5cm}\sum\limits_{\substack{x\in\{0,1\}^n\\ \bigoplus_{y=1}^n x_y=0}}\sum\limits_{\substack{x'\in\{0,1\}^n\\ \bigoplus_{y=1}^n x'_{y}=1}}(n-2k_{x,x'})\Tr[\rho_x\rho_{x'}]
\end{eqnarray}
Since $\rho_x$ and $\rho_{x'}$ have opposite parity, the number of bit positions on which they differ is $k_{x,x'} \in \{1,3,5,\ldots\}$. Furthermore, $0 \leq \Tr[\rho_x \rho_{x'}] \leq 1$ holds. Thus, to maximize Eq.~(\ref{mnnn2}), the factor $n - 2k_{x,x'}$ must be as large as possible, which is attainable only when the number of differing bits is $k_{x,x'} = k = 1$, yielding a value of $(n-2)$ alone. Therefore, the remaining term in Eq.~(\ref{mnnn2}) can be expressed as
\begin{eqnarray}
    \sum\limits_{\substack{x\in\{0,1\}^n\\ \bigoplus_{y=1}^n x_y=0}}\Tr[\rho_x\sum\limits_{\substack{x'\in\{0,1\}^n\\ \bigoplus_{y=1}^n x'_{y}=1}}\rho_{x'}]= \sum\limits_{\substack{x\in\{0,1\}^n\\ \bigoplus_{y=1}^n x_y=0}}\Tr[\rho_x]\leq2^{n-1}
\end{eqnarray}
Substituting this into Eq.~(\ref{mnnn2}) yields

\begin{eqnarray}\label{mnnn3}
    \sum_{y=1}^n\Tr[{M^y_{n,1} M^y_{n,2}}]\leq(n-2)2^{n-1}
\end{eqnarray}
Putting the value in Eq.~(\ref{A1n}), we get  
\begin{equation}
    \mathscr{I}^{n\rightarrow n-2}_{opt}\leq 2^{n}\sqrt{n(n-1)}
\end{equation} 
which gives,, the quantum success probability from Eq.~(\ref{B6nn-1}) as
\begin{equation}\label{nn-2mq}
\mathcal{P}^{n\rightarrow n-2}_Q\leq \dfrac{1}{2}+\dfrac{\sqrt{n(n-1)}}{2n}
\end{equation}
Note that the values of quantum success probability $\mathcal{P}^{n\rightarrow n-2}_Q$ in Table \ref{tab:rac_comparison} satisfy Eq.~(\ref{nn-2mq}) for $n=\{3,4,5\}$.

%In the case of the $4 \rightarrow 2$ protocol, there does not exist any global unitary acting solely on the first two qubits that is able to map one orthonormal basis into another orthonormal basis, which is necessary to achieve the optimal success probability.

%%%%%%%%%%%%%%%%%%%%%%%%%%%%%%%%%%%%%%%%%%%

%We demonstrate that any such $n$-bit communication task where one qubit information is send from Alice to Bob with the assistance of a two-qubit maximally entangled state, not only we get quantum supremacy over classical resources but also demonstrate significant advantage over all the existing quantum RAC protocols.  

%%%%%%%%%%%%%%%%%%%%%%%%%%%%%%%%%%%%%%%%%%%%%%%%%%%%%%%%%%%%%%%%

 %%%%%%%%%%%%%%%%%%%%%%%%%%%%%%%%%%%%%%%%%%%%%%%%%%%%%%%%%%%%%%%%%%%%%%%%%%%%%%%%%
%%%%%%%%%%%%%%%%%%%%%%%%%%%%%%%%%%%%%%%%%%%%%%%%%%%%%%%%%%%%%%%%%%%%%%%%%%%%%%%%%
\section{Summary and Discussion}\label{conc}
In sum, we demonstrate quantum supremacy by leveraging quantum correlations in a hybrid PM communication game where Alice and Bob pre-share an entangled state of fixed dimension and communicate qubits. In particular, we consider a bipartite $n \to l$ entanglement-assisted PMRAC game where Alice applies a set of local unitary operations on her system to encode the information from the given input and communicates $l$ qubit information to Bob, who subsequently performs a joint measurement on the whole system to decode the information of interest.

First, we constructed two cases of $4 \rightarrow l, \forall l\in[2]$ entanglement-assisted PMRAC depending on the number of qubits communicated. We then analytically derived the optimal quantum success probabilities for each case, demonstrating quantum supremacy over classical RACs. Moreover, the optimal quantum success probability of $4\rightarrow 1$ certifies the shared state to be a maximally entangled two-qubit state, Alice's unitary operations, and Bob's measurements. In contrast, for $4\rightarrow 2$ we prove that the optimal value necessarily requires a GHZ state. We then derive the upper bounds of the quantum success probabilities for $5\rightarrow l,\forall l\in[3]$ entanglement-assisted PMRAC and demonstrate a classical advantage. Further, we consider a particular case for a general $n$-bit entanglement-assisted PMRAC game when $n-2$ qubits are communicated and derive the optimal quantum success probability. We leave the general $n$-bit cases for any arbitrary $l(n>l)$ for future work due to substantial analytical and computational demands.  The study of the robustness of our protocols in the presence of noise could be another interesting line of work.

Our conceptual and technical framework lays the necessary foundation for a systematic investigation of entanglement-assisted communication games in PMRAC scenarios. It would be interesting to study other entanglement-assisted communication games in PM scenario by using the analytical technique developed here. % it would be interesting to understand whether our framework can be upgraded to a device-independent setting and to quantify how imperfect state preparation and measurement affect the success probability.

%%%%%%%%%%%%%%%%%%%%%%%%%%%%%%%%%%
\section*{ACKNOWLEDGMENTS} 
RP acknowledges the financial support from the Council of Scientific and Industrial Research (CSIR, 09/1001(12429)/2021-EMR-I), Government of India. PR acknowledges the support by KIAS individual Grant No. QP100601 at the Korea Institute for Advanced Study. AKP acknowledges support from Research Grant No. SG160 of IIT Hyderabad, India.

%%%%%%%%%%%%%%%%%%%%%%%%%%%%%%%%%%
\appendix
\begin{widetext}   		
\section{Detailed derivation of optimal quantum success probability for \texorpdfstring{$4\rightarrow 1$-bit entanglement-assisted PMRAC}{4-bit entanglement-assisted PMRAC}}
\label{app:optimal}
  
To optimize the correlation $\mathscr{I}^{4\rightarrow 1}_{opt}=4\max\left(\sum\limits_{y=1}^{4}\omega_{y} \right)$ in Eq.~(\ref{optbn42}), we first expand $\omega_1$ from Eq.~(\ref{omegai42aap}) in the main text as
\begin{eqnarray}
\omega_1&=&\frac{1}{\sqrt{4}}\sqrt{\Tr[4\openone_4+\qty(\{M^1_{4,1}, M^1_{4,2}\}+\{M^1_{4,1}, M^1_{4,3}\}+\{M^1_{4,1}, M^1_{4,4}\}+\{M^1_{4,2}, M^1_{4,3}\}+\{M^1_{4,2}, M^1_{4,4}\}+\{M^1_{4,3}, M^1_{4,4}\})\openone_4]}\nonumber\\
&=&\sqrt{4+\frac{1}{2}\Tr[M^1_{4,1} M^1_{4,2}+M^1_{4,1}M^1_{4,3}+M^1_{4,1} M^1_{4,4}+M^1_{4,2}M^1_{4,3}+M^1_{4,2} M^1_{4,4}+M^1_{4,3} M^1_{4,4}]}\label{aw1}
\end{eqnarray}
where $\Big\langle\{M^1_{4,i} , M^1_{4,j} \}\Big\rangle_\frac{\openone_4}{4} = \Tr[\{M^1_{4,i} , M^1_{4,j} \}\frac{\openone_4}{4}]$. 

Similarly, we derive,
\begin{eqnarray}\label{omegai142aap}
\omega_2&=&
\sqrt{4+\frac{1}{2}\Tr[-M^1_{4,1}M^1_{4,2}+M^1_{4,1} M^1_{4,3}-M^1_{4,1}M^1_{4,4}-M^1_{4,2} M^1_{4,3}+M^1_{4,2}M^1_{4,4}-M^1_{4,3}M^1_{4,4}]}\label{aw2}\\
\omega_3&=&
\sqrt{4+\frac{1}{2}\Tr[M^2_{4,1} M^2_{4,2}+M^2_{4,1}M^2_{4,3}+M^2_{4,1} M^2_{4,4}+M^2_{4,2}M^2_{4,3}+M^2_{4,2} M^2_{4,4}+M^2_{4,3} M^2_{4,4}]}\label{aw3}\\
\omega_4&=&
\sqrt{4+\frac{1}{2}\Tr[-\{M^2_{4,1} M^2_{4,2}-M^2_{4,1}M^2_{4,3}+M^2_{4,1}M^2_{4,4}+M^2_{4,2}M^2_{4,3}-M^2_{4,2} M^2_{4,4}-M^2_{4,3} M^2_{4,4}]}\label{aw4}
\end{eqnarray}

Using the convex inequality $\sum\limits_{y=1}^{n}\omega_{y}\leq \sqrt{n \sum\limits_{y=1}^{n} (\omega_{y})^{2}}$ where equality holds only when $\omega_y=\omega_{y'}, \forall y\neq y'\in [n]$, hence we can rewrite $\mathscr{I}^{4\rightarrow 1}$ as 
\begin{eqnarray}\label{A142aap}
\mathscr{I}^{4\rightarrow 1}&\leq&4\qty( \max \sqrt{4 \ \qty(\omega_1^2+\omega_2^2+\omega_3^2+\omega_4^2)} )\nonumber \\
&\leq&4\max\Bigg[8 \ \Big[8+\frac{1}{2}\qty(\Tr[M^1_{4,1} M^1_{4,3}]+\Tr[M^1_{4,2} M^1_{4,4}]+\Tr[M^2_{4,1} M^2_{4,4}]+\Tr[M^2_{4,2} M^2_{4,3}])\Big]\Bigg]^{1/2}
\end{eqnarray}
From Eq.~(\ref{obserg42}) of the main text, the mutually orthogonal state satisfies
\begin{eqnarray}
&&\rho_{0000}+\rho_{1111}+\rho_{0011}+\rho_{1100}=\openone_4;\ \ \ \  \rho_{0001}+\rho_{1110}+\rho_{0010}+\rho_{1101}=\openone_4;
	\ \\
	&&\rho_{0101}+\rho_{1010}+\rho_{0110}+\rho_{1001}=\openone_4;\ \ \ \ 
	\rho_{0100}+\rho_{1011}+\rho_{0111}+\rho_{1000}=\openone_4.
\end{eqnarray}
Now, we can write
\begin{eqnarray}
    \Tr[M^1_{4,1}  M^1_{4,3}]&=&\Tr[(\rho_{0000}-\rho_{1111}+\rho_{0011}-\rho_{1100})(\rho_{0001}-\rho_{1110}+\rho_{0010}-\rho_{1101})]\nonumber\\
    &=&\begin{cases}
        4-4\qty(\Tr[\rho_{1111}(\rho_{0001}+\rho_{0010})]+\Tr[\rho_{1100}(\rho_{0001}+\rho_{0010})])\\
        4-4\qty(\Tr[\rho_{0000}(\rho_{1110}+\rho_{1101})]+\Tr[\rho_{0011}(\rho_{1110}+\rho_{1101})])
    \end{cases}\nonumber\\
    &=& 4-2\qty(\Tr[\rho_{1111}(\rho_{0001}+\rho_{0010})]+\Tr[\rho_{1100}(\rho_{0001}+\rho_{0010})]+\Tr[\rho_{0000}(\rho_{1110}+\rho_{1101})]+\Tr[\rho_{0011}(\rho_{1110}+\rho_{1101})])\quad\quad
\end{eqnarray}
Hence, maximum value of $\Tr[M^1_{4,1}  M^1_{4,3}]=4$, only when $\Tr[\rho_{1111}\rho_{0001}]=\Tr[\rho_{1111}\rho_{0010}]=\Tr[\rho_{1100}\rho_{0001}]=\Tr[\rho_{1100}\rho_{0010}]=\Tr[\rho_{0000}\rho_{1110}]=\Tr[\rho_{0000}\rho_{1101}]=\Tr[\rho_{0011}\rho_{1110}]=\Tr[\rho_{0011}\rho_{1101}]=0$. Similarly,
\begin{eqnarray}
    \Tr[M^1_{4,2}  M^1_{4,4}]&=& 4-2\qty(\Tr[\rho_{1010}(\rho_{0100}+\rho_{0111})]+\Tr[\rho_{1001}(\rho_{0100}+\rho_{0111})]+\Tr[\rho_{0101}(\rho_{1011}+\rho_{1000})]+\Tr[\rho_{0110}(\rho_{1011}+\rho_{1000})])\quad\quad\\
    \Tr[M^2_{4,1}  M^2_{4,4}]&=&4-2\qty(\Tr[\rho_{1111}(\rho_{0100}+\rho_{1000})]+\Tr[\rho_{0011}(\rho_{0100}+\rho_{1000})]+\Tr[\rho_{0000}(\rho_{1011}+\rho_{0111})]+\Tr[\rho_{1100}(\rho_{1011}+\rho_{0111})])\\
    \Tr[M^2_{4,2}  M^2_{4,3}]&=&4-2\qty(\Tr[\rho_{1010}(\rho_{0001}+\rho_{1101})]+\Tr[\rho_{0110}(\rho_{0001}+\rho_{1101})]+\Tr[\rho_{0101}(\rho_{1110}+\rho_{0010})]+\Tr[\rho_{1001}(\rho_{1110}+\rho_{0010})])
\end{eqnarray}
Also, the maximum value of $\Tr[M^1_{4,2}  M^1_{4,4}]=\Tr[M^2_{4,1}  M^2_{4,4}]=\Tr[M^2_{4,2}  M^2_{4,3}]=4$ is achieved when
\begin{eqnarray}\label{rho0}
    &&\Tr[\rho_{1010}\rho_{0100}]=\Tr[\rho_{1010}\rho_{0111}]=\Tr[\rho_{1001}\rho_{0100}]=\Tr[\rho_{1001}\rho_{0111}]=\Tr[\rho_{0101}\rho_{1011}]=\Tr[\rho_{0101}\rho_{1000}]=\Tr[\rho_{0110}\rho_{1011}]=0\nonumber\\
    &&\Tr[\rho_{0110}\rho_{1000}]=\Tr[\rho_{1111}\rho_{0100}]=\Tr[\rho_{1111}\rho_{1000}]=\Tr[\rho_{0011}\rho_{0100}]=\Tr[\rho_{0011}\rho_{1000}]=\Tr[\rho_{0000}\rho_{1011}]=\Tr[\rho_{0000}\rho_{0111}]=0\nonumber\\
    &&\Tr[\rho_{1100}\rho_{1011}]=\Tr[\rho_{1100}\rho_{0111}]=0\Tr[\rho_{1010}\rho_{0001}]=\Tr[\rho_{1010}\rho_{1101}]=\Tr[\rho_{0110}\rho_{0001}]=\Tr[\rho_{0110}\rho_{1101}]=\Tr[\rho_{0101}\rho_{1110}]=0\nonumber\\
    &&\Tr[\rho_{0101}\rho_{0010}]=\Tr[\rho_{1001}\rho_{1110}]=\Tr[\rho_{1001}\rho_{0010}]=0\quad\quad
\end{eqnarray}
which further gives
\begin{eqnarray}\label{con1}
    \Tr[M^1_{4,1}  M^1_{4,3} ]+ \Tr[M^1_{4,2} M^1_{4,4} ]+ \Tr[M^2_{4,1} M^2_{4,4}]+ \Tr[M^2_{4,2}  M^2_{4,3}]=16
\end{eqnarray}
which is Eq. (\ref{AMr}) in the main text.

\section{Proof of statement (ii) and (iii) of Corollary \ref{coro1} of the main text}\label{app:coro1}
At the optimal success probability, we have $\omega_1=\omega_2=\omega_3=\omega_4$. Consequently, by considering $\omega_1=\omega_2$ from Eqs.~(\ref{aw1}) and (\ref{aw2}) we obtain
\begin{eqnarray}\label{m41}
     \Tr[M^1_{4,1}  M^1_{4,2}]+\Tr[M^1_{4,1} M^1_{4,4}]+\Tr[M^1_{4,2} M^1_{4,3}]+\Tr[M^1_{4,3} M^1_{4,4}]=0
\end{eqnarray}
Eq.~(\ref{m41}) is further reduced to 
\begin{eqnarray}\label{crho1}
    &&\Tr[(\rho_{0000}-\rho_{1111}+\rho_{0011}-\rho_{1100})(\rho_{0101}-\rho_{1010}+\rho_{0110}-\rho_{1001})]+\Tr[(\rho_{0001}-\rho_{1110}+\rho_{0010}-\rho_{1101})(\rho_{0100}-\rho_{1011}+\rho_{0111}-\rho_{1000})]\nonumber\\
    &&+\Tr[\rho_{0000}(\rho_{0100}-\rho_{1000})]+\Tr[\rho_{1111}(\rho_{1011}-\rho_{0111})]+\Tr[\rho_{0011}(\rho_{0111}-\rho_{1011})]+\Tr[\rho_{1100}(\rho_{1000}-\rho_{0100})]+\Tr[\rho_{0101}(\rho_{0001}-\rho_{1101})]\nonumber\\
    &&+\Tr[\rho_{1010}(\rho_{1110}-\rho_{0010})]+\Tr[\rho_{0110}(\rho_{0010}-\rho_{1110})]+\Tr[\rho_{1001}(\rho_{1101}-\rho_{0001})]=0
\end{eqnarray}
Similarly $\omega_3=\omega_4$ implies that 
\begin{eqnarray}\label{rhoc2}
    \Tr[M^2_{4,1} M^2_{4,2}]+\Tr[M^2_{4,1}M^2_{4,3}]+\Tr[M^2_{4,2} M^2_{4,4}]+\Tr[M^2_{4,3} M^2_{4,4}]=0
\end{eqnarray}
which can further be simplified as
\begin{eqnarray}\label{crho2}
    &&\Tr[(\rho_{0000}-\rho_{1111}-\rho_{0011}+\rho_{1100})(\rho_{0101}-\rho_{1010}-\rho_{0110}+\rho_{1001})]+\Tr[(\rho_{0001}-\rho_{1110}-\rho_{0010}+\rho_{1101})(\rho_{0100}-\rho_{1011}-\rho_{0111}+\rho_{1000})]\nonumber\\
    &&+\Tr[\rho_{0000}(\rho_{0001}-\rho_{0010})]+\Tr[\rho_{1111}(\rho_{1110}-\rho_{1101})]+\Tr[\rho_{0011}(\rho_{0010}-\rho_{0001})]+\Tr[\rho_{1100}(\rho_{1101}-\rho_{1110})]+\Tr[\rho_{0101}(\rho_{0100}-\rho_{0111})]\nonumber\\
    &&+\Tr[\rho_{1010}(\rho_{1011}-\rho_{1000})]+\Tr[\rho_{0110}(\rho_{0111}-\rho_{0100})]+\Tr[\rho_{1001}(\rho_{1000}-\rho_{1011})]=0
\end{eqnarray}
Similarly, using $\omega_1=\omega_4$ or $\omega_2=\omega_3$ and applying the condition of Eq.~(\ref{m41}) and (\ref{rhoc2}), we obtain
\begin{eqnarray}\label{c2m1}
    \Tr[M^1_{4,1}  M^1_{4,3} ]+ \Tr[M^1_{4,2} M^1_{4,4} ]= \Tr[M^2_{4,1} M^2_{4,4}]+ \Tr[M^2_{4,2}  M^2_{4,3}]
\end{eqnarray}
Using Eq.~(\ref{con1}) in Eq.~(\ref{c2m1}) we get 
\begin{eqnarray}\label{c2m}
    \Tr[M^1_{4,1}  M^1_{4,3} ]+ \Tr[M^1_{4,2} M^1_{4,4} ]= \Tr[M^2_{4,1} M^2_{4,4}]+ \Tr[M^2_{4,2}  M^2_{4,3}]=8
\end{eqnarray}
Now, lets first derive the relation between $\rho_{0000}$ with all other $\rho$'s. We already know that in the same basis $\rho_{0000}$ is orthogonal to $\rho_{1111},\rho_{0011}$ and $\rho_{1100}$, i.e, $\Tr[\rho_{0000}\rho_{1111}]=\Tr[\rho_{0000}\rho_{0011}]=\Tr[\rho_{0000}\rho_{1100}]=0$. And from Eq.~(\ref{rho0}) we know that $\Tr[\rho_{0000}\rho_{1101}]=\Tr[\rho_{0000}\rho_{1110}]=\Tr[\rho_{0000}\rho_{0111}]=\Tr[\rho_{0000}\rho_{1011}]=0$. 

To obtain the remaining relations involving $\rho_{0000}$, we note that each trace term in Eqs.~(\ref{crho1}) and (\ref{crho2}) must be zero since $\Tr[\rho_x\rho_{x'}],\forall x\neq x'$ is a positive quantity. Therefore, we arrive at
\begin{eqnarray}
    &&\Tr[\rho_{0000}(\rho_{0101}-\rho_{1010}+\rho_{0110}-\rho_{1001})]=\Tr[\rho_{0000}(\rho_{0101}-\rho_{1010}-\rho_{0110}+\rho_{1001})]=0\label{rho01/4}\\
    &&\Tr[\rho_{0000}(\rho_{0100}-\rho_{1000})]=\Tr[\rho_{0000}(\rho_{0001}-\rho_{0010})]=0
\end{eqnarray}
This implies $\Tr[\rho_{0000}\rho_{0100}]=\Tr[\rho_{0000}\rho_{1000}]$ and $\Tr[\rho_{0000}\rho_{0001}]=\Tr[\rho_{0000}\rho_{0010}]$. Now, we can write 
\begin{eqnarray}
    \Tr[\rho_{0000}]=1\Rightarrow \Tr[\rho_{0000}\openone_4]&&=\begin{cases}
        \Tr[\rho_{0000}(\rho_{0100}+\rho_{1011}+\rho_{0111}+\rho_{1000})]=1\\
        \Tr[\rho_{0000}(\rho_{0001}+\rho_{1110}+\rho_{0010}+\rho_{1101})]=1
    \end{cases}\nonumber\\
    &&=\begin{cases}
        \Tr[\rho_{0000}(\rho_{0100}+\rho_{1000})]=1\\
        \Tr[\rho_{0000}(\rho_{0001}+\rho_{0010})]=1
        \end{cases}\Rightarrow \begin{cases}
        \Tr[\rho_{0000}\rho_{0100}]=\Tr[\rho_{0000}\rho_{1000}]=\frac{1}{2}\\
        \Tr[\rho_{0000}\rho_{0001}]=\Tr[\rho_{0000}\rho_{0010}]=\frac{1}{2}
        \end{cases}
\end{eqnarray}
Similarly, from Eqs.~(\ref{c2m}) and (\ref{rho01/4}), we get
$\Tr[\rho_{0000}\rho_{0101}]=\Tr[\rho_{0000}\rho_{1010}]=\Tr[\rho_{0000}\rho_{0110}]=\Tr[\rho_{0000}\rho_{1001}]=\frac{1}{4}$. The relationship among the various $\Tr[\rho_x \rho_{x'}]$'s are summarized in the following table.
\newpage
\begin{table}[ht]
\centering
\caption{Relations among all the states $\rho$, that is, the quantities $\Tr[\rho_x \rho_{x'}]$ for all $x,x' \in [15]$.}\label{tab1}
\begin{small}
\setlength{\tabcolsep}{4.5pt}
\renewcommand{\arraystretch}{1.2}
\begin{tabular}{c|cccccccccccccccc}
\hline
\hline
 & $\rho_{0000}$ & $\rho_{0001}$ & $\rho_{0010}$ & $\rho_{0011}$ & $\rho_{0100}$ & $\rho_{0101}$ & $\rho_{0110}$ & $\rho_{0111}$ & $\rho_{1000}$ & $\rho_{1001}$ & $\rho_{1010}$ & $\rho_{1011}$ & $\rho_{1100}$ & $\rho_{1101}$ & $\rho_{1110}$ & $\rho_{1111}$ \\
\hline
$\rho_{0000}$    & 1   & 1/2 & 1/2 & 0   & 1/2 & 1/4 & 1/4 & 0   & 1/2 & 1/4 & 1/4 & 0   & 0   & 0   & 0   & 0   \\
$\rho_{0001}$    & 1/2 & 1   & 0   & 1/2 & 1/4 & 1/2 & 0   & 1/4 & 1/4 & 1/2 & 0   & 1/4 & 0   & 0   & 0   & 0   \\
$\rho_{0010}$    & 1/2 & 0   & 1   & 1/2 & 1/4 & 0   & 1/2 & 1/4 & 1/4 & 0   & 1/2 & 1/4 & 0   & 0   & 0   & 0   \\
$\rho_{0011}$    & 0   & 1/2 & 1/2 & 1   & 0   & 1/4 & 1/4 & 1/2 & 0   & 1/4 & 1/4 & 1/2 & 0   & 0   & 0   & 0   \\
$\rho_{0100}$    & 1/2 & 1/4 & 1/4 & 0   & 1   & 1/2 & 1/2 & 0   & 0   & 0   & 0   & 0   & 1/2 & 1/4 & 1/4 & 0   \\
$\rho_{0101}$    & 1/4 & 1/2 & 0   & 1/4 & 1/2 & 1   & 0   & 1/2 & 0   & 0   & 0   & 0   & 1/4 & 1/2 & 0   & 1/4 \\
$\rho_{0110}$    & 1/4 & 0   & 1/2 & 1/4 & 1/2 & 0   & 1   & 1/2 & 0   & 0   & 0   & 0   & 1/4 & 0   & 1/2 & 1/4 \\
$\rho_{0111}$    & 0   & 1/4 & 1/4 & 1/2 & 0   & 1/2 & 1/2 & 1   & 0   & 0   & 0   & 0   & 0   & 1/4 & 1/4 & 1/2 \\
$\rho_{1000}$    & 1/2 & 1/4 & 1/4 & 0   & 0   & 0   & 0   & 0   & 1   & 1/2 & 1/2 & 0   & 1/2 & 1/4 & 1/4 & 0   \\
$\rho_{1001}$    & 1/4 & 1/2 & 0   & 1/4 & 0   & 0   & 0   & 0   & 1/2 & 1   & 0   & 1/2 & 1/4 & 1/2 & 0   & 1/4 \\
$\rho_{1010}$ & 1/4 & 0   & 1/2 & 1/4 & 0   & 0   & 0   & 0   & 1/2 & 0   & 1   & 1/2 & 1/4 & 0   & 1/2 & 1/4 \\
$\rho_{1011}$ & 0   & 1/4 & 1/4 & 1/2 & 0   & 0   & 0   & 0   & 0   & 1/2 & 1/2 & 1   & 0   & 1/4 & 1/4 & 1/2 \\
$\rho_{1100}$ & 0   & 0   & 0   & 0   & 1/2 & 1/4 & 1/4 & 0   & 1/2 & 1/4 & 1/4 & 0   & 1   & 1/2 & 1/2 & 0   \\
$\rho_{1101}$ & 0   & 0   & 0   & 0   & 1/4 & 1/2 & 0   & 1/4 & 1/4 & 1/2 & 0   & 1/4 & 1/2 & 1   & 0   & 1/2 \\
$\rho_{1110}$ & 0   & 0   & 0   & 0   & 1/4 & 0   & 1/2 & 1/4 & 1/4 & 0   & 1/2 & 1/4 & 1/2 & 0   & 1   & 1/2 \\
$\rho_{1111}$ & 0   & 0   & 0   & 0   & 0   & 1/4 & 1/4 & 1/2 & 0   & 1/4 & 1/4 & 1/2 & 0   & 1/2 & 1/2 & 1   \\
\hline
\hline
\end{tabular}
\end{small}
\end{table}
which is statement (ii) in Corollary \ref{coro1}. 

Substituting Eqs.~(\ref{m41}), (\ref{rhoc2}) and (\ref{c2m}) into Eqs.~(\ref{aw1})–(\ref{aw4}), we find that $\omega_y = 2\sqrt{2},\forall y \in [4]$. Consequently, in order to achieve the maximal quantum success probability, Bob’s observables must obey $B_{4,y} = \mathscr{M}_{4,y}$ for all $y \in [4]$, which corresponds to statement (iii) of the Corollary~\ref{coro1}.
%%%%%%%%%%%%%%%%%%%%%%%%%%%%%%%%%%%%%%%%%%%%%%%%%%%%%%%%%%%%%%%%%%%%%%%%%%%%%%%%%%%%%%%%%%%%%%%%%%%%%%
\section{General form of maximally entangled two qubit state in Proposition \ref{prop1}}\label{app:maxent}
The general form \cite{Horodecki1995} of a two-qubit state $\rho$ is
\begin{eqnarray}\label{maxen}
    \rho = \frac{1}{4}\Big[\openone_2\otimes\openone_2+\sum_{i=x,y,z} r_i \ \sigma_i \otimes \openone+\sum_{j=x,y,z} s_j \  \openone\otimes \sigma_j+\sum_{i,j=x,y,z}t_{ij} \ \sigma_i\otimes \sigma_j\Big]\nonumber\\
\end{eqnarray}
with $\sum\limits_{i=x,y,z} r_{i}^{2} +\sum\limits_{j=x,y,z} s_{j}^{2}+\sum\limits_{i,j=x,y,z} t_{ij}^{2}\leq 3$ where the equality holds for pure state. However, for maximally entangled state, Eq. (\ref{maxen}) reduces to
 \begin{eqnarray}\label{maxpure}
    \rho = \frac{1}{4}\Big[\openone_2\otimes\openone_2+\sum_{i,j=x,y,z}^3 t_{ij} \ \sigma_i\otimes \sigma_j\Big]
\end{eqnarray}
with $\sum\limits_{i,j=x,y,z} t_{ij}^{2}=3$. By explicitly writing Eq. (\ref{maxpure}), we get 
 \begin{eqnarray}
    \rho &=& \frac{1}{4}\Big[\openone_2\otimes\openone_2+\sigma_x\otimes(t_{xx}\sigma_x+t_{xy}\sigma_y+t_{xz}\sigma_z)+\sigma_y\otimes(t_{yx}\sigma_x+t_{yy}\sigma_y+t_{yz}\sigma_z)+\sigma_z\otimes(t_{zx}\sigma_x+t_{zy}\sigma_y+t_{zz}\sigma_z)\Big]\nonumber\\
    &=&\frac{1}{4}\Big[\openone_2\otimes\openone_2+c_1 \sigma_x\otimes A_1+c_2 \sigma_y\otimes A_2+c_3 \sigma_z\otimes A_3\Big]\label{gnst}                                                         
\end{eqnarray}
where we define 
\begin{eqnarray}
    A_{1}=\frac{(t_{xx}\sigma_x + t_{xy}\sigma_y + t_{xz}\sigma_z)}{c_1}, A_{2}=\frac{(t_{yx}\sigma_x + t_{yy}\sigma_y + t_{yz}\sigma_z)}{c_2}, A_{3} =\frac{(t_{zx}\sigma_x + t_{zy}\sigma_y + t_{zz}\sigma_z)}{c_3}
\end{eqnarray}

%$A_1$, $A_2$, and $A_3$ as  respectively. Note that $A_i^{2}=\openone_2\ \forall i=\{1,2,3\}$. 
with $c_1 = \sqrt{t_{xx}^2 + t_{xy}^2 + t_{xz}^2}$, $c_2 = \sqrt{t_{yx}^2 + t_{yy}^2 + t_{yz}^2}$, and $c_3 = \sqrt{t_{zx}^2 + t_{zy}^2 + t_{zz}^2}$. 

Since for pure state  $\rho^2-\rho=0$, we get
\begin{eqnarray}
    &&(c_1^2+c_2^2+c_3^2-3)\openone_4+\sigma_x\otimes (-2c_1 A_1+ic_2 c_3 (A_2 A_3-A_3 A_2))+\sigma_y\otimes (-2c_2 A_2+ic_1 c_3 (A_3 A_1-A_1 A_3))\nonumber\\
    &&+\sigma_y\otimes (-2c_3 A_3+ic_1 c_2 (A_1 A_2-A_2 A_1))=0
\end{eqnarray}
which gives the following conditions

\begin{minipage}{0.23\textwidth}
\begin{equation}
c_1^2 + c_2^2 + c_3^2 = 3
\label{e1}
\end{equation}
\end{minipage}\hfill
\begin{minipage}{0.23\textwidth}
\begin{equation}
c_1 A_1 = c_2 c_3 \, i \frac{[A_2, A_3]}{2}
\label{e2}
\end{equation}
\end{minipage}\hfill
\begin{minipage}{0.23\textwidth}
\begin{equation}
c_2 A_2 = c_1 c_3 \, i \frac{[A_3, A_1]}{2}
\label{e3}
\end{equation}
\end{minipage}\hfill
\begin{minipage}{0.23\textwidth}
\begin{equation}
c_3 A_3 = c_1 c_2 \, i \frac{[A_1, A_2]}{2}
\label{e4}
\end{equation}
\end{minipage}
\\

Note that the values of $c_1,c_2$ and $c_3$ should satisfy Eq.~(\ref{e1}) to (\ref{e4}). By taking $A_i=\Vec{n}_i.\vec{\sigma}$ with $|\vec{n_i}|=1$, we get $[A_i,A_j]=[\Vec{n}_i .\vec{\sigma}, \Vec{n}_j .\vec{\sigma}]=2\ i \ (\vec{n}_i\times\vec{n}_j),\forall i\neq j \in \{x,y,z\}$ that transforms Eq.~(\ref{e2}),(\ref{e3}), and (\ref{e4}) to
\begin{eqnarray}
    &&-(\vec{n}_2\cross \vec{n}_3)c_2 c_3=c_1 \vec{n}_1\quad\Rightarrow c_1^2 = -\vec{n}_1.(\vec{n}_2\cross \vec{n}_3) c_1c_2 c_3\label{e5}\\
     &&-(\vec{n}_3\cross \vec{n}_1)c_1 c_3=c_2 \vec{n}_2\quad\Rightarrow c_2^2 = -\vec{n}_2.(\vec{n}_3\cross \vec{n}_1) c_1c_2 c_3=-\vec{n}_1.(\vec{n}_2\cross \vec{n}_3) c_1c_2 c_3\label{e6}\\
      &&-(\vec{n}_1\cross \vec{n}_2)c_1 c_2=c_3\vec{n}_3\quad\Rightarrow c_3^2 = -\vec{n}_3.(\vec{n}_1\cross \vec{n}_2) c_1c_2 c_3=-\vec{n}_1.(\vec{n}_2\cross \vec{n}_3) c_1c_2 c_3\label{e7}
\end{eqnarray}
Now, Eqs.~(\ref{e5}), (\ref{e6}) and (\ref{e7})  give $ c_1^2=c_2^2=c_3^2$ and therefore, from Eq. (\ref{e1}), we get 
\begin{eqnarray}
    -\vec{n}_1.(\vec{n}_2\cross \vec{n}_3) c_1c_2c_3=1\label{e8}
\end{eqnarray}
which consequently gives two conditions
\begin{eqnarray}\label{e9}
   c_1c_2c_3=\pm 1 \ \ and \ \ \vec{n}_1.(\vec{n}_2\cross \vec{n}_3)=\mp 1
\end{eqnarray}
This gives $c_1=\pm 1$, $c_2=\pm 1$ and $c_3=\pm 1$ satisfying $c_1c_2c_3=\pm 1$. Also, $\vec{n}_1$ is either parallel or anti-parallel to $\vec{n}_2 \cross \vec{n}_3$, and similarly, $\vec{n}_2$ and $\vec{n}_3$ are aligned either parallel or anti-parallel to $\vec{n}_3 \cross \vec{n}_1$ and $\vec{n}_1 \cross \vec{n}_2$, respectively. Consequently, Eqs. (\ref{e2}), (\ref{e3}), and (\ref{e4}) lead to 
\begin{eqnarray}\label{e10}
    [A_i,A_j]_{i\neq j}= 2i\epsilon_{ijk}  A_k
\end{eqnarray}
 which is satisfied by three mutually anti-commuting qubit observables. That is $A_1$, $A_2$ and $A_3$ are mutually anti-commuting qubit operators.

%By taking the squared norm we get $|\vec{n}_i|^2|\vec{n}_j|^2-(\vec{n}_i.\vec{n}_j)^2=1\Rightarrow 1-(\vec{n}_i.\vec{n}_j)^2=1$\begin{eqnarray} |\vec{n}_i\times\vec{n}_j|^2=|-\Vec{n}_k|^2 \quad (\text{Taking the squared norm in both-side})\Rightarrow|\vec{n}_i|^2|\vec{n}_j|^2-(\vec{n}_i.\vec{n}_j)^2 \end{eqnarray}Therefore, we get $\vec{n}_i.\vec{n}_j=0,\quad \forall i\neq j$Now, using $\vec{n}_i.\vec{n}_j=0, \forall i\neq j$ we derive the anti-commuting relation as\begin{eqnarray} \{A_i,A_j\}_{i\neq j}=\{\Vec{n}_i .\vec{\sigma}, \Vec{n}_j .\vec{\sigma}\}=2(\Vec{n}_i.\Vec{n}_j) \openone_2=0\end{eqnarray}}Again Solving Eq.~(\ref{e1}) to (\ref{e4}) we get   $c_1=\pm1, c_2=\pm 1$ and $c_3=\pm1$. In our scenario, the shared entangled state is $\rho_{000} = \frac{1}{4}\big(\openone_4 - M_1 - M_2 - M_3\big)$. Here, comparing this with Eq.~(\ref{gnst}) we get \begin{eqnarray} M_1=\mp\sigma_x\otimes A_1;M_2=\mp\sigma_y\otimes A_2; M_3=\mp\sigma_z\otimes A_3\quad \text{(one can choose other combinations also)}\end{eqnarray}

Next, by choosing $c_1=c_2=c_3=-1$, the shared entangled state $\rho$ in Eq.~(\ref{gnst}) can be written as 
\begin{eqnarray}
    \rho = \frac{1}{4}\big(\openone_4 - \sigma_x\otimes A_1 - \sigma_y\otimes A_2 - \sigma_z\otimes A_3\big)
\end{eqnarray}

Hence, one can always use a global unitary to get the initial stats as
\begin{eqnarray}\label{initials}
    \rho_{000}=U \rho \ U^{\dagger}=\frac{1}{4}(\openone_4-Q_1\otimes Q_2 - P_1\otimes P_2 - R_1\otimes R_2)
\end{eqnarray}
where $\{P_1, Q_1, R_1\}$ and $\{P_2, Q_2,R_2\}$ are the sets of mutually anti-commuting qubit operators. This provides our choice of two-qubit operators expressed as a tensor product of observables is in the most general form given in Eq.~(\ref{entan}) of the main text.\\

%%%%%%%%%%%%%%%%%%%%%%%%%%%%%%%%%%%%%%%%%%%%%%%%%%%%%%%%%%%%%%%%%%%%%%%%%%%%%%%%%%%%%%%%%%%%%%%%%%
\section{Construction of Unitaries from the self-testing condition}\label{app:unitary}
From Corollary \ref{coro1} and Proposition \ref{prop1}, we know the states satisfies the following conditions,
\begin{eqnarray}\label{state basis}
&&\rho_{0000}+\rho_{1111}+\rho_{0011}+\rho_{1100}=\openone_4;\ \ \ \  \rho_{0001}+\rho_{1110}+\rho_{0010}+\rho_{1101}=\openone_4; \\
	&&\rho_{0101}+\rho_{1010}+\rho_{0110}+\rho_{1001}=\openone_4;\ \ \ \ 
	\rho_{0100}+\rho_{1011}+\rho_{0111}+\rho_{1000}=\openone_4.
\end{eqnarray}
and the shared states are maximally entangled. Hence, we can write the forms of state in the following way, which will satisfy the Eq.~(\ref{state basis})
\begin{eqnarray}
    &&\rho_{0000}=\frac{\openone_4-u-v-w}{4};\quad \rho_{1111}=\frac{\openone_4+u+v-w}{4};\quad \rho_{0011}=\frac{\openone_4+u-v+w}{4};\quad 
    \rho_{1100}=\frac{\openone_4-u+v+w}{4};\label{s1} \\
    &&\rho_{0001}=\frac{\openone_4+d-e-f}{4};\quad \rho_{1110}=\frac{\openone_4-d+e-f}{4};\quad \rho_{0010}=\frac{\openone_4-d-e+f}{4};\quad  \rho_{1101}=\frac{\openone_4+d+e+f}{4};\label{s2} \\
    &&\rho_{0100}=\frac{\openone_4-x-y+z}{4};\quad \rho_{0111}=\frac{\openone_4+x+y+z}{4};\quad \rho_{1011}=\frac{\openone_4+x-y-z}{4};\quad\rho_{1000}=\frac{\openone_4-x+y-z}{4};  \label{s3}\\
    &&\rho_{0101}=\frac{\openone_4-p-q-r}{4};\quad \rho_{0110}=\frac{\openone_4+p+q-r}{4};\quad \rho_{1010}=\frac{\openone_4+p-q+r}{4};\quad \rho_{1001}=\frac{\openone_4-p+q+r}{4}\label{s4}; 
\end{eqnarray}
where $\{u,v,w\},\{d,e,f\},\{x,y,z\}$, and $\{p,q,r\}$ mutually anti-commuting. The initial state shared by Alice and Bob is a maximally entangled state of the form     \begin{eqnarray}
\label{ent}
    \rho=\frac{1}{4}\left(\openone_4-Q_1\otimes Q_2-P_1\otimes P_2-R_1\otimes R_2\right)
    \end{eqnarray}
as proved in Appendix.~\ref{app:maxent}. Now, comparing $\rho_{0000}$ from Eq.~(\ref{s1}) to Eq.~(\ref{ent}) we get
\begin{eqnarray}\label{QPR}
    u=Q_1\otimes Q_2, \quad v=P_1\otimes P_2, \quad w=R_1\otimes R_2
\end{eqnarray}
where $[u,v]=[v,w]=[w,a]=0,$ with $u=\pm vw$. The qubit operators $P_1$, $Q_1$ and $R_1$ are mutually anti-commuting and the same holds for $P_2$, $Q_2$ and $R_2$. 

Without loss of generality, we consider that Alice and Bob initially share the state $\rho\equiv\rho_{000}$, so that $U_{0000}=\openone_2$. Then 
 \begin{eqnarray}\label{Uijk}
    \rho_{x}=(U^\dagger_{x}\otimes\openone_2) \ \rho_{0000} \ (U_{x}\otimes\openone_2), \quad \forall x\in \{0,1\}^4
\end{eqnarray}
Let us now generate $\rho_{1100}$ from $\rho_{0000}$, i.e., $\rho_{1100}=U^{\dagger}_{1100} \ \rho_{0000}U_{1100}$. This requires $U_{1100}^{\dagger}Q_1 U_{1100}=Q_1$, $U_{1100}^{\dagger} P_1 U_{1100}=-P_1$ and $U_{1100}^{\dagger} R_1 U_{1100}=-R_1$. This in turn gives $U_{1100}=i\ Q_1$. A similar derivation gives $U_{0011}=-i \ P_1$ and $U_{1111}=-i \ R_1$. 

Now, to generate the states in Eq. (\ref{s2}), we simply need to find a grand unitary $U^1_G$ that transforms the basis in Eq. (\ref{s1}) to Eq. (\ref{s2}), i.e., 
\begin{eqnarray}
{U^1_G}^\dagger\{\rho_{0000},\rho_{1111},\rho_{0011},\rho_{1100}\} U^1_G \rightarrow \{\rho_{0001},\rho_{1110},\rho_{0010},\rho_{1101}\}
\end{eqnarray}
which require 
\begin{eqnarray}\label{gu1}
({U^1_{G}}^\dagger\otimes\openone_2)u(U^1_{G}\otimes\openone_2)=-d;\quad ({U^1_{G}}^\dagger\otimes\openone_2)v(U^1_{G}\otimes\openone_2)=e;\quad ({U^1_{G}}^\dagger\otimes\openone_2)w(U^1_{G}\otimes\openone_2)=f
\end{eqnarray}
Now, after putting Eq.~(\ref{s1}) and (\ref{s2}) in Table.~\ref{tab1} we get the following result.
\begin{eqnarray}
    \Tr[ve]=4; \quad\Tr[ud]=\Tr[ue]=\Tr[uf]=\Tr[vd]=\Tr[vf]=\Tr[wd]=\Tr[we]=\Tr[wf]=0\label{tr01}
\end{eqnarray}
In the two qubit scenario $\Tr[ve]=4$. This is only possible when $v=e$, i.e., $e=P_1\otimes P_2$. Since the grand unitary($U_G^1$) operates solely on the first qubit, we can conclude that the second qubits of $d$ and $f$ are identical to those of $u$ and $w$, respectively. Hence, without loss of generality, we can write
\begin{eqnarray}
    d=(\alpha_1 Q_1+\beta_1 P_1 + \gamma_1 R_1)\otimes Q_2; \quad  f=(\alpha_2 Q_1+\beta_2 P_1 + \gamma_2 R_1)\otimes R_2; \quad\forall \{||\alpha_i||,||\beta_i||,||\gamma_i||\}\in[0,1]
\end{eqnarray}
Again we know from Eq.~(\ref{tr01}) that $\Tr[ud]=\Tr[wf]=0$ implies $\alpha_1=\gamma_2=0$, which gives $ d=(\beta_1 P_1 + \gamma_1 R_1)\otimes Q_2; \quad  f=(\alpha_2 Q_1+\beta_2 P_1 )\otimes R_2$. Further we know that
\begin{eqnarray}
    [d,e]=0\Rightarrow \beta_1=0;\quad [e,f]=0 \Rightarrow \beta_2=0
\end{eqnarray}
Hence, $d=R_1\otimes Q_2$ and $f=Q_1\otimes R_2$. By avoiding the second qubit in Eq.~(\ref{gu1}), we can write
\begin{eqnarray}\label{gu1n}
    &&{U^1_{G}}^\dagger Q_1U^1_{G}=-R_1;\quad {U^1_{G}}^\dagger P_1U^1_{G}=P_1;\quad {U^1_{G}}^\dagger R_1 U^1_{G}=Q_1
\end{eqnarray}
Thus, the general form of $U^1_{G}$ can be written as 
\begin{equation}\label{ug1}
U^1_{G}=p_1\openone_2+q_1 Q_1+r_1 P_1+s_1 R_1
\end{equation}where $\{||p_1||,||q_1||,||r_1||,||s_1||\}\in[0,1]$. Putting Eq.~(\ref{gu1n}) into Eq. (\ref{ug1}), we get $p_1=-\frac{1}{\sqrt{2}},q_1=0,r_1=-\frac{i}{\sqrt{2}}$ and $ s_1=0$. Hence, the grand unitary is derived as follows
\begin{eqnarray}
    U^1_{G}=-\frac{\openone_2+i \ P_1}{\sqrt{2}}\equiv \frac{-\openone_2+U_{0011}}{\sqrt{2}}
\end{eqnarray}
Similarly, to generate the states in Eq.~(\ref{s3}), we simply need to find a grand unitary $U^2_G$ that transforms the basis in Eq. (\ref{s1}) to Eq. (\ref{s3}), i.e., 
\begin{eqnarray}
{U^2_G}^\dagger\{\rho_{0000},\rho_{1111},\rho_{0011},\rho_{1100}\} U^2_G \rightarrow \{\rho_{0100},\rho_{0111},\rho_{1011},\rho_{1000}\}
\end{eqnarray}
which require 
\begin{eqnarray}\label{gu2}
({U^2_{G}}^\dagger\otimes\openone_2)u(U^2_{G}\otimes\openone_2)=x;\quad ({U^2_{G}}^\dagger\otimes\openone_2)v(U^2_{G}\otimes\openone_2)=y;\quad ({U^2_{G}}^\dagger\otimes\openone_2)w(U^2_{G}\otimes\openone_2)=- z
\end{eqnarray}
Now, after putting Eq.~(\ref{s1}) and (\ref{s3}) in Table.~\ref{tab1} we get the following relations
\begin{eqnarray}
    \Tr[ux]=4; \quad\Tr[uy]=\Tr[uz]=\Tr[vx]=\Tr[vy]=\Tr[vz]=\Tr[wx]=\Tr[wy]=\Tr[wz]=0\label{tr02}
\end{eqnarray}
In two qubit scenario $\Tr[ux]=4$ it is only possible when $u=x$, i.e., $x=Q_1\otimes Q_2$. As the grand unitary $U_G^2$ acts only on the first qubit, the second qubit of $y$ and $z$ must coincide with those of $v$ and $w$, respectively. Hence, we can write
\begin{eqnarray}
    y=(\alpha_3 Q_1+\beta_3 P_1 + \gamma_3 R_1)\otimes P_2; \quad  z=(\alpha_4 Q_1+\beta_4 P_1 + \gamma_4 R_1)\otimes R_2; \quad\forall \{||\alpha_i||,||\beta_i||,||\gamma_i||\}\in[0,1]
\end{eqnarray}
Again, we know from Eq.~(\ref{tr02}) that $\Tr[vy]=\Tr[wz]=0$ implies $\beta_3=\gamma_4=0$, which gives $ y=(\alpha_3 Q_1 + \gamma_3 R_1)\otimes P_2; \quad  z=(\alpha_4 Q_1+\beta_4 P_1 )\otimes R_2$. In addition, we know that
\begin{eqnarray}
    [x,y]=0\Rightarrow \alpha_3=0;\quad [x,z]=0 \Rightarrow \alpha_4=0
\end{eqnarray}
Hence, $y=R_1\otimes P_2$ and $z=P_1\otimes R_2$. Hence, by avoiding the second qubit in Eq.~(\ref{gu2}), we can write
\begin{eqnarray}\label{gu2n}
    &&{U^2_{G}}^\dagger Q_1U^2_{G}=Q_1;\quad {U^2_{G}}^\dagger P_1U^2_{G}=R_1;\quad {U^2_{G}}^\dagger R_1 U^2_{G}=-P_1
\end{eqnarray}
 The general form of $U^2_{G}$ can be written as 
\begin{equation}\label{ug2}
U^2_{G}=p_2\openone_2+q_2 Q_1+r_2 P_1+s_2 R_1
\end{equation}where $\{||p_2||,||q_2||,||r_2||,||s_2||\}\in[0,1]$. Putting Eq.~(\ref{gu2n}) into Eq. (\ref{ug2}), we get $p_2=\frac{1}{\sqrt{2}},q_2=\frac{i}{\sqrt{2}},r_2=0$ and $ s_2=0$. Hence, the grand unitary is derived as follows
\begin{eqnarray}
    U^2_{G}=\frac{\openone_2+i \ Q_1}{\sqrt{2}}\equiv \frac{\openone_2+U_{1100}}{\sqrt{2}}
\end{eqnarray}
Similarly, to generate the states in Eq.~(\ref{s4}), we simply need to find a grand unitary $U^3_G$ that transforms the basis in Eq. (\ref{s1}) to Eq. (\ref{s4}), i.e., 
\begin{eqnarray}\label{U3rho}
{U^3_G}^\dagger\{\rho_{0000},\rho_{1111},\rho_{0011},\rho_{1100}\} U^3_G \rightarrow \{\rho_{0101},\rho_{0110},\rho_{1010},\rho_{1001}\}
\end{eqnarray}
which require 
\begin{eqnarray}\label{gu3}
({U^3_{G}}^\dagger\otimes\openone_2)u(U^3_{G}\otimes\openone_2)=p;\quad ({U^3_{G}}^\dagger\otimes\openone_2v(U^3_{G}\otimes\openone_2)=q;\quad ({U^3_{G}}^\dagger\otimes\openone_2)w(U^3_{G}\otimes\openone_2)=r
\end{eqnarray}
We substitute Eqs.~(\ref{s2}) and (\ref{s4}) into Table~\ref{tab1} and obtain
\begin{eqnarray}\label{cpx}
    \Tr[dp]=-4;\Tr[dq]=\Tr[dr]=\Tr[ep]=\Tr[eq]=\Tr[er]=\Tr[fp]=\Tr[fq]=\Tr[fr]=0
\end{eqnarray}
This implies $p=-d=-R_1\otimes Q_2$. Similarly putting Eq.~(\ref{s3}) and (\ref{s4}) in Table.~\ref{tab1} we get
\begin{eqnarray}
    \Tr[zq]=-4;\Tr[xp]=\Tr[xq]=\Tr[xr]=\Tr[yp]=\Tr[yq]=\Tr[yr]=\Tr[zp]=\Tr[zr]=0
\end{eqnarray}
This further implies $r=-z=-P_1\otimes R_2$.\\
 As the grand unitary($U_G^3$) operates solely on the first qubit, we can write
\begin{eqnarray}
    q=(\alpha_5 Q_1+\beta_5 P_1 + \gamma_5 R_1)\otimes P_2; \quad\forall \{||\alpha_i||,||\beta_i||,||\gamma_i||\}\in[0,1]
\end{eqnarray}
From Eq.~(\ref{cpx}), we observe that $\Tr[eq]=0$ leads to $\beta_5=0$, which in turn yields $ q=(\alpha_5 Q_1 + \gamma_5 R_1)\otimes P_2$. Moreover, we have 
\begin{eqnarray}
    [p,q]=0\Rightarrow \gamma_5=0
\end{eqnarray}
 Hence, $q=Q_1\otimes P_2$. As the grand unitary acts only on the first qubit, the second qubit in Eq.~(\ref{gu3}) can be dropped, yielding
\begin{eqnarray}\label{gu3n}
    &&{U^3_{G}}^\dagger Q_1U^3_{G}=-R_1;\quad {U^3_{G}}^\dagger P_1U^3_{G}=Q_1;\quad {U^3_{G}}^\dagger R_1 U^3_{G}=-P_1
\end{eqnarray}
 Thus, the general form of $U^3_{G}$ can be written as 
\begin{equation}\label{ug3}
U^3_{G}=p_3\openone_2+q_3 Q_1+r_3 P_1+s_3 R_1
\end{equation}where $\{||p_3||,||q_3||,||r_3||,||s_3||\}\in[0,1]$. Putting Eq.~(\ref{gu3n}) into Eq. (\ref{ug3}), we get $p_3=\frac{1}{2},q_3=\frac{i}{2},r_3=\frac{i}{2}$ and $ s_3=-\frac{i}{2}$. Hence, the grand unitary is derived as follows
\begin{eqnarray}
    U^3_G=\frac{\openone_2 +i\ (Q_1+P_1-R_1)}{2}\equiv \frac{\openone_2+U_{1100}-U_{0011}+U_{1111}}{2}
\end{eqnarray}
%%%%%%%%%%%%%%%%%%%%%%%%%%%%%%%%%%%%%%%%%%%%%%%%%%%%%%%%%%%%%%%%%%%%%%%%%%%%%%%%%%%%%%%%%%%%%%%%%%%%%%%%%%%%%%%%%%%%%%%%%%%%%%%%%%%%%%%%%%%%%%%%%%%%%%%%%%%%%%%%%%
\section{Detailed derivation of optimal quantum success probability for \texorpdfstring{$4 \to 2$}{4→2} entanglement-assisted PMRAC}
\label{optm42}
In the $4\rightarrow 2$ entanglement-assisted PMRAC game, the quantum success probability from Eq. (\ref{VQSP}) in the main text can be explicitly written as
\begin{eqnarray}
	\label{B5}
	\nonumber\mathcal{P}^{4\rightarrow2}_{Q}=\dfrac{1}{2}+\dfrac{1}{128}&\Tr&[(\rho_{0000}-\rho_{1111}+\rho_{0011}-\rho_{1100}+\rho_{0101}-\rho_{1010}+\rho_{0110}-\rho_{1001})B_{4,1}\\
	\nonumber
	&+&(\rho_{0000}-\rho_{1111}+\rho_{0011}-\rho_{1100}-\rho_{0101}+\rho_{1010}-\rho_{0110}+\rho_{1001})B_{4,2}\\
	\nonumber
	&+&
	(\rho_{0000}-\rho_{1111}-\rho_{0011}+\rho_{1100}+\rho_{0101}-\rho_{1010}-\rho_{0110}+\rho_{1001})B_{4,3}\\
	\nonumber
	&+&
	(\rho_{0000}-\rho_{1111}-\rho_{0011}+\rho_{1100}-\rho_{0101}+\rho_{1010}+\rho_{0110}-\rho_{1001})B_{4,4}\\
	&+&
	(\rho_{0001}-\rho_{1110}+\rho_{0010}-\rho_{1101}+\rho_{0100}-\rho_{1011}+\rho_{0111}-\rho_{1000})B_{4,1}\\
	\nonumber
	&+&(\rho_{0001}-\rho_{1110}+\rho_{0010}-\rho_{1101}-\rho_{0100}+\rho_{1011}-\rho_{0111}+\rho_{1000})B_{4,2}\\
	\nonumber
	&+&
	(\rho_{0001}-\rho_{1110}-\rho_{0010}+\rho_{1101}+\rho_{0100}-\rho_{1011}-\rho_{0111}+\rho_{1000})B_{4,3}\\
	\nonumber
	&+&
	(-\rho_{0001}+\rho_{1110}+\rho_{0010}-\rho_{1101}+\rho_{0100}-\rho_{1011}-\rho_{0111}+\rho_{1000})B_{4,4}]
\end{eqnarray}
Let us consider first term in Eq. (\ref{B5}). To attain its maximal value, the states $\rho_{0000},\rho_{1111},\rho_{0011},\rho_{1100},\rho_{0101},\rho_{1010},\rho_{0110}$, and $\rho_{1001}$ must all be eigenstates of $B_{4,1}$. This requirement implies that these states form a set of mutually orthogonal pure states and satisfy 
\begin{eqnarray}
    &&\rho_{0000}+\rho_{1111}+\rho_{0011}+\rho_{1100}+\rho_{0101}+\rho_{1010}+\rho_{0110}+\rho_{1001}
=\openone_8\\
&&\rho_{0001}+\rho_{1110}+\rho_{0010}+\rho_{1101}+\rho_{0100}+\rho_{1011}+\rho_{0111}+\rho_{1000}=\openone_8
\end{eqnarray}

An analogous reasoning applies to the remaining terms in Eq. (\ref{B5}). Let us assume 
\begin{eqnarray}
	\nonumber
	N^1_{4,1}&=&\rho_{0000}-\rho_{1111}+\rho_{0011}-\rho_{1100}+\rho_{0101}-\rho_{1010}+\rho_{0110}-\rho_{1001}\\
	\nonumber
	N^2_{4,1}&=&\rho_{0000}-\rho_{1111}+\rho_{0011}-\rho_{1100}-\rho_{0101}+\rho_{1010}-\rho_{0110}+\rho_{1001}\\
	\nonumber
	N^3_{4,1}&=&
	\rho_{0000}-\rho_{1111}-\rho_{0011}+\rho_{1100}+\rho_{0101}-\rho_{1010}-\rho_{0110}+\rho_{1001}\\
	\nonumber
	N^4_{4,1}&=&
	\rho_{0000}-\rho_{1111}-\rho_{0011}+\rho_{1100}-\rho_{0101}+\rho_{1010}+\rho_{0110}-\rho_{1001}\\
	N^1_{4,2}&=&
	\rho_{0001}-\rho_{1110}+\rho_{0010}-\rho_{1101}+\rho_{0100}-\rho_{1011}+\rho_{0111}-\rho_{1000}\\
	\nonumber
	N^2_{4,2}&=&\rho_{0001}-\rho_{1110}+\rho_{0010}-\rho_{1101}-\rho_{0100}+\rho_{1011}-\rho_{0111}+\rho_{1000}\\
	\nonumber
	N^3_{4,2}&=&
	\rho_{0001}-\rho_{1110}-\rho_{0010}+\rho_{1101}+\rho_{0100}-\rho_{1011}-\rho_{0111}+\rho_{1000}\\
	\nonumber
	N^4_{4,2}&=&-\rho_{0001}+\rho_{1110}+\rho_{0010}-\rho_{1101}+\rho_{0100}-\rho_{1011}-\rho_{0111}+\rho_{1000}
\end{eqnarray}

Then the quantum success probability in Eq. (\ref{B5}) can be casted as
\begin{eqnarray}\label{QSP42aap}
	\mathcal{P}^{4\rightarrow2}_{Q}&=&\dfrac{1}{2}+\dfrac{\mathscr{I}_{4\rightarrow2}}{128}.
\end{eqnarray}
where the function $\mathscr{I}^{4\rightarrow2}$ is given by 
\begin{eqnarray}
\mathscr{I}^{4\rightarrow2}=\Tr[(N^1_{4,1}+N^1_{4,2}) B_{4,1}+ (N^2_{4,1}+N^2_{4,2}) B_{4,2}+ (N^3_{4,1}+N^3_{4,2}) B_{4,3}+( N^4_{4,1}+N^4_{4,2}) B_{4,4}]
\end{eqnarray}

Now, we derive the QM maximum value of $\mathscr{I}^{4\rightarrow2}$. The correlation function in Eq.~(\ref{p1}) from the main text can be written as
\begin{eqnarray}
\mathscr{I}^{4\rightarrow2}&=&\Tr[\sum\limits_{y=1}^{4}\mu_y\mathscr{N}_{4,y} B_{4,y} ]\label{wab43aap}
\end{eqnarray}
where $\mathscr{N}_{y}$ and $\omega_y$'s are defined as following
\begin{eqnarray}
     &&\mathscr{N}_{4,1} = \frac{N^1_{4,1}+N^1_{4,2}}{\mu_1} \ ; \ \ \mathscr{N}_{4,2} = \frac{N^2_{4,1}+N^2_{4,2}}{\mu_2} \ ; \ \ \mathscr{N}_{4,3} = \frac{N^3_{4,1}+N^3_{4,2}}{\mu_3} \ ; \ \ \mathscr{N}_{4,4} = \frac{ N^4_{4,1}+N^4_{4,2}}{\mu_4}; \ \ \mu_y=||N^j_{4,1}+N^j_{4,2}||\ \forall j=y. \nonumber \label{omegai43aap}
\end{eqnarray}
  such that $||\mathscr{N}_{4,y}||=1$ and $||\cdot||$ is the scaled Frobenious norm, given by $||\mathcal{O}||=\frac{1}{\sqrt{8}}\sqrt{Tr[\mathcal{O}^{\dagger}\mathcal{O}]}$.
  
  Interestingly, as $\mathscr{N}_{4,y}$ and $ B_{4,y}$ are dichotomic normalized observables, the maximum value of $\mathscr{N}_{4,y} B_{4,y}$ occurs only when $\mathscr{N}_{4,y}=B_{4,y}$, which further implies
\begin{equation}\label{optbn42aap}
\mathscr{I}^{4\rightarrow2}_{opt}= \Tr[\max\left(\sum\limits_{y=1}^{4}\mu_{y} \right)\openone_8]=8\max\left(\sum\limits_{y=1}^{4}\mu_{y} \right)
\end{equation}

Now, by evaluating $\mu_1$ from Eq.~(\ref{omegai43aap}), we get
\begin{eqnarray}
\mu_1&=&\frac{1}{\sqrt{8}}\sqrt{\Tr[2\openone_8+\{N^1_{4,1}, N^1_{4,2}\}\openone_8]}=\sqrt{2+\Tr[N^1_{4,1} N^1_{4,2}]/4}
\end{eqnarray}
Similarly,
\begin{eqnarray}
&&\mu_{2}=\sqrt{2  + \Tr[N^2_{4,1}N^2_{4,2}]/4};\quad
\mu_{3}=\sqrt{2+\Tr[N^3_{4,1}N^3_{4,2}]/4};\quad
\mu_{4}=\sqrt{2+\Tr[N^4_{4,1}N^4_{4,2}]/4} \label{omegai143aap}
\end{eqnarray}

Next, using the convex inequality $\sum\limits_{y=1}^{n}\mu_y\leq \sqrt{n \sum\limits_{y=1}^{n} (\mu_{y})^{2}}$ where equality holds only when $\mu_y=\mu_{y}, \forall y\neq y'\in [n]$, from Eqs.~(\ref{optbn42aap}), we get
\begin{eqnarray}\label{A143aap}
\mathscr{I}^{4\rightarrow2}&\leq&8\qty( \max \sqrt{4 \ \qty(\mu_1^2+\mu_2^2+\mu_3^2+\mu_4^2)} )\nonumber \\
&\leq&8\max\Bigg[4 \ \Big(8+\Tr[N^1_{4,1} N^1_{4,2}+N^2_{4,1},N^2_{4,2} +N^3_{4,1},N^3_{4,2}+N^4_{4,1}, N^4_{4,2}]/4\Big)\Bigg]^\frac{1}{2}
\end{eqnarray}
Rewriting Alice's anti-commuting relations in terms of $\rho_x$ and simplifying, we get
\begin{eqnarray} \nonumber
\mathscr{I}^{4\rightarrow2}\leq&&8\max \Bigg[4 \ \Big[12-\Big( \Tr[\rho_{0000}( \rho_{1110}+\rho_{1101}+\rho_{1011}+\rho_{0111})]+\Tr[\rho_{1111}(\rho_{0001}+\rho_{0010}+\rho_{0100}+\rho_{1000})]\nonumber\\
&&\hspace{0.8cm}+\Tr[\rho_{0011}( \rho_{1110}+\rho_{1101}+\rho_{0100}+\rho_{1000})]+ \Tr[\rho_{1100}(\rho_{0001}+\rho_{0010}+\rho_{1011}+\rho_{0111})]\nonumber\\
&&\hspace{0.8cm}+\Tr[\rho_{0101}(\rho_{1110}+\rho_{0010}+\rho_{1011}+\rho_{1000})]+\Tr[\rho_{1010}(\rho_{0001}+\rho_{1101}+\rho_{0100}+\rho_{0111})]\\
&&\hspace{0.8cm}+\Tr[\rho_{0110}(\rho_{0001}+\rho_{1101}+\rho_{1011}+\rho_{1000})]+\Tr[\rho_{1001}(\rho_{1110}+\rho_{0010}+\rho_{0100}+\rho_{0111})]\Big)\Big]\Bigg]^\frac{1}{2}\nonumber
\end{eqnarray}
Thus, the maximum quantum value of the Bell functional $\mathscr{I}^{4\rightarrow 2}$ is 
\begin{equation}
    \mathscr{I}^{4\rightarrow 2}_{opt}=32\sqrt{3} \label{optbn143aap}
\end{equation} 
when $\Tr[\rho_{0000}( \rho_{1110}+\rho_{1101}+\rho_{1011}+\rho_{0111})]=\Tr[\rho_{1111}(\rho_{0001}+\rho_{0010}+\rho_{0100}+\rho_{1000})]=\Tr[\rho_{0011}( \rho_{1110}+\rho_{1101}+\rho_{0100}+\rho_{1000})]=\Tr[\rho_{1100}(\rho_{0001}+\rho_{0010}+\rho_{1011}+\rho_{0111})]=\Tr[\rho_{0101}(\rho_{1110}+\rho_{0010}+\rho_{1011}+\rho_{1000})]=\Tr[\rho_{1010}(\rho_{0001}+\rho_{1101}+\rho_{0100}+\rho_{0111})]=\Tr[\rho_{0110}(\rho_{0001}+\rho_{1101}+\rho_{1011}+\rho_{1000})]=\Tr[\rho_{1001}(\rho_{1110}+\rho_{0010}+\rho_{0100}+\rho_{0111})]=0$ as $\Tr[\rho_{i} \rho_{j}]\in \{0,1\}$.  Also, this gives $ \Tr[N^1_{4,1} N^1_{4,2}+N^2_{4,1},N^2_{4,2} +N^3_{4,1},N^3_{4,2}+N^4_{4,1}, N^4_{4,2}]=16$.

Hence, the quantum success probability from Eq.~(\ref{QSP42aap}) is derived as
\begin{equation}
\mathcal{P}^{4\rightarrow2}_Q=\dfrac{1}{2}+\dfrac{32\sqrt{3}}{128}=\dfrac{1}{2}+\frac{\sqrt{3}}{4} = 0.933
\end{equation}
which is Eq. (\ref{p4max}) in the main text.

\subsection{An illustrative \texorpdfstring{4$\rightarrow$2}{4→2} example of states and measurements that attains the optimal success probability}\label{app:42u}
\noindent
 Let, the joint state shared by Alice and Bob is a GHZ state of the following form,
{\small\begin{eqnarray}
\rho_{0000} = \frac{1}{8}\Big(
\openone_2 \otimes \openone_2 \otimes \openone_2
+ \sigma_z \otimes \sigma_z \otimes \openone_2
+ \sigma_z \otimes \openone_2 \otimes \sigma_z
+ \openone_2 \otimes \sigma_z \otimes \sigma_z
+ \sigma_x \otimes \sigma_x \otimes \sigma_x
- \sigma_x \otimes \sigma_y \otimes \sigma_y
- \sigma_y \otimes \sigma_x \otimes \sigma_y
- \sigma_y \otimes \sigma_y \otimes \sigma_x
\Big)\nonumber \\
\end{eqnarray}}
where Alice holds first two qubits while Bob holds one.

Alice performs local unitary operations $(U_{x\in[15]} = U_{x\in\{0,1\}^4})$ on her part of the system and subsequently transmits it to Bob. In this manner, with certain states belonging to the same basis, the first set of states can be generated as
{\small\begin{eqnarray}
    \bigl(\{U_{15},U_{12},U_{10},U_9,U_6,U_5,U_3,U_0\}^\dagger\otimes \openone_2\bigr)\,
    \rho_{0000}\,
    \bigl(\{U_{0},U_{3},U_{5},U_{6},U_{9},U_{10},U_{12},U_{15}\}\otimes \openone_2\bigr)
    \rightarrow \{\rho_{0000}, \rho_{0011},\rho_{0101},\rho_{0110}, \rho_{1001}, \rho_{1010}, \rho_{1100}, \rho_{1111}\}\, \nonumber \\
\end{eqnarray}}
where the relevant two–qubit unitaries are specified by
{\small\begin{eqnarray}
    &&U_0=\openone_2 \otimes \openone_2,\quad
    U_3=\sigma_z\otimes \openone_2,\quad
    U_5=\sigma_x\otimes \openone_2,\quad
    U_6=\sigma_x\otimes \sigma_z,\quad
    U_9=\openone_2\otimes\sigma_x,\quad
    U_{10}=\sigma_z\otimes\sigma_x,\quad
    U_{12}=\sigma_x\otimes\sigma_x,\quad
    U_{15}=\sigma_z\sigma_x\otimes\sigma_x\nonumber 
\end{eqnarray}}
Analogously, the second set of states is produced by
{\small\begin{eqnarray}
    \bigl(\{U_{14},U_{13},U_{11},U_8,U_7,U_4,U_2,U_1\}^\dagger\otimes\openone_2\bigr)\,
    \rho_{0000}\,
    \bigl(\{U_{1},U_{2},U_{4},U_{7},U_{8},U_{11},U_{13},U_{14}\}\otimes\openone_2\bigr)
    \rightarrow \{\rho_{0001}, \rho_{0010},\rho_{0100},\rho_{0111}, \rho_{1000}, \rho_{1011}, \rho_{1101}, \rho_{1110}\}\, .\nonumber \\
\end{eqnarray}}
The associated unitary operators are obtained in Python and are given by
\begin{eqnarray}
&&U_1 =
\begin{pmatrix}
0.6505 + 0.2909 i & 0.4710 + 0.1293 i & 0.1324 - 0.4858 i & 0.0015 - 0.0008 i \\
-0.0952 + 0.0799 i & 0.5104 - 0.2831 i & 0.1441 + 0.3768 i & -0.6573 + 0.2217 i \\
-0.4151 - 0.5518 i & 0.3163 + 0.2556 i & 0.4335 - 0.3934 i & 0.0377 + 0.1176 i \\
0.0029 + 0.0019 i & 0.1328 - 0.4876 i & 0.4721 + 0.1347 i & 0.2726 - 0.6552 i
\end{pmatrix}\\
&&U_2 =
\begin{pmatrix}
0.6857 - 0.1709 i & 0.4563 - 0.1945 i & -0.1975 - 0.4642 i & 0.0029 - 0.0012 i \\
0.2285 - 0.1949 i & 0.2998 + 0.0097 i & 0.3094 + 0.5641 i & -0.6258 + 0.1194 i \\
0.6006 - 0.2230 i & -0.6049 + 0.2175 i & 0.2270 + 0.1886 i & 0.2944 + 0.0516 i \\
-0.0026 - 0.0010 i & 0.1987 + 0.4597 i & -0.4555 + 0.1923 i & 0.1809 + 0.6870 i
\end{pmatrix}\\
&&U_4 =
\begin{pmatrix}
0.2888 - 0.4172 i & 0.0031 + 0.0008 i & 0.6356 + 0.3027 i & 0.4080 + 0.2836 i \\
0.3442 - 0.0981 i & -0.2874 - 0.6176 i & 0.0758 + 0.1610 i & -0.1680 - 0.5899 i \\
0.6016 - 0.0993 i & -0.0545 - 0.1681 i & -0.2190 - 0.6517 i & 0.0560 + 0.3481 i \\
0.4063 + 0.2784 i & -0.3066 + 0.6409 i & 0.0010 - 0.0034 i & 0.2863 - 0.4131 i
\end{pmatrix}\\
&&U_7 =
\begin{pmatrix}
-0.4944 - 0.1044 i & 0.0001 + 0.0009 i & 0.0371 - 0.7098 i & -0.0978 + 0.4795 i \\
-0.1400 + 0.0116 i & -0.1679 + 0.3799 i & -0.1098 + 0.5579 i & -0.2905 + 0.6326 i \\
-0.6818 + 0.1377 i & -0.1402 - 0.5530 i & 0.1880 + 0.3690 i & 0.0434 - 0.1314 i \\
0.0996 - 0.4808 i & -0.7069 - 0.0483 i & -0.0010 + 0.0002 i & 0.4969 + 0.1000 i
\end{pmatrix}\\
&&U_8 =
\begin{pmatrix}
0.4429 - 0.2267 i & -0.5788 + 0.4089 i & 0.0019 - 0.0004 i & -0.2304 - 0.4440 i \\
-0.3536 - 0.1243 i & 0.1655 + 0.0012 i & 0.6763 + 0.1107 i & 0.1555 - 0.5817 i \\
-0.3500 + 0.4906 i & -0.2989 + 0.6173 i & 0.0941 - 0.1346 i & 0.3194 + 0.1933 i \\
-0.2270 - 0.4441 i & -0.0014 + 0.0002 i & -0.4026 - 0.5843 i & 0.4423 - 0.2284 i
\end{pmatrix}\\
&&U_{11} =
\begin{pmatrix}
0.2891 - 0.3938 i & -0.3273 + 0.6328 i & -0.0015 + 0.0023 i & 0.4086 + 0.2947 i \\
0.5087 + 0.1420 i & 0.3735 + 0.3806 i & 0.3675 - 0.2714 i & -0.4751 - 0.0488 i \\
-0.4544 - 0.1442 i & -0.3214 + 0.3237 i & -0.3229 - 0.4296 i & -0.5196 - 0.0767 i \\
-0.4126 - 0.2915 i & -0.0019 + 0.0014 i & 0.6224 + 0.3393 i & -0.2829 + 0.4028 i
\end{pmatrix}\\
&&U_{11} =
\begin{pmatrix}
0.2891 - 0.3938 i & -0.3273 + 0.6328 i & -0.0015 + 0.0023 i & 0.4086 + 0.2947 i \\
0.5087 + 0.1420 i & 0.3735 + 0.3806 i & 0.3675 - 0.2714 i & -0.4751 - 0.0488 i \\
-0.4544 - 0.1442 i & -0.3214 + 0.3237 i & -0.3229 - 0.4296 i & -0.5196 - 0.0767 i \\
-0.4126 - 0.2915 i & -0.0019 + 0.0014 i & 0.6224 + 0.3393 i & -0.2829 + 0.4028 i
\end{pmatrix}\\
&&U_{13} =
\begin{pmatrix}
-0.0034 + 0.0038 i & -0.2448 + 0.4395 i & -0.4215 - 0.2303 i & 0.4143 - 0.5870 i \\
0.3400 + 0.5120 i & 0.3523 + 0.0875 i & 0.5725 - 0.2231 i & 0.1372 - 0.3067 i \\
0.2441 + 0.2243 i & 0.1747 + 0.5927 i & -0.2889 - 0.2325 i & -0.0970 + 0.6012 i \\
0.5817 + 0.4170 i & -0.4264 - 0.2270 i & -0.2445 + 0.4411 i & -0.0013 + 0.0087 i
\end{pmatrix}\\
&&U_{14} =
\begin{pmatrix}
-0.0008 + 0.0014 i & -0.4920 - 0.0800 i & 0.0765 - 0.4905 i & -0.6778 - 0.2137 i \\
-0.2808 - 0.1102 i & -0.1509 + 0.3202 i & -0.2182 - 0.5789 i & 0.5596 - 0.2964 i \\
-0.1949 - 0.6065 i & 0.4147 + 0.4539 i & 0.3496 + 0.0057 i & -0.2839 - 0.1153 i \\
-0.2192 + 0.6746 i & -0.0752 + 0.4926 i & 0.4918 + 0.0817 i & -0.0046 - 0.0030 i
\end{pmatrix}
\end{eqnarray}

Building on the optimization condition introduced in the main text, we can express the relevant observables in the following general form
 \begin{eqnarray} B_{1} &=& \frac{\sqrt{3}}{2}X_1 - \frac{12}{83}X_{14} - \frac{1}{7}X_{12} - \frac{12}{85}X_{13} + \frac{8}{57}X_{15}  + \frac{11}{94}X_{25} - \frac{8}{69}X_{27} + \frac{9}{79}X_{29} + \frac{7}{62}X_{31} + \frac{9}{80}X_{24} + \frac{9}{80}X_{28} + \frac{10}{89}X_{26}  \nonumber \\&&+ \frac{11}{100}X_{30} + \frac{2}{23}X_{39} + \frac{7}{81}X_{36} - \frac{7}{82}X_{45} - \frac{7}{82}X_{46} + \frac{7}{82}X_{40} - \frac{8}{95}X_{47} + \frac{8}{95}X_{44} + \frac{8}{95}X_{37} \end{eqnarray}

\begin{eqnarray} B_{2} &=& \frac{\sqrt{3}}{2}X_6 - \frac{14}{97}X_{12} + \frac{1}{7}X_{13} + \frac{1}{7}X_{14} + \frac{1}{7}X_{15}- \frac{5}{42}X_{18} - \frac{9}{76}X_{16} + \frac{2}{17}X_{22} + \frac{11}{94}X_{23} - \frac{7}{60}X_{17}  - \frac{5}{43}X_{19} + \frac{3}{26}X_{21} - \frac{4}{35}X_{20} \nonumber \\&&+ \frac{1}{11}X_{32} + \frac{9}{100}X_{33} + \frac{3}{34}X_{41} + \frac{8}{91}X_{38} + \frac{5}{57}X_{43} - \frac{7}{80}X_{34} + \frac{8}{93}X_{35}+ \frac{7}{82}X_{42} \end{eqnarray}

\begin{eqnarray} B_{3} &=& \frac{\sqrt{3}}{2}X_2 + \frac{4}{27}X_8 - \frac{13}{89}X_9 + \frac{7}{48}X_{11} - \frac{7}{48}X_{10} + \frac{5}{42}X_{17} - \frac{9}{76}X_{19} - \frac{2}{17}X_{16} + \frac{5}{43}X_{18} - \frac{11}{95}X_{24}+ \frac{11}{95}X_{26} + \frac{3}{26}X_{25} + \frac{4}{35}X_{27} \nonumber \\&&  + \frac{9}{98}X_{35} + \frac{7}{79}X_{34}- \frac{7}{79}X_{37} - \frac{5}{57}X_{33} + \frac{7}{81}X_{36} - \frac{8}{93}X_{39} + \frac{3}{35}X_{32}+ \frac{3}{35}X_{40} \end{eqnarray}

\begin{eqnarray}
B_{4} &=& 
- \frac{\sqrt{3}}{2}X_5 - \frac{5}{34}X_9 - \frac{7}{48}X_8 - \frac{10}{69}X_{11} - \frac{14}{97}X_{10} - \frac{1}{8}X_{21} - \frac{1}{8}X_{20} + \frac{1}{8}X_{23} - \frac{7}{61}X_{22} - \frac{9}{79}X_{31} + \frac{5}{44}X_{29} + \frac{11}{97}X_{30} - \frac{11}{97}X_{28}  \nonumber \\ && - \frac{3}{34}X_{42} + \frac{3}{34}X_{38} + \frac{2}{23}X_{44} + \frac{2}{23}X_{43} - \frac{4}{47}X_{41}+ \frac{5}{59}X_{45} - \frac{6}{71}X_{47} + \frac{8}{95}X_{46} 
\end{eqnarray}
where we defined the $X_{i\in[57]}$ operators in the following form
{\small\begin{eqnarray}
    &&X_0=\openone_2 \otimes \openone_2 \otimes \openone_2,X_1=\openone_2 \otimes \sigma_z \otimes \sigma_z,X_2=\openone_2 \otimes \sigma_z \otimes \sigma_z,X_3=\sigma_x \otimes \sigma_y \otimes \sigma_y, X_4=\sigma_y \otimes \sigma_x \otimes \sigma_y, X_5=\sigma_y \otimes \sigma_y \otimes \sigma_x, X_6=\sigma_z \otimes \openone_2 \otimes \sigma_z,\nonumber\\
    &&X_7=\sigma_z \otimes \sigma_z \otimes \openone_2, X_8=\sigma_x \otimes \sigma_y \otimes \sigma_x, X_9=\sigma_y \otimes \sigma_y \otimes \sigma_y, X_{10}=\sigma_x \otimes \sigma_x \otimes \sigma_y, X_{11}=\sigma_y \otimes \sigma_x \otimes \sigma_x, X_{12}=\sigma_x \otimes \sigma_y \otimes \sigma_z, X_{13}=\openone_2 \otimes \sigma_z \otimes \sigma_y,\nonumber\\
    && X_{14}=\sigma_z \otimes \openone_2 \otimes \sigma_y, X_{15}=\sigma_y \otimes \sigma_x \otimes \sigma_z, X_{16}=\sigma_x \otimes \openone_2 \otimes \sigma_z, X_{17}=\sigma_z \otimes \sigma_y \otimes \sigma_y, X_{18}=\sigma_x \otimes \sigma_z \otimes \openone_2, X_{19}=\sigma_z \otimes \sigma_x \otimes \sigma_x, X_{20}=\sigma_z \otimes \sigma_x \otimes \sigma_y,\nonumber\\ &&X_{21}=\sigma_z \otimes \sigma_y \otimes \sigma_x, X_{22}=\sigma_y \otimes \openone_2 \otimes \sigma_z, X_{23}=\sigma_y \otimes \openone_2 \otimes \sigma_z, X_{24}=\openone_2 \otimes \sigma_y \otimes \openone_2, X_{25}=\sigma_z \otimes \sigma_y \otimes \sigma_z, X_{26}=\sigma_y \otimes \openone_2 \otimes \sigma_x, X_{27}=\sigma_x \otimes \openone_2 \otimes \sigma_y, \nonumber\\
    &&X_{28}=\sigma_x \otimes \openone_2 \otimes \sigma_x, X_{29}=\sigma_y \otimes \openone_2 \otimes \sigma_y, X_{30}=\openone_2 \otimes \sigma_x \otimes \openone_2, X_{31}=\sigma_z \otimes \sigma_x \otimes \sigma_z, X_{32}=\openone_2 \otimes \sigma_y \otimes \sigma_x, X_{33}=\openone_2 \otimes \sigma_y \otimes \sigma_x, X_{34}=\openone_2 \otimes \sigma_x \otimes \sigma_y,\nonumber\\
    &&X_{35}=\sigma_y \otimes \sigma_z \otimes \sigma_z, X_{36}=\openone_2 \otimes \sigma_x \otimes \sigma_z, X_{37}=\sigma_y \otimes \sigma_z \otimes \sigma_y, X_{38}=\sigma_x \otimes \openone_2 \otimes \openone_2, X_{39}=\sigma_z \otimes \sigma_x \otimes \openone_2, X_{40}=\sigma_x \otimes \sigma_z \otimes \sigma_x, X_{41}=\openone_2 \otimes \sigma_x \otimes \sigma_x\nonumber\\
    &&X_{42}=\sigma_x \otimes \sigma_z \otimes \sigma_z, X_{43}=\openone_2 \otimes \sigma_y \otimes \sigma_y, X_{44}=\sigma_x \otimes \sigma_z \otimes \sigma_y, X_{45}=\openone_2 \otimes \sigma_y \otimes \sigma_z, X_{46}=\sigma_y \otimes \sigma_z \otimes \sigma_x, X_{47}=\sigma_z \otimes \sigma_y \otimes \openone_2, X_{48}=\sigma_z \otimes \sigma_z \otimes \sigma_x,\nonumber\\
    &&X_{49}=\sigma_x \otimes \sigma_y \otimes \openone_2, X_{50}=\sigma_x \otimes \sigma_x \otimes \sigma_z, X_{51}=\sigma_z \otimes \sigma_z \otimes \sigma_z, X_{52}=\sigma_y \otimes \sigma_y \otimes \sigma_z, X_{53}=\sigma_y \otimes \sigma_x \otimes \openone_2, X_{54}=\openone_2 \otimes \sigma_z \otimes \sigma_x, X_{55}=\sigma_z \otimes \openone_2 \otimes \sigma_x\nonumber\\
    &&X_{56}=\sigma_y \otimes \sigma_y \otimes \openone_2, X_{57}=\sigma_x \otimes \sigma_x \otimes \openone_2\nonumber
\end{eqnarray}}
Substituting all the observables and states in Eq.~(\ref{B5}), we get the optimal quantum success probability as 

\begin{equation}
(\mathcal{P}^{4\rightarrow2}_Q)^{opt}\approx 0.933
\end{equation}
which is Eq. (\ref{p4max}) in the main text.
%%%%%%%%%%%%%%%%%%%%%%%%%%%%%%%%%%%%%%%%%%%%%%%%%%%%%%%%%%%%%%%%%%%%%%%%%%%%%%%%%%%%%%%%%%%%%%%%%%%%%%%%%%%%%%%%%%%%%%%%%%%%%%%%%%%%%%%%%%%%%%%%%%%%%%%%%%%%%%%%%%%%%%%%%%%%%%%%%%%%%%%%%%%%%%%%%%%%%%%%%%%%%%%%%%%%%%%%%%%%%%%%%%%%%%%%%%%%%

%%%%%%%%%%%%%%%%%%%%%%%%%%%%%%%%%%%%%%%%%%%%%%%%%%%%%%%%%%%%%%%%%%%%%%%%%%%%%%%%%%%%%%%%%%%%%%%%%%%%%%%%%%%%%%%%%%%%%%%%%%%%%%%%%%%%%%%%%%%%%%%%%%%%%%%%%%%%%%%%%%%%%%%%%%%%%%%%%%%%%%%%%%%%%%%%%%%%%%%%%%%%%%%%%%%%%%%%%%%%%%%%%%%%%%%%%%%%%

\section{Detailed derivation of optimal quantum success probability for \texorpdfstring{$5\rightarrow 1$}{5→1} entanglement-assisted PMRAC}
\label{optm51}
In the $5\rightarrow 1$ entanglement-assisted PMRAC game, the quantum success probability from Eq. (\ref{VQSP}) in the main text can be explicitly written as
\begin{eqnarray}
\mathcal{P}^{5\rightarrow 1}_{Q}=\dfrac{1}{160}\sum_{\{x_1,x_2,x_3,x_4,x_5\}\in\{0,1\}} \Tr[\rho_{x_1 x_2 x_3 x_4 x_5} \qty((-1)^{x_1}\Pi_{5,1}^{x_1}+(-1)^{x_2}\Pi_{5,2}^{x_2}+(-1)^{x_3}\Pi_{5,3}^{x_3}+(-1)^{x_4}\Pi_{5,4}^{x_4}+(-1)^{x_5}\Pi_{5,5}^{x_5})]
	\label{eq58}
\end{eqnarray}
Considering $\Pi_{5,y}^{0}=(\textbf{1}+B_{5,y})/2, \forall y\in[5]$, Eq. (\ref{eq58}) can be rewritten as
\begin{eqnarray}\label{suc51}
	\nonumber
	\mathcal{P}^{5\rightarrow 1}_{Q}=\dfrac{1}{2}&+&\dfrac{1}{320}\ Tr[(M^1_{5,1}+M^1_{5,3}+M^1_{5,5}+M^1_{5,7}+M^1_{5,2}+M^1_{5,4}+M^1_{5,6}+M^1_{5,8})B_{5,1}+(M^1_{5,1}+M^1_{5,3}-M^1_{5,5}-M^1_{5,7}+M^1_{5,2}+M^1_{5,4}\nonumber\\
    &&\hspace{0.8cm}-M^1_{5,6}-M^1_{5,8})B_{5,2}+(M^1_{5,1}-M^1_{5,3}+M^1_{5,5}-M^1_{5,7}+M^1_{5,2}-M^1_{5,4}+M^1_{5,6}-M^1_{5,8})B_{5,3}+(M^2_{5,1}+M^2_{5,3}+M^2_{5,5}\nonumber\\
    &&\hspace{0.8cm}+M^2_{5,7}+M^2_{5,2}+M^2_{5,4}+M^2_{5,6}+M^2_{5,8})B_{5,4}+(M^3_{5,1}+M^3_{5,3}+M^3_{5,5}+M^3_{5,7}+M^3_{5,2}+M^3_{5,4}+M^3_{5,6}+M^3_{5,8})B_{5,5}]
\end{eqnarray}
This leads us to define two-qubit states of the form 
\begin{eqnarray}
	\label{obserg51}
	\nonumber
	M^1_{5,1}&=&\rho_0+\rho_3-\rho_{29}-\rho_{30}; \ \ \ 
	M^1_{5,3}=
	\rho_{5}+\rho_{6}-\rho_{24}-\rho_{27}; \ \ \
	M^1_{5,5}=\rho_{9}+\rho_{10}-\rho_{20}-\rho_{23}; \ \ \ M^1_{5,7}=\rho_{12}+\rho_{15}-\rho_{17}-\rho_{18}\\
    \nonumber
	M^2_{5,1}&=&
	\rho_{0}-\rho_{3}+\rho_{29}-\rho_{30};\ \ \ 
 M^2_{5,3}=
	\rho_{5}-\rho_{6}+\rho_{24}-\rho_{27};\ \ \
	M^2_{5,5}=
	\rho_{9}-\rho_{10}+\rho_{20}-\rho_{23};\ \ \ 
	M^2_{5,7}=\rho_{12}-\rho_{15}+\rho_{17}-\rho_{18}\\
    M^3_{5,1}&=&
	\rho_{0}-\rho_{3}-\rho_{29}+\rho_{30};\ \ \ 
 M^3_{5,3}=
	-\rho_{5}+\rho_{6}+\rho_{24}-\rho_{27};\ \ \ 
    M^3_{5,5}=
	-\rho_{9}+\rho_{10}+\rho_{20}-\rho_{23};\ \ \ 
	M^3_{5,7}=\rho_{12}-\rho_{15}-\rho_{17}+\rho_{18}\ \ \\
    M^1_{5,2}&=&\rho_{1}+\rho_{2}-\rho_{28}-\rho_{31};\ \ \ M^1_{5,4}=\rho_{4}+\rho_{7}-\rho_{25}-\rho_{26};\ \ \
    M^1_{5,6}=\rho_{8}+\rho_{11}-\rho_{21}-\rho_{22};\ \ \ M^1_{5,8}=\rho_{13}+\rho_{14}-\rho_{16}-\rho_{19}\nonumber\\
    M^2_{5,2}&=&\rho_{1}-\rho_{2}+\rho_{28}-\rho_{31};\ \ \ M^2_{5,4}=\rho_{4}-\rho_{7}+\rho_{25}-\rho_{26};\ \ \
    M^2_{5,6}=\rho_{8}-\rho_{11}+\rho_{21}-\rho_{22};\ \ \ M^2_{5,8}=\rho_{13}-\rho_{14}+\rho_{16}-\rho_{19}\nonumber\\
    M^3_{5,2}&=&-\rho_{1}+\rho_{2}+\rho_{28}-\rho_{31};\ \ \ M^3_{5,4}=\rho_{4}-\rho_{7}-\rho_{25}+\rho_{26};\ \ \
    M^3_{5,6}=\rho_{8}-\rho_{11}-\rho_{21}+\rho_{22};\ \ \ M^3_{5,8}=-\rho_{13}+\rho_{14}+\rho_{16}-\rho_{19}\nonumber
\end{eqnarray}

where
\begin{eqnarray}
	&&\rho_0+\rho_3+\rho_{29}+\rho_{30}=\openone_4; \ \ \ 
	\rho_{5}+\rho_{6}+\rho_{24}+\rho_{27}=\openone_4;\ \ \ \rho_{9}+\rho_{10}+\rho_{20}+\rho_{23}=\openone_4; \ \ \ \rho_{12}+\rho_{15}+\rho_{17}+\rho_{18}=\openone_4\nonumber\\
     &&\rho_{1}+\rho_{2}+\rho_{28}+\rho_{31}=\openone_4;\ \ \ \rho_{4}+\rho_{7}+\rho_{25}+\rho_{26}=\openone_4\nonumber; \ \ \rho_{8}+\rho_{11}+\rho_{21}+\rho_{22}=\openone_4;\ \ \ \rho_{13}+\rho_{14}+\rho_{16}+\rho_{19}=\openone_4
\end{eqnarray}
Then the quantum success probability in Eq. (\ref{suc51}) can be casted as
\begin{eqnarray}\label{QSP51aap}
	\mathcal{P}^{5\rightarrow 1}_{Q}&=&\dfrac{1}{2}+\dfrac{\mathscr{I}_{5\rightarrow 1}}{320}
\end{eqnarray}
where the function $\mathscr{I}^{5\rightarrow1}$ is given by 
\begin{eqnarray}\label{I51}
\mathscr{I}^{5\rightarrow 1}&=&Tr[(M^1_{5,1}+M^1_{5,3}+M^1_{5,5}+M^1_{5,7}+M^1_{5,2}+M^1_{5,4}+M^1_{5,6}+M^1_{5,8})B_{5,1}+(M^1_{5,1}+M^1_{5,3}-M^1_{5,5}-M^1_{5,7}+M^1_{5,2}+M^1_{5,4}\nonumber\\
    &&\hspace{0.cm}-M^1_{5,6}-M^1_{5,8})B_{5,2}+(M^1_{5,1}-M^1_{5,3}+M^1_{5,5}-M^1_{5,7}+M^1_{5,2}-M^1_{5,4}+M^1_{5,6}-M^1_{5,8})B_{5,3}+(M^2_{5,1}+M^2_{5,3}+M^2_{5,5}\nonumber\\
    &&\hspace{0.4cm}+M^2_{5,7}+M^2_{5,2}+M^2_{5,4}+M^2_{5,6}+M^2_{5,8})B_{5,4}+(M^3_{5,1}+M^3_{5,3}+M^3_{5,5}+M^3_{5,7}+M^3_{5,2}+M^3_{5,4}+M^3_{5,6}+M^3_{5,8})B_{5,5}]
\end{eqnarray}

Now, we derive the QM maximum value of $\mathscr{I}_{5\rightarrow 1}$. The correlation function in Eq.~(\ref{I51})  can be written as
\begin{eqnarray}
\mathscr{I}^{5\rightarrow 1}&=&\Tr[\sum\limits_{y=1}^{5}\omega_y\mathscr{M}_{5,y} B_{5,y} ]\label{wab51aap}
\end{eqnarray}
where $\mathscr{M}_{5,y}$ and $\omega_y$'s are defined as following
\begin{eqnarray}
     &&\mathscr{M}_{5,1} = \frac{M^1_{5,1}+M^1_{5,3}+M^1_{5,5}+M^1_{5,7}+M^1_{5,2}+M^1_{5,4}+M^1_{5,6}+M^1_{5,8}}{\omega_1} \ ; \ \ \mathscr{M}_{5,2} = \frac{M^1_{5,1}+M^1_{5,3}-M^1_{5,5}-M^1_{5,7}+M^1_{5,2}+M^1_{5,4}-M^1_{5,6}-M^1_{5,8}}{\omega_2}\nonumber\\
     &&\mathscr{M}_{5,3} = \frac{M^1_{5,1}-M^1_{5,3}+M^1_{5,5}-M^1_{5,7}+M^1_{5,2}-M^1_{5,4}+M^1_{5,6}-M^1_{5,8}}{\omega_3}\ ; \ \ \mathscr{M}_{5,4} = \frac{M^2_{5,1}+M^2_{5,3}+M^2_{5,5}+M^2_{5,7}+M^2_{5,2}+M^2_{5,4}+M^2_{5,6}+M^2_{5,8}}{\omega_4}\nonumber\\
     &&\mathscr{M}_{5,5} = \frac{M^3_{5,1}+M^3_{5,3}+M^3_{5,5}+M^3_{5,7}+M^3_{5,2}+M^3_{5,4}+M^3_{5,6}+M^3_{5,8}}{\omega_5}\ \nonumber \\
  &&\omega_1=||M^1_{5,1}+M^1_{5,3}+M^1_{5,5}+M^1_{5,7}+M^1_{5,2}+M^1_{5,4}+M^1_{5,6}+M^1_{5,8}||;\ \
    \omega_2=||M^1_{5,1}+M^1_{5,3}-M^1_{5,5}-M^1_{5,7}+M^1_{5,2}+M^1_{5,4}-M^1_{5,6}-M^1_{5,8}|| \nonumber\\&&\omega_3=||M^1_{5,1}-M^1_{5,3}+M^1_{5,5}-M^1_{5,7}+M^1_{5,2}-M^1_{5,4}+M^1_{5,6}-M^1_{5,8}||;\ \
    \omega_4=||M^2_{5,1}+M^2_{5,3}+M^2_{5,5}+M^2_{5,7}+M^2_{5,2}+M^2_{5,4}+M^2_{5,6}+M^2_{5,8}|| \nonumber\\
     &&\omega_5=||M^3_{5,1}+M^3_{5,3}+M^3_{5,5}+M^3_{5,7}+M^3_{5,2}+M^3_{5,4}+M^3_{5,6}+M^3_{5,8}|| \quad  \label{omegai51aap}
\end{eqnarray}
  such that $||\mathscr{M}_{5,y}||=1$ and $||\cdot||$ is the scaled Frobenious norm, given by $||\mathcal{O}||=\frac{1}{\sqrt{4}}\sqrt{Tr[\mathcal{O}^{\dagger}\mathcal{O}]}$.
  
  Interestingly, as $\mathscr{M}_{5,y}$ and $ B_{5,y}$ are dichotomic normalized observables, the maximum value of $\mathscr{M}_{5,y} B_{5,y}$ occurs only when $\mathscr{M}_{5,y}=B_{5,y}$, which further implies
\begin{equation}\label{optbn51aap}
\mathscr{I}^{5\rightarrow 1}_{opt}= \Tr[\max\left(\sum\limits_{y=1}^{5}\omega_{y} \right)\openone_4]=4\max\left(\sum\limits_{y=1}^{5}\omega_{y} \right)
\end{equation}
Now, by evaluating $\omega_1$ from Eq.~(\ref{omegai51aap}), we get
\begin{eqnarray}
\omega_1
&=&\Bigg[8+\frac{1}{2}\Big(\Tr[M^1_{5,1}M^1_{5,3}+M^1_{5,1}M^1_{5,5}+M^1_{5,1}M^1_{5,7} + M^1_{5,1}M^1_{5,2} + M^1_{5,1}M^1_{5,4} + M^1_{5,1}M^1_{5,6} +M^1_{5,1}M^1_{5,8} +M^1_{5,3}M^1_{5,5}] \nonumber\\
&+& \Tr[ M^1_{5,3}M^1_{5,7} + M^1_{5,3}M^1_{5,2}+M^1_{5,3}M^1_{5,4}+ M^1_{5,3}M^1_{5,6} +M^1_{5,3}M^1_{5,8}+M^1_{5,5}M^1_{5,7}+M^1_{5,5}M^1_{5,2}+M^1_{5,5}M^1_{5,4} ]\nonumber\\
&+&\Tr[M^1_{5,5}M^1_{5,6}+M^1_{5,5} M^1_{5,8} +  + M^1_{5,7}M^1_{5,4} + M^1_{5,7}M^1_{5,6}+ M^1_{5,7}M^1_{5,8} + M^1_{5,2}M^1_{5,4} + M^1_{5,2}M^1_{5,6}]\nonumber\\
&+&\Tr[M^1_{5,2}M^1_{5,8} + M^1_{5,4} M^1_{5,6} + M^1_{5,4}M^1_{5,8} + M^1_{5,6}M^1_{5,8}]\Big) \Bigg]^\frac{1}{2}
\end{eqnarray}

Similarly, evaluating $\omega_2,\omega_3,\omega_4$ and $\omega_5$ and using the convex inequality $\sum\limits_{y=1}^{5}\omega_{y}\leq \sqrt{5 \sum\limits_{y=1}^{5} (\omega_{y})^{2}}$ where equality holds only when $\omega_y=\omega_{y'} \forall y\neq y\in [5]$, from Eqs.~(\ref{optbn51aap}), we get
\begin{eqnarray}\label{A151aap}
\mathscr{I}^{5\rightarrow 1}&\leq&4\qty( \max \sqrt{5\ \qty(\omega_1^2+\omega_2^2+\omega_3^2+\omega_4^2+\omega_5^2)} )\nonumber \\
&\leq&4\max\Big[5 \ \Big[40+\frac{1}{2}\Big(\Tr[ M^1_{5,1}M^1_{5,3}+M^2_{5,1}M^2_{5,3}+M^3_{5,1}M^3_{5,3}]+\Tr[ M^1_{5,1}M^1_{5,5}+M^2_{5,1}M^2_{5,5}+M^3_{5,1}M^3_{5,3}]\nonumber\\
&&+\Tr[ -M^1_{5,1}M^1_{5,7}+\{M^2_{5,1},M^2_{5,7}\}+M^3_{5,1}M^3_{5,7}]+\Tr[ 3M^1_{5,1}M^1_{5,2}+M^2_{5,1}M^2_{5,2}+M^3_{5,1}M^3_{5,2}]+Tr[ M^1_{5,1}M^1_{5,4}\nonumber\\
&&+M^2_{5,1}M^2_{5,4}+M^3_{5,1}M^3_{5,4}]+\Tr[ M^1_{5,1}M^1_{5,6}+M^2_{5,1}M^2_{5,6}+M^3_{5,1}M^3_{5,6}]+Tr[ -M^1_{5,1}M^1_{5,8}+M^2_{5,1}M^2_{5,8}\nonumber\\
&&+M^3_{5,1}M^3_{5,8}]+\Tr[ -M^1_{5,3}M^1_{5,5}+M^2_{5,3}M^2_{5,5}+M^3_{5,3}M^3_{5,5}]+\Tr[ M^1_{5,3}M^1_{5,7}+M^2_{5,3}M^2_{5,7}+M^3_{5,3}M^3_{5,7}] \nonumber\\
&&+\Tr[ M^1_{5,3}M^1_{5,2}+M^2_{5,3}M^2_{5,2}+M^3_{5,3}M^3_{5,2}]+\Tr[ 3M^1_{5,3}M^1_{5,4}+M^2_{5,3}M^2_{5,4}+M^3_{5,3}M^3_{5,4}]+Tr[ -M^1_{5,3}M^1_{5,6}\nonumber\\
&&+M^2_{5,3}M^2_{5,6}+M^3_{5,3}M^3_{5,6}]+\Tr[M^1_{5,3}M^1_{5,8}+M^2_{5,3}M^2_{5,8}+M^3_{5,3}M^3_{5,8}]+Tr[M^1_{5,5}M^1_{5,7}+M^2_{5,5}M^2_{5,7}\nonumber\\
&&+M^3_{5,5}M^3_{5,7}]+\Tr[M^1_{5,5}M^1_{5,2}+M^2_{5,5}M^2_{5,2}+M^3_{5,5}M^3_{5,2}]+\Tr[-M^1_{5,5}M^1_{5,4}+M^2_{5,5}M^2_{5,4}+M^3_{5,5}M^3_{5,4}]\nonumber\\
&&+\Tr[3M^1_{5,5}M^1_{5,6}+M^2_{5,5}M^2_{5,6}+M^3_{5,5}M^3_{5,6}]+\Tr[M^1_{5,5}M^1_{5,8}+M^2_{5,5}M^2_{5,8}+M^3_{5,5}M^3_{5,8}]+Tr[-M^1_{5,7}M^1_{5,2}\nonumber\\
&&+M^2_{5,7}M^2_{5,2}+M^3_{5,7}M^3_{5,2}]+\Tr[M^1_{5,7}M^1_{5,4}+M^2_{5,7}M^2_{5,4}+M^3_{5,7}M^3_{5,4}]+Tr[M^1_{5,7}M^1_{5,6}+M^2_{5,7}M^2_{5,6}\nonumber\\
&&+M^3_{5,7}M^3_{5,6}]+\Tr[3M^1_{5,7}M^1_{5,8}+M^2_{5,7}M^2_{5,8}+M^3_{5,7}M^3_{5,8}]+\Tr[M^1_{5,2}M^1_{5,4}+M^2_{5,2}M^2_{5,4}+M^3_{5,2}M^3_{5,4}]\nonumber\\
&&+\Tr[M^1_{5,2}M^1_{5,6}+M^2_{5,2}M^2_{5,6}+M^3_{5,2}M^3_{5,6}]+\Tr[-M^1_{5,2}M^1_{5,8}+M^2_{5,2}M^2_{5,8}+M^3_{5,2}M^3_{5,8}]+Tr[-M^1_{5,4}M^1_{5,6}\nonumber\\
&&+M^2_{5,4}M^2_{5,6}+M^3_{5,4}M^3_{5,6}]+Tr[M^1_{5,4}M^1_{5,8}+M^2_{5,4}M^2_{5,8}+M^3_{5,4}M^3_{5,8}]+Tr[M^1_{5,6}M^1_{5,8}+M^2_{5,6}M^2_{5,8}+M^3_{5,6}M^3_{5,8}]\Big)\Big]\Big]^\frac{1}{2}
\end{eqnarray}
Here in Eq.~(\ref{A151aap}), each of the $\Tr[..]$ are unique functions, hence we can write the following things.
\begin{eqnarray}\label{5to1eq1}
    \Tr[ M^1_{5,1}M^1_{5,3}+M^2_{5,1}M^2_{5,3}+M^3_{5,1}M^3_{5,3}]&=& 4-4\Tr[\rho_0\rho_{27}+ \rho_3\rho_{24}+\rho_{29}\rho_{6}+ \rho_{30}\rho_{5}]
\end{eqnarray}
Eq.~(\ref{5to1eq1}) will be maximum when 
$\Tr[\rho_0\rho_{27}]=\Tr[\rho_3\rho_{24}]=\Tr[\rho_{29}\rho_{6}]= \Tr[\rho_{30}\rho_{5}]=0$, which gives
\begin{eqnarray}
    \Tr[ M^1_{5,1}M^1_{5,3}+M^2_{5,1}M^2_{5,3}+M^3_{5,1}M^3_{5,3}]=4
\end{eqnarray}
Similarly, from the second to the twenty-eight block, we get 
\begin{eqnarray}
    \label{ea51}
 &&   \Tr[ M^1_{5,1}M^1_{5,5}+M^2_{5,1}M^2_{5,5}+M^3_{5,1}M^3_{5,3}]=
    \Tr[ -M^1_{5,1}M^1_{5,7}+\{M^2_{5,1},M^2_{5,7}\}+M^3_{5,1}M^3_{5,7}]=
    \Tr[ -M^1_{5,3}M^1_{5,5}+M^2_{5,3}M^2_{5,5}+M^3_{5,3}M^3_{5,5}]\nonumber \\ 
    &&=\Tr[ M^1_{5,3}M^1_{5,7}+M^2_{5,3}M^2_{5,7}+M^3_{5,3}M^3_{5,7}]=
    \Tr[ M^1_{5,2}M^1_{5,4}+\{M^2_{5,2}, M^2_{5,4}\}+M^3_{5,2}M^3_{5,4}]=
     \Tr[ M^1_{5,2}M^1_{5,6}+\{M^2_{5,2}, M^2_{5,6}\}+M^3_{5,2}M^3_{5,6}]\nonumber \\ 
    &&=\Tr[ -M^1_{5,2}M^1_{5,8}+M^2_{5,2} M^2_{5,8}+M^3_{5,2}M^3_{5,8}]=
     \Tr[ -M^1_{5,4}M^1_{5,6}+M^2_{5,4} M^2_{5,6}+M^3_{5,4}M^3_{5,6}]=
    \Tr[ M^1_{5,4}M^1_{5,8}+M^2_{5,4}M^2_{5,8}+M^3_{5,4}M^3_{5,8}]\nonumber \\ 
    &&=\Tr[ M^1_{5,6}M^1_{5,8}+M^2_{5,6}M^2_{5,8}+M^3_{5,6}M^3_{5,8}]=4
\end{eqnarray}
and, 
    \begin{eqnarray}\label{ea52}
   && \Tr[ 3M^1_{5,1}M^1_{5,2}+M^2_{5,1}M^2_{5,2}+M^3_{5,1}M^3_{5,2}]=
    \Tr[ M^1_{5,1}M^1_{5,4}+M^2_{5,1}M^2_{5,4}+M^3_{5,1}M^3_{5,4}]=
     \Tr[ M^1_{5,1}M^1_{5,6}+M^2_{5,1}M^2_{5,6}+M^3_{5,1}M^3_{5,6}] \nonumber \\
   && =\Tr[-M^1_{5,1}M^1_{5,8}+M^2_{5,1}M^2_{5,8}+M^3_{5,1}M^3_{5,8}]=
    \Tr[ M^1_{5,3}M^1_{5,2}+M^2_{5,3}M^2_{5,2}+M^3_{5,3}M^3_{5,2}]=
    \Tr[ 3M^1_{5,3}M^1_{5,4}+M^2_{5,3}M^2_{5,4}+M^3_{5,3}M^3_{5,4}]\nonumber \\
   &&=\Tr[ -M^1_{5,3}M^1_{5,6}+M^2_{5,3}M^2_{5,6}+M^3_{5,3}M^3_{5,6}]=
    \Tr[ M^1_{5,3}M^1_{5,8}+M^2_{5,3}M^2_{5,8}+M^3_{5,3}M^3_{5,8}]=\Tr[ M^1_{5,5}M^1_{5,7}+M^2_{5,5}M^2_{5,7}+M^3_{5,5}M^3_{5,7}]\nonumber \\
   &&=\Tr[ M^1_{5,5}M^1_{5,2}+M^2_{5,5}M^2_{5,2}+M^3_{5,5}M^3_{5,2}]=
   \Tr[ -M^1_{5,5}M^1_{5,4}+M^2_{5,5}M^2_{5,4}+M^3_{5,5}M^3_{5,4}]=
    \Tr[ 3M^1_{5,5}M^1_{5,6}+M^2_{5,5}M^2_{5,6}+M^3_{5,5}M^3_{5,6}]\nonumber \\
   &&=\Tr[ M^1_{5,5}M^1_{5,8}+M^2_{5,5}M^2_{5,8}+M^3_{5,5}M^3_{5,8}]=
    \Tr[ -M^1_{5,7}M^1_{5,2}+M^2_{5,7}M^2_{5,2}+M^3_{5,7}M^3_{5,2}]=
   \Tr[ M^1_{5,7}M^1_{5,6}+M^2_{5,7}M^2_{5,6}+M^3_{5,7}M^3_{5,6}]\nonumber \\
   &&=\Tr[ M^1_{5,7}M^1_{5,8}+M^2_{5,7}M^2_{5,8}+M^3_{5,7}M^3_{5,8}]=12
    \end{eqnarray}

Eqs.~(\ref{ea51})–(\ref{ea52}) reach their maximum only when all the $\Tr[..]=0$. Substituting all these equations into Eq.~(\ref{A151aap}), we obtain the following.
\begin{eqnarray}
    \mathscr{I}_{opt}^{5\rightarrow 1}\leq 4\sqrt{5(154)}=4\sqrt{770}
\end{eqnarray}
Hence, the quantum success probability in Eq.~(\ref{QSP51aap}) becomes
\begin{eqnarray}
    \mathcal{P}^{5\rightarrow 1}_{Q}&\leq&\dfrac{1}{2}+\dfrac{4\sqrt{770}}{320}\approx 0.847
\end{eqnarray}
which is Eq. (\ref{p51max}) in the main text.
%%%%%%%%%%%%%%%%%%%%%%%%%%%%%%%%%%%%%%%%%%%%%%%%%%%%%%%%%%%%%%%%%%%%%%%%%%%%%%
\section{Detailed derivation of optimal quantum success probability for $5\rightarrow 2$  entanglement-assisted PMRAC }
\label{optm52}
In the $5\rightarrow 2$ entanglement-assisted PMRAC game, the quantum success probability from Eq. (\ref{VQSP}) in the main text can be explicitly written as
\begin{eqnarray}\label{suc52}
	\nonumber
	\mathcal{P}^{5\rightarrow 2}_{Q}=\dfrac{1}{2}&+&\dfrac{1}{320}\ Tr[(N^1_{5,1}+N^1_{5,3}+N^1_{5,2}+N^1_{5,4})B_{5,1}+(N^1_{5,1}-N^1_{5,3}+N^1_{5,2}-N^1_{5,4})B_{5,2}+(N^2_{5,1}+N^2_{5,3}+N^2_{5,2}+N^2_{5,4})B_{5,3}\nonumber\\
    &&\hspace{1.3cm}+(N^3_{5,1}+N^3_{5,3}+N^3_{5,2}+N^3_{5,4})B_{5,4}+(N^4_{5,1}+N^4_{5,3}+N^4_{5,2}+N^4_{5,4})B_{5,5}]
\end{eqnarray}
where we define two-qubit states of the form 
\begin{eqnarray}
	\label{obserg52}
	\nonumber
	N^1_{5,1}&=&\rho_0+\rho_3-\rho_{29}-\rho_{30}+\rho_{5}+\rho_{6}-\rho_{24}-\rho_{27}; \ \ \ 
	N^1_{5,3}=\rho_{9}+\rho_{10}-\rho_{20}-\rho_{23}+\rho_{12}+\rho_{15}-\rho_{17}-\rho_{18} \ \ \\\
    N^2_{5,1}&=&\rho_0+\rho_3-\rho_{29}-\rho_{30}-
	\rho_{5}-\rho_{6}+\rho_{24}+\rho_{27}; \ \ \ N^2_{5,3}=\rho_{9}+\rho_{10}-\rho_{20}-\rho_{23}-\rho_{12}-\rho_{15}+\rho_{17}+\rho_{18} \nonumber
	\\
    \nonumber
	N^3_{5,1}&=&
	\rho_{0}-\rho_{3}+\rho_{29}-\rho_{30}+
	\rho_{5}-\rho_{6}+\rho_{24}-\rho_{27};\ \ \
	N^3_{5,3}=
	\rho_{9}-\rho_{10}+\rho_{20}-\rho_{23}+\rho_{12}-\rho_{15}+\rho_{17}-\rho_{18}\\
    N^4_{5,1}&=&
	\rho_{0}-\rho_{3}-\rho_{29}+\rho_{30}-\rho_{5}+\rho_{6}+\rho_{24}-\rho_{27};\ \ \ 
    N^4_{5,3}=-\rho_{9}+\rho_{10}+\rho_{20}-\rho_{23}+\rho_{12}-\rho_{15}-\rho_{17}+\rho_{18}\\
    N^1_{5,2}&=&\rho_{1}+\rho_{2}-\rho_{28}-\rho_{31}+\rho_{4}+\rho_{7}-\rho_{25}-\rho_{26};\ \ \ 
    N^1_{5,4}=\rho_{8}+\rho_{11}-\rho_{21}-\rho_{22}+\rho_{13}+\rho_{14}-\rho_{16}-\rho_{19} \nonumber \\
    N^2_{5,2}&=&\rho_{1}+\rho_{2}-\rho_{28}-\rho_{31}-\rho_{4}-\rho_{7}+\rho_{25}+\rho_{26};\ \ \ N^2_{5,4}=\rho_{8}+\rho_{11}-\rho_{21}-\rho_{22}-\rho_{13}-\rho_{14}+\rho_{16}+\rho_{19}\nonumber\\
    N^3_{5,2}&=&\rho_{1}-\rho_{2}+\rho_{28}-\rho_{31}+\rho_{4}-\rho_{7}+\rho_{25}-\rho_{26};\ \ \
    N^3_{5,4}=\rho_{8}-\rho_{11}+\rho_{21}-\rho_{22}+\rho_{13}-\rho_{14}+\rho_{16}-\rho_{19}\nonumber\\
    N^4_{5,2}&=&-\rho_{1}+\rho_{2}+\rho_{28}-\rho_{31}+\rho_{4}-\rho_{7}-\rho_{25}+\rho_{26};\ \ \
    N^4_{5,4}=\rho_{8}-\rho_{11}-\rho_{21}+\rho_{22}-\rho_{13}+\rho_{14}+\rho_{16}-\rho_{19}\nonumber
\end{eqnarray}

and the states satisfy,
\begin{eqnarray}
	&&\rho_0+\rho_3+\rho_{29}+\rho_{30}+\rho_{5}+\rho_{6}+\rho_{24}+\rho_{27}=\openone_8;\ \ \ \rho_{9}+\rho_{10}+\rho_{20}+\rho_{23}+\rho_{12}+\rho_{15}+\rho_{17}+\rho_{18}=\openone_8\nonumber\\
     &&\rho_{1}+\rho_{2}+\rho_{28}+\rho_{31}+\rho_{4}+\rho_{7}+\rho_{25}+\rho_{26}=\openone_8; \ \ \  \rho_{8}+\rho_{11}+\rho_{21}+\rho_{22}+\rho_{13}+\rho_{14}+\rho_{16}+\rho_{19}=\openone_8
\end{eqnarray}
Then the quantum success probability in Eq. (\ref{suc52}) can be casted as
\begin{eqnarray}\label{QSP52aap}
	\mathcal{P}^{5\rightarrow 2}_{Q}&=&\dfrac{1}{2}+\dfrac{\mathscr{I}_{5\rightarrow 2}}{320}
\end{eqnarray}
where the function $\mathscr{I}^{5\rightarrow2}$ is given by 
\begin{eqnarray}\label{I52a}
\mathscr{I}^{5\rightarrow 2}&=&Tr[(N^1_{5,1}+N^1_{5,3}+N^1_{5,2}+N^1_{5,4})B_{5,1}+(N^1_{5,1}-N^1_{5,3}+N^1_{5,2}-N^1_{5,4})B_{5,2}+(N^2_{5,1}+N^2_{5,3}+N^2_{5,2}+N^2_{5,4})B_{5,3}\nonumber\\
    &&+(N^3_{5,1}+N^3_{5,3}+N^3_{5,2}+N^3_{5,4})B_{5,4}+(N^4_{5,1}+N^4_{5,3}+N^4_{5,2}+N^4_{5,4})B_{5,5}]
\end{eqnarray}

Now, we derive the QM maximum value of $\mathscr{I}_{5\rightarrow 2}$. The correlation function in Eq.~(\ref{I52a})  can be written as
\begin{eqnarray}
\mathscr{I}^{5\rightarrow 2}&=&\Tr[\sum\limits_{y=1}^{5}\mu_y\mathscr{N}_{5,y} B_{5,y} ]\label{wab52aap}
\end{eqnarray}
where $\mathscr{N}_{5,y}$ and $\mu_y$'s are defined as following
\begin{eqnarray}
     &&\mathscr{N}_{5,1} = \frac{N^1_{5,1}+N^1_{5,3}+N^1_{5,2}+N^1_{5,4}}{\mu_1}; \ \ \mathscr{N}_{5,2} = \frac{N^1_{5,1}-N^1_{5,3}+N^1_{5,2}-N^1_{5,4}}{\mu_2}; \ \ \mathscr{N}_{5,3} = \frac{N^2_{5,1}+N^2_{5,3}+N^2_{5,2}+N^2_{5,4}}{\mu_3}\nonumber\\
     &&\mathscr{N}_{5,4} = \frac{N^3_{5,1}+N^3_{5,3}+N^3_{5,2}+N^3_{5,4}}{\mu_4}; \ \ \mathscr{N}_{5,5} = \frac{N^4_{5,1}+N^4_{5,3}+N^4_{5,2}+N^4_{5,4}}{\mu_5};\  \ \mu_1=||N^1_{5,1}+N^1_{5,3}+N^1_{5,2}+N^1_{5,4}||\nonumber\\
     &&\mu_2=||N^1_{5,1}-N^1_{5,3}+N^1_{5,2}-N^1_{5,4}||; \ \ \mu_3=||N^2_{5,1}+N^2_{5,3}+N^2_{5,2}+N^2_{5,4}||; \ \ \mu_4=||N^3_{5,1}+N^3_{5,3}+N^3_{5,2}+N^3_{5,4}||\nonumber\\
     &&\mu_5=||N^4_{5,1}+N^4_{5,3}+N^4_{5,2}+N^4_{5,4}|| \label{omegai52aap}
\end{eqnarray}
  such that $||\mathscr{N}_{5,y}||=1$.
  
  The maximization condition $\mathscr{N}_{5,y}=B_{5,y}$ gives
  \begin{equation}
\mathscr{I}^{5\rightarrow 2}_{opt}= \Tr[\max\left(\sum\limits_{y=1}^{5}\mu_{y} \right)\openone_8]=8\max\left(\sum\limits_{y=1}^{5}\mu_{y} \right)
\end{equation}
Now, by evaluating $\mu_y,\forall y\in[5]$ from Eq.~(\ref{omegai52aap}), we get
\begin{eqnarray}\label{omegai152aap}
\mu_1
&=&\Bigg[4+\frac{1}{4}\Tr[N^1_{5,1}N^1_{5,3}+ N^1_{5,1}N^1_{5,2} + N^1_{5,1}N^1_{5,4} +  N^1_{5,3}N^1_{5,2}+N^1_{5,3}N^1_{5,4}+  N^1_{5,2} N^1_{5,4}] \Bigg]^\frac{1}{2}\\
\mu_2&=&
\Bigg[4+\frac{1}{4}\Tr[-N^1_{5,1}N^1_{5,3}+ N^1_{5,1}N^1_{5,2} - N^1_{5,1}N^1_{5,4} -  N^1_{5,3}N^1_{5,2}+N^1_{5,3}N^1_{5,4}-N^1_{5,2} N^1_{5,4}] \Bigg]^\frac{1}{2}\\
\mu_3&=&
\Bigg[4+\frac{1}{4}\Tr[N^2_{5,1}N^2_{5,3}+ N^2_{5,1}N^2_{5,2} + N^2_{5,1}N^2_{5,4} +  N^2_{5,3}N^2_{5,2}+N^2_{5,3}N^2_{5,4}+  N^2_{5,2}N^2_{5,4}] \Bigg]^\frac{1}{2}\\
\mu_4&=&\Bigg[4+\frac{1}{4}\Tr[N^3_{5,1}N^3_{5,3}+ N^3_{5,1}N^3_{5,2} + N^3_{5,1}N^3_{5,4} +N^3_{5,3}N^3_{5,2}+N^3_{5,3}N^3_{5,4}+  N^3_{5,2},N^3_{5,4}] \Bigg]^\frac{1}{2}\\
\mu_5&=&\Bigg[4+\frac{1}{4}\Tr[N^4_{5,1}N^4_{5,3}+ N^4_{5,1}N^4_{5,2} + N^4_{5,1}N^4_{5,4} +N^4_{5,3}N^4_{5,2}+N^4_{5,3}N^4_{5,4}+  N^4_{5,2} N^4_{5,4}] \Bigg]^\frac{1}{2}
\end{eqnarray}
Next, using the convex inequality $\sum\limits_{y=1}^{5}\mu_{y}\leq \sqrt{5 \sum\limits_{y=1}^{5} (\mu_{y})^{2}}$, we get
\begin{eqnarray}\label{A152aap}
\mathscr{I}^{5\rightarrow 2}&\leq&8\qty( \max \sqrt{5\ \qty(\mu_1^2+\mu_2^2+\mu_3^2+\mu_4^2+\mu_5^2)} )\nonumber \\
&\leq&8\max\Bigg[5 \ \Big[20+\frac{1}{4}\Tr[ N^2_{5,1}N^2_{5,3}+N^3_{5,1}N^3_{5,3}+N^4_{5,1}N^4_{5,3}]+\frac{1}{4}\Tr[ 2N^1_{5,1}N^1_{5,2}+N^2_{5,1}N^2_{5,2}+N^3_{5,1}N^3_{5,2}+N^4_{5,1}N^4_{5,2}]\nonumber\\
&&+\frac{1}{4}\Tr[ N^2_{5,1}N^2_{5,4}+N^3_{5,1}N^3_{5,4}+N^4_{5,1}N^4_{5,4}]+\frac{1}{4}\Tr[N^2_{5,2}N^2_{5,4}+N^3_{5,2}N^3_{5,4}+N^4_{5,2}N^4_{5,4}]+\frac{1}{4}\Tr[N^2_{5,3}N^2_{5,2}+N^3_{5,3}N^3_{5,2}+N^4_{5,3}N^4_{5,2}]\nonumber\\
&&+\frac{1}{4}\Tr[ 2N^1_{5,3}N^1_{5,4}+N^2_{5,3}N^2_{5,4}+N^3_{5,3}N^3_{5,4}+N^4_{5,3}N^4_{5,4}]\Big]\Bigg]^\frac{1}{2}
\end{eqnarray}
Here in Eq.~(\ref{A152aap}), each of the $\Tr[..]$ has unique functions, and hence we can write the following 
\begin{eqnarray}
    &&\hspace{-0.25cm}\Tr[ N^2_{5,1}N^2_{5,3}+N^3_{5,1}N^3_{5,3}+N^4_{5,1}N^4_{5,3}]\nonumber\\
    &&=8-4 \ Tr\big[\rho_0\rho_{15}+\rho_0\rho_{23}+\rho_3\rho_{12}+\rho_3\rho_{20}+\rho_5\rho_{10}+\rho_5\rho_{18}+\rho_6\rho_{17}+\rho_6\rho_{9}+\rho_{24}\rho_{15}+\rho_{24}\rho_{23}+\rho_{27}\rho_{12}+\rho_{27}\rho_{20}+\rho_{29}\rho_{10}\nonumber\\
    && +\rho_{29}\rho_{18}+\rho_{30}\rho_{17}+\rho_{30}\rho_{9}\big]\nonumber\\
    &&=8 \quad (\text{At max all the} \Tr[\rho_x \rho_{x'}]_{x\neq x'}=0)
\end{eqnarray}
Similarly, we derive
\begin{eqnarray}
&&\Tr[N^2_{5,3}N^2_{5,2}+N^3_{5,3}N^3_{5,2}+N^4_{5,3}N^4_{5,2}]=\Tr[ 2N^1_{5,3}N^1_{5,4}+N^2_{5,3}N^2_{5,4}+N^4_{5,3}N^3_{5,4}+N^4_{5,3}N^4_{5,4}]\nonumber \\
&&=\Tr[ 2N^1_{5,1}N^1_{5,2}+N^2_{5,1}N^2_{5,2}+N^3_{5,1}N^3_{5,2}+N^4_{5,1}N^4_{5,2}]=\Tr[ N^2_{5,1}N^2_{5,4}+N^3_{5,1}N^3_{5,4}+N^4_{5,1}N^4_{5,4}]=24 
\end{eqnarray}
and,
\begin{eqnarray}
&&\hspace{-0.25cm}\Tr[N^2_{5,2}N^2_{5,4}+N^3_{5,2}N^3_{5,4}+N^4_{5,2}N^4_{5,4}]=8\\ 
\end{eqnarray}

Now, putting all these equations into Eq.~(\ref{A152aap}), we get
\begin{eqnarray}
    \mathscr{I}_{opt}^{5\rightarrow 2}\leq8\sqrt{5(20+28)}=32\sqrt{15}
\end{eqnarray}
Hence the quantum success probability in Eq.~(\ref{QSP52aap}) becomes
\begin{eqnarray}
    \mathcal{P}^{5\rightarrow 2}_{Q}&\leq&\dfrac{1}{2}+\dfrac{32\sqrt{15}}{320}\approx 0.887
\end{eqnarray}  
which is Eq. (\ref{p52max}) in the main text.
%%%%%%%%%%%%%%%%%%%%%%%%%%%%%%%%%%%%%%%%%%%%%%%%%%%%%%%%%%%%%%%%%%%%%%%%%%%%%%
\section{Detailed derivation of optimal quantum success probability for $5\rightarrow 3$  entanglement-assisted PMRAC }
\label{optm53}
In the $5\rightarrow 3$ entanglement-assisted PMRAC game, the quantum success probability from Eq. (\ref{VQSP}) in the main text can be explicitly written as
\begin{eqnarray}\label{suc53}
	\mathcal{P}^{5\rightarrow 2}_{Q}=\dfrac{1}{2}&+&\dfrac{1}{320}\Tr[(L^1_{5,1}+L^1_{5,2})B_{5,1}+(L^2_{5,1}+L^2_{5,2})B_{5,2}+(L^3_{5,1}+L^3_{5,2})B_{5,3}+(L^4_{5,1}+L^4_{5,2})B_{5,4}+(L^5_{5,1}+L^5_{5,2})B_{5,5}]
\end{eqnarray}
where the two-qubit states are of the form
\begin{eqnarray}
    &&L^1_{5,1}=\rho_0+\rho_3-\rho_{29}-\rho_{30}+\rho_{5}+\rho_{6}-\rho_{24}-\rho_{27}+\rho_{9}+\rho_{10}-\rho_{20}-\rho_{23}+\rho_{12}+\rho_{15}-\rho_{17}-\rho_{18}\label{531} \\
    &&L^2_{5,1}=\rho_0+\rho_3-\rho_{29}-\rho_{30}+\rho_{5}+\rho_{6}-\rho_{24}-\rho_{27}-\rho_{9}-\rho_{10}+\rho_{20}+\rho_{23}-\rho_{12}-\rho_{15}+\rho_{17}+\rho_{18}\\
    &&L^3_{5,1}=\rho_0+\rho_3-\rho_{29}-\rho_{30}-
	\rho_{5}-\rho_{6}+\rho_{24}+\rho_{27}+\rho_{9}+\rho_{10}-\rho_{20}-\rho_{23}-\rho_{12}-\rho_{15}+\rho_{17}+\rho_{18}\\
    &&L^4_{5,1}=\rho_{0}-\rho_{3}+\rho_{29}-\rho_{30}+
	\rho_{5}-\rho_{6}+\rho_{24}-\rho_{27}+
	\rho_{9}-\rho_{10}+\rho_{20}-\rho_{23}+\rho_{12}-\rho_{15}+\rho_{17}-\rho_{18}\\
    &&L^5_{5,1}=\rho_{0}-\rho_{3}-\rho_{29}+\rho_{30}-\rho_{5}+\rho_{6}+\rho_{24}-\rho_{27}-\rho_{9}+\rho_{10}+\rho_{20}-\rho_{23}+\rho_{12}-\rho_{15}-\rho_{17}+\rho_{18}\\
    &&L^1_{5,2}=\rho_{1}+\rho_{2}-\rho_{28}-\rho_{31}-\rho_{4}-\rho_{7}+\rho_{25}+\rho_{26}+\rho_{8}+\rho_{11}-\rho_{21}-\rho_{22}-\rho_{13}-\rho_{14}+\rho_{16}+\rho_{19}r\\
    &&L^2_{5,2}=\rho_{1}+\rho_{2}-\rho_{28}-\rho_{31}-\rho_{4}-\rho_{7}+\rho_{25}+\rho_{26}-\rho_{8}-\rho_{11}+\rho_{21}+\rho_{22}+\rho_{13}+\rho_{14}-\rho_{16}-\rho_{19}\\
    &&L^3_{5,3}=\rho_{1}+\rho_{2}-\rho_{28}-\rho_{31}-\rho_{4}-\rho_{7}+\rho_{25}+\rho_{26}+\rho_{8}+\rho_{11}-\rho_{21}-\rho_{22}-\rho_{13}-\rho_{14}+\rho_{16}+\rho_{19}\\
    &&L^4_{5,2}=\rho_{1}-\rho_{2}+\rho_{28}-\rho_{31}+\rho_{4}-\rho_{7}+\rho_{25}-\rho_{26}+\rho_{8}-\rho_{11}+\rho_{21}-\rho_{22}+\rho_{13}-\rho_{14}+\rho_{16}-\rho_{19}\\
    &&L^5_{5,2}=-\rho_{1}+\rho_{2}+\rho_{28}-\rho_{31}+\rho_{4}-\rho_{7}-\rho_{25}+\rho_{26}+\rho_{8}-\rho_{11}-\rho_{21}+\rho_{22}-\rho_{13}+\rho_{14}+\rho_{16}-\rho_{19}\label{532}
\end{eqnarray}
and
\begin{eqnarray}
	&&\rho_0+\rho_3+\rho_{29}+\rho_{30}+\rho_{5}+\rho_{6}+\rho_{24}+\rho_{27}+\rho_{9}+\rho_{10}+\rho_{20}+\rho_{23}+\rho_{12}+\rho_{15}+\rho_{17}+\rho_{18}=\openone_{16}\nonumber\\
     &&\rho_{1}+\rho_{2}+\rho_{28}+\rho_{31}+\rho_{4}+\rho_{7}+\rho_{25}+\rho_{26}+\rho_{8}+\rho_{11}+\rho_{21}+\rho_{22}+\rho_{13}+\rho_{14}+\rho_{16}+\rho_{19}=\openone_{16}
\end{eqnarray}
Then the quantum success probability in Eq. (\ref{suc53}) can be casted as
\begin{eqnarray}\label{QSP53aap}
	\mathcal{P}^{5\rightarrow 3}_{Q}&=&\dfrac{1}{2}+\dfrac{\mathscr{I}_{5\rightarrow 3}}{320}
\end{eqnarray}
where the function $\mathscr{I}^{5\rightarrow3}$ is given by 
\begin{eqnarray}\label{I53}
\mathscr{I}^{5\rightarrow 3}&=&\Tr[(L^1_{5,1}+L^1_{5,2})B_{5,1}+(L^2_{5,1}+L^2_{5,2})B_{5,2}+(L^3_{5,1}+L^3_{5,2})B_{5,3}+(L^4_{5,1}+L^4_{5,2})B_{5,4}+(L^5_{5,1}+L^5_{5,2})B_{5,5}]
\end{eqnarray}

Now, we derive the QM maximum value of $\mathscr{I}_{5\rightarrow 3}$. The correlation function in Eq.~(\ref{I53})  can be written as
\begin{eqnarray}
\mathscr{I}^{5\rightarrow 3}&=&\Tr[\sum\limits_{y=1}^{5}\eta_y\mathscr{L}_{5,y} B_{5,y} ]\label{wab53aap}
\end{eqnarray}
where $\mathscr{L}_{5,y}$ and $\eta_y$'s are defined as following
\begin{eqnarray}
     &&\mathscr{L}_{5,y} = \frac{L^y_{5,1}+L^y_{5,2}}{\eta_y};\quad \eta_y=||L^y_{5,1}+L^y_{5,2}||,\quad \forall y\in [5] \label{omegai53aap}
\end{eqnarray}
  such that $||\mathscr{L}_{5,y}||=1$ and $||\cdot||$ is the scaled Frobenious norm, given by $||\mathcal{O}||=\frac{1}{\sqrt{16}}\sqrt{Tr[\mathcal{O}^{\dagger}\mathcal{O}]}$.
  
  Interestingly, since $\mathscr{L}_{5,y}$ and $ B_{5,y}$ are dichotomic normalized observables, the maximum value of $\mathscr{L}_{5,y} B_{5,y}$ occurs only when $\mathscr{L}_{5,y}=B_{5,y}$, which further implies
\begin{equation}\label{optbn53aap}
\mathscr{I}^{5\rightarrow 2}_{opt}= \Tr[\max\left(\sum\limits_{y=1}^{5}\eta_y \right)\openone_{16}]=16\max\left(\sum\limits_{y=1}^{5}\eta_y \right)
\end{equation}
Now, by evaluating $\eta_y$ from Eq.~(\ref{omegai53aap}), we get
\begin{eqnarray}
\eta_y
&=&\Bigg[2+\frac{1}{8}\Tr[L^y_{5,1}L^y_{5,2}] \Bigg]^\frac{1}{2},\quad \forall y\in[5]
\end{eqnarray}
Next, using the convex inequality $\sum\limits_{y=1}^{5}\eta_{y}\leq \sqrt{5 \sum\limits_{y=1}^{5} (\eta_{y})^{2}}$ where equality holds only when $\eta_y=\eta_{y'} \forall y\neq y'\in [5]$, from Eqs.~(\ref{optbn53aap}), we get

\begin{eqnarray}\label{A153aap}
\mathscr{I}^{5\rightarrow 3}&\leq&16\qty( \max \sqrt{5\ \qty(\eta_1^2+\eta_2^2+\eta_3^2+\eta_4^2+\eta_5^2)} )\nonumber \\
&\leq&16\max\Bigg[5 \ \Big[10+\frac{1}{8}\Tr[ L^1_{5,1}L^1_{5,2}+L^2_{5,1}L^2_{5,2}+L^3_{5,1}L^3_{5,2}+L^4_{5,1}L^4_{5,2}+L^5_{5,1}L^5_{5,2}]\Big]\Bigg]^\frac{1}{2}
\end{eqnarray}
Now, each $\Tr[L^y_{5,1}L^y_{5,2}] \ \forall y\in [5]$ can be further rewritten in terms of $\Tr[\rho_x\rho_{x'}]$ using Eq.~(\ref{531})-(\ref{532}). For maximizing Eq.~(\ref{A153aap}), each $\Tr[\rho_x\rho_{x'}]$ with $x,x \in \{0,1,2,...,31\}$ should be zero, which gives
\begin{eqnarray}
    \mathscr{I}_{opt}^{5\rightarrow 3}\leq 64\sqrt{5}
\end{eqnarray}
Hence the quantum success probability in Eq.~(\ref{QSP53aap}) becomes
\begin{eqnarray}
    \mathcal{P}^{5\rightarrow 3}_{Q}\leq\dfrac{1}{2}+\dfrac{64\sqrt{5}}{320}\approx 0.947
\end{eqnarray}
which is Eq. (\ref{p53max}) in the main text.
%%%%%%%%%%%%%%%%%%%%%%%%%%%%%%%%%%%%%%%%%%%%%%%%%%%%%%%%%%%%%%%%%%%%%%%%%%%%%%%%%%%%%%%%%%%%%%%%%%%%%%%%%%%%%%%%%%%%%%%%%%%%%%%%%%%%%%%%%%

\end{widetext}
%\section*{References}
%\bibliographystyle{unsrt}
\bibliographystyle{apsrev4-2}
\bibliography{references}

\end{document}